\documentclass[aip,amsmath,amssymb,preprint,]{revtex4-1}
\usepackage[
            bookmarksnumbered=true,
            bookmarksopen=true,
            colorlinks,
            linkcolor=blue,
            anchorcolor=blue,
            citecolor=blue,
            urlcolor=blue
            ]{hyperref}
\usepackage{graphicx}
\usepackage{epstopdf}
\usepackage{mathrsfs}
\usepackage{dcolumn}
\usepackage{bm}
\usepackage[mathlines]{lineno}
\usepackage{subfigure}
\usepackage[utf8]{inputenc}
\usepackage[T1]{fontenc}
\usepackage{mathptmx}
\usepackage{float}
\usepackage{amsmath}
\usepackage{amsthm}
\newtheorem{theorem}{Theorem}

\newtheorem{remark}{Remark}
\begin{document}

\preprint{\textit{Chaos}}

\title[Pattern formation of parasite-host model induced by fear effect]{Pattern formation of parasite-host model induced by fear effect}
\author{Yong Ye}
\author{Yi Zhao}
 \email{zhao.yi@hit.edu.cn}
\author{Jiaying Zhou}
\affiliation{School of Science, Harbin Institute of Technology (Shenzhen), Shenzhen 518055, China}
\date{\today}

\begin{abstract}
In this paper, based on the epidemiological microparasite model, a parasite-host model is established by considering the fear effect of susceptible individuals on infectors. We explored the pattern formation with the help of numerical simulation, and analyzed the effects of fear effect, infected host mortality, population diffusion rate and  reducing reproduction ability of infected hosts on population activities in different degrees. Theoretically, we give the general conditions for the stability of the model under non-diffusion and considering the Turing instability caused by diffusion. Our results indicate how fear affects the distribution of the uninfected and infected hosts in the habitat and quantify the influence of the fear factor on the spatiotemporal pattern of the population. In addition, we analyze the influence of natural death rate, reproduction ability of infected hosts, and diffusion level of uninfected (infected) hosts on the spatiotemporal pattern, respectively. The results present that the growth of pattern induced by intensified fear effect follows the certain rule: cold spots $\rightarrow$ cold spots-stripes $\rightarrow$ cold stripes $\rightarrow$ hot stripes $\rightarrow$ hot spots-stripes $\rightarrow$ hot spots. Interestingly, the natural mortality and fear effect take the opposite effect on the growth order of the pattern. From the perspective of biological significance, we find that the degree of fear effect can reshape the distribution of population to meet the previous rule.
\end{abstract}

\maketitle

\begin{quotation}
With the development of reaction-diffusion equation, the research on the pattern formation of population model has been widely concerned. Among them, the research on spatiotemporal dynamics of predator-prey model is particularly rich. Recently, Wang et al.~\cite{ref2} first proposed the mathematical expression of fear effect and considered the fear effect in the traditional predator-prey model. It was found that the addition of fear effect brought complex dynamic phenomena. Since then, many predator-prey models with fear effect have been studied. Considering that the fear effect also exists between uninfected host and infected host. Therefore, this paper attempts to introduce the fear effect into the  epidemiological microparasite model and construct a parasite-host model with fear effect. With the help of computer, we simulate the distribution of population under different degrees of fear, and combined with other ecological factors to explore how the distribution of population will change under the interference of various factors. Our results show that the growth of pattern satisfies some laws under different ecological factors. Hopefully, this work will provide us further understanding of population dynamics in a real environment stimulated by the fear effect.
\end{quotation}

\section{Introduction}\label{section1}
Ecologists recognize that diseases and parasites play an important role in population dynamics~\cite{ref3,ref9,ref13,ref14}. The spatial components of ecological interactions have been identified as an important factor in how populations operate and form. However, understanding the role of space is challenging both theoretically and empirically~\cite{ref15,ref16,ref17,ref18}. In recent years, many studies have shown that pattern formation in parasite-host model is an appropriate tool to understand the basic mechanism of parasite spatiotemporal dynamics. In 2003, Hwang and Kuang established a parasite-host ordinary differential equation (ODE) model \cite{ref3} based on the work of Ebert et al.~\cite{ref9}. On this basis, the dynamics and pattern formation of the reaction-diffusion parasite-host model were studied in~\cite{ref1,ref8}, the model is as follows:

\begin{equation}\label{1.1}
\left\{\begin{array}{lll}
\frac{\partial S}{\partial t}-d_{1} \Delta S=r(S+\rho I)(1-a(S+I))-\frac{\beta S I}{S+I}, & x \in \Omega, & t>0, \\
\frac{\partial I}{\partial t}-d_{2} \Delta I=\frac{\beta S I}{S+I}-\mu I, & x \in \Omega, & t>0, \\
\frac{\partial S}{\partial \mathbf{n}}=\frac{\partial I}{\partial \mathbf{n}}=0, & x \in \partial \Omega, & t>0, \\
S(x, 0)=S_{0}(x), I(x, 0)=I_{0}(x), & x \in \Omega, &
\end{array}\right.
\end{equation}
where $\frac{\beta S I}{S+I}$ denotes the frequency-dependent transmission~\cite{ref1,ref3}, the variable $S$ represents density of uninfected (susceptible) hosts, and $I$ represents density of infected (infective) hosts. The habitat $\Omega \subset \mathbb{R}^{n}$ is a positive bounded region with smooth boundary $\partial \Omega$, $x$ represents location in the habitat, $\mathbf{n}$ stands for the outward unit normal vector on $\partial \Omega$, $d_{1}$ and $d_{2}$ respectively represent self diffusion coefficients of uninfected host and infected host, and $\Delta $ denotes Laplacian operator. In this paper we assume that the habitat is closed, which means that the infected (uninfected) host population inside the habitat cannot go out, while the infected (uninfected) host population outside the habitat cannot enter, and there is no host population on the boundary. That is to say, the boundary condition we consider is zero-flux (i.e., Neumann boundary). The biological significance of other parameters are described in Table~\ref{tab2}. It is worth noting that all parameters are positive and $0 \leq \rho \leq 1$.

In recent years, the relevant experiments on risk perception (i.e., indirect effect) of biological population have been proposed. Zanette et al.'s experiment on the perception of predation risk by songbirds shows that the number of offspring produced each year is reduced by $40 \%$ just by the perception of predation risk. So the perception of predation risk takes obviously significant effect on the dynamics of biological population~\cite{ref11}. Abbey-Lee et al. adopts playback technology to conduct experiments on the perceived predation risk in nest-box populations of wild great tits (Parus major), to investigate the effects of nonconsumptive on the predation behavior, the morphology of bird predators, and the individual responses to predation~\cite{ref12}. Their results show that the sensitivity of individuals to predation risk leads to the different adaptability. As a typical and representative indirect effect, the fear factor is used to describe the physiological changes caused by the stress behavior of the prey population since the prey needs to be alert to predators coming at any time~\cite{ref10}. In order to characterize the influence of anti-predator behavior on predator-prey system, Wang et al. firstly propose a predator-prey model with consideration of the fear factor in the growth of the prey~\cite{ref2}. After that, the researchers consider that the outbreak and spread of infectious diseases would bring people fear. In literature~\cite{ref36}, they think that the susceptible population has a fear effect, which shows that the fear effect can reduce the growth rate of the susceptible population. Following their ideas, we assume that the uninfected population has anti-infected behavior, that is, the fear effect on infected population. We then consider the growth function of the uninfected population with fear effect as follows: $ \frac{\mathrm{d} S}{\mathrm{~d} t}=[F(k, I) r] S $. $F(k, I)$ accounts for the cost of anti-infected due to uninfected population, and the parameter $k$ reflects the level of fear which drives anti-infected behavior of the uninfected host. Which is similar to \cite{ref2,ref4}, in Table~\ref{tab1}, we show the conditions that the fear factor $F(k, I)$ satisfies.

\begin{table*}[h]
    \centering
        \caption{Conditions that fear factor $F(k, I)$ satisfies.}
    \begin{tabular}{cclll}
    \hline\hline
  Conditions & Statements \\
  \hline
 $F(0, I)=1$ & When there is no fear, the maximum birth rate of uninfected population did not decrease \\
 $F(k, 0)=1$ & When there is no infected host, the maximum birth rate of uninfected population did not decrease \\
 $\lim _{k \rightarrow \infty} F(k, I)=0 $ & When anti-infected behavior is large enough, the uninfected production reduces to $0$ \\
 $\lim _{I \rightarrow \infty} F(k, I)=0$ & When infected population is large enough, the uninfected production reduces to $0$ \\
 $\frac{\partial F(k, I)}{\partial k}<0$ & When anti-infected behavior increases, the uninfected production decreases \\
 $\frac{\partial F(k, I)}{\partial I}<0$ & When infected population increases, the uninfected production decreases \\
 \hline\hline
    \end{tabular}
    \label{tab1}
\end{table*}

Population dynamics and epidemic dynamics models considering fear effect have been widely studied \cite{ref4,ref5,ref6,ref7,ref10,ref26,ref28,ref29,ref30,ref31,ref32,ref33,ref35,ref36,ref37}. This paper focus on the influence of fear factor on the pattern formation of host parasite model and the complex dynamic changes. Following \cite{ref2}, we introduce the fear effect $F(k, I):=\frac{1}{1+k I} $ as the fear effect in model~(\ref{1.1}). which satisfies the conditions in Table~\ref{tab1}. Then we obtained the following model:
\begin{equation}\label{1.2}
\left\{\begin{array}{lll}
\frac{\partial S}{\partial t}-d_{1} \Delta S=\frac{rS}{1+kI}+r\rho I-ra(S+\rho I)(S+I)-\frac{\beta S I}{S+I}, & x \in \Omega, & t>0, \\
\frac{\partial I}{\partial t}-d_{2} \Delta I=\frac{\beta S I}{S+I}-\mu I, & x \in \Omega, & t>0, \\
\frac{\partial S}{\partial \mathbf{n}}=\frac{\partial I}{\partial \mathbf{n}}=0, & x \in \partial \Omega, & t>0, \\
S(x, 0)=S_{0}(x), I(x, 0)=I_{0}(x), & x \in \Omega. &
\end{array}\right.
\end{equation}

The main structure of this paper is as follows. In Section~\ref{section2}, we discuss the existence condition of equilibrium in the non-diffusion model~(\ref{1.2}), and obtain the stability condition through the general linear stability analysis. In Section~\ref{section3}, we analyze the diffusion model~(\ref{1.2}) and find out the Turing space where Turing instability occurs. Then, the hexagonal and stripe pattern of model~(\ref{1.2}) are studied by using amplitude equation near the critical value of control parameters. In Section~\ref{section4}, we use numerical simulation to illustrate the different patterns we found. Finally, the results and future work are discussed.

\section{Model without Diffusion}\label{section2}
Since this paper mainly discusses Turing instability caused by diffusion, and the premise of Turing instability is to ensure that the positive equilibrium is stable without diffusion. We first consider the case of ODE model without diffusion, i.e., $d_{1}=d_{2}=0$. Then the existence and stability conditions of nontrivial (positive) equilibria will be given.
\subsection{Existence of Equilibria}
If
\begin{equation}\label{1.3}
\left\{\begin{array}{lll}
\frac{rS_{n*}}{1+kI_{n*}}+r\rho I_{n*}-ra(S_{n*}+\rho I_{n*})(S_{n*}+I_{n*})-\frac{\beta S_{n*} I_{n*}}{S_{n*}+I_{n*}}=0, \\
\frac{\beta S_{n*} I_{n*}}{S_{n*}+I_{n*}}-\mu I_{n*}=0,n=(1,2,3).
\end{array}\right.
\end{equation}
\begin{theorem}\label{theorem1}
 Model~(\ref{1.2}) has a trivial equilibrium $E_{0}=(0,0)$ and a semi-trivial equilibrium $E_{1}=(\frac{1}{a},0)$. Furthermore,when $B(k)^2-4A(k)C(k)> 0$ and $C(k)<0$, model~(\ref{1.2}) has a positive equilibrium $E_{2*}=(S_{2*},I_{2*})$, where $S_{2*}=\frac{-B(k) +\sqrt{B(k)^{2}-4 A(k) C(k)}}{2 A(k)}$, $I_{2*}=\frac{(\beta-\mu)S_{2*}}{\mu}$, ($S_{2*}>S_{3*}$, $I_{2*}>I_{3*}$); If $B(k)^2-4A(k)C(k)= 0$, the positive equilibrium $E_{1*}=(S_{1*},I_{1*})$, where $S_{1*}=\frac{-B(k)}{2 A(k)}$, $I_{1*}=\frac{(\beta-\mu)S_{1*}}{\mu}$ and $\beta > \mu$.
\end{theorem}
\begin{proof}
We can calculate that model~(\ref{1.2}) has a trivial equilibrium $E_{0}=(0,0)$ and a semi-trivial equilibrium $E_{1}=(\frac{1}{a},0)$. Furthermore, it can be obtained from Eq.~(\ref{1.3}) that
\begin{equation}\label{1.4}
I_{n*}=\frac{(\beta-\mu)S_{n*}}{\mu},
\end{equation}
and
\begin{align}\label{1.5}
rS_{n*}(S_{n*}+I_{n*})&+r\rho I_{n*}(S_{n*}+I_{n*})(1+kI_{n*})\nonumber\\
&-ra(S_{n*}+\rho I_{n*})(S_{n*}+I_{n*})^2(1+kI_{n*})-\beta S_{n*} I_{n*}(1+kI_{n*})=0,
\end{align}
By taking Eq.~(\ref{1.4}) into Eq.~(\ref{1.5}), we can get
\begin{align}\label{1.6}
r(1+\frac{(\beta-\mu)}{\mu})&+r\rho \frac{(\beta-\mu)}{\mu}(1+\frac{(\beta-\mu)}{\mu})(1+k\frac{(\beta-\mu)S_{n*}}{\mu})\nonumber\\
&-raS_{n*}(1+\rho\frac{(\beta-\mu)}{\mu})(1+\frac{(\beta-\mu)}{\mu})^2(1+k\frac{(\beta-\mu)S_{n*}}{\mu})\nonumber\\
&-\beta \frac{(\beta-\mu)}{\mu}(1+k\frac{(\beta-\mu)S_{n*}}{\mu})=0,
\end{align}
that is
\begin{align}\label{1.7}
r(1+m)&+r\rho m(1+m)(1+kmS_{n*})\nonumber\\
&-raS_{n*}(1+\rho m)(1+m)^2(1+kmS_{n*})\nonumber\\
&-\beta m(1+kmS_{n*})=0,
\end{align}
where
\begin{equation}\label{1.8}
m=\frac{\beta-\mu}{\mu}.
\end{equation}
Then, we obtain
\begin{align}\label{1.9}
kram(1+\rho m)(1+m)^2S_{n*}^2&+(k\beta m^2+ra(1+\rho m)(1+m)^2-kr\rho m^2(1+m))S_{n*}\nonumber\\
&+\beta m-r(1+m)(1+\rho m)=0.
\end{align}
Let
\begin{align}\label{1.10}
A(k)&=kram(1+\rho m)(1+m)^2>0,\\
B(k)&=k\beta m^2+ra(1+\rho m)(1+m)^2-kr\rho m^2(1+m),\\
C(k)&=\beta m-r(1+m)(1+\rho m),
\end{align}
that is
\begin{equation}\label{1.11}
A(k)S_{n*}^2+B(k)S_{n*}+C(k)=0.
\end{equation}
Eq.~(\ref{1.11}) has the following positive solutions:
\begin{equation}\label{1.12}
S_{n*}=\frac{-B(k) \pm \sqrt{B(k)^{2}-4 A(k) C(k)}}{2 A(k)}.
\end{equation}
We then obtain the following results:
\begin{itemize}
  \item Let $B(k)^2-4A(k)C(k)<0$, then model~(\ref{1.2}) has no positive equilibrium.
  \item Let $B(k)^2-4A(k)C(k)=0$, when $B(k)<0$, model~(\ref{1.2}) has a positive equilibrium $E_{1*}=(S_{1*},I_{1*})$, where $S_{1*}=\frac{-B(k)}{2 A(k)}$, $I_{1*}=\frac{(\beta-\mu)S_{1*}}{\mu}$.
  \item Let $B(k)^2-4A(k)C(k)>0$,\begin{enumerate}
                                   \item when $B(k)<0$ and $C(k)>0$, model~(\ref{1.2}) has two positive equilibrium $E_{(2,3)*}=(S_{(2,3)*},I_{(2,3)*})$, where $S_{(2,3)*}=\frac{-B(k) \pm \sqrt{B(k)^{2}-4 A(k) C(k)}}{2 A(k)}$, $I_{(2,3)*}=\frac{(\beta-\mu)S_{(2,3)*}}{\mu}$,
                                   \item when $C(k)<0$, model~(\ref{1.2}) has a positive equilibrium $E_{2*}=(S_{2*},I_{2*})$ where $S_{2*}=\frac{-B(k) +\sqrt{B(k)^{2}-4 A(k) C(k)}}{2 A(k)}$, $I_{2*}=\frac{(\beta-\mu)S_{2*}}{\mu}$, ($S_{2*}>S_{3*}$, $I_{2*}>I_{3*}$),
                                   \item when $B(k)>0$ and $C(k)>0$, model~(\ref{1.2}) has no positive equilibrium.
                                 \end{enumerate}
\end{itemize}
\begin{remark}\label{remark1}
\begin{itemize}
  \item When $B(k_1)=k_1\beta m^2+ra(1+\rho m)(1+m)^2-k_1r\rho m^2(1+m)=0$, we obtain $k_1=\frac{ra(1+\rho m)(1+m)^2}{r\rho m^2(1+m)-\beta m^2}$.
  \item When $C(k)=\beta m-r(1+m)(1+\rho m)=0$, we obtain $\beta=\frac{r(1+\rho m)(1+m)}{m}$.
\end{itemize}
\end{remark}
In summary, according to Remark~\ref{remark1}, we can find that when $C(k)=\beta m-r(1+m)(1+\rho m)>0$ and $\beta >\frac{r(1+\rho m)(1+m)}{m}$ that is, $B(k)=k\beta m^2+ra(1+\rho m)(1+m)^2-kr\rho m^2(1+m)>0$. There are no two positive equilibria in the model~(\ref{1.2}). The proof of the Theorem~\ref{theorem1} is completed.
\end{proof}
\subsection{Stability Analysis}
In this subsection, we will analyze the stability of trivial equilibrium $E_{0}=(0,0)$, semi-trivial equilibrium $E_{1}=(\frac{1}{a},0)$ and nontrivial equilibrium (positive equilibrium) $E_{n*}=(S_{n*},I_{n*})$, $(n=1,2)$.\\
\begin{theorem}\label{theorem2}
\begin{description}
  \item[(1)] $E_{0}=(0,0)$ is a saddle point;
  \item[(2)] If $\beta < \mu$, then $E_{1}=(\frac{1}{a},0)$ is a stable node; if $\beta > \mu$, $E_{1}=(\frac{1}{a},0)$ is a saddle point;
  \item[(3)] When Theorem~\ref{theorem1} holds, if $\ Det({J_{n*}}) =a_{10}b_{01}-a_{01}b_{10} > 0$, there are two conditions:
  \begin{itemize}
    \item $r(m+1)-\beta m> 0$ and $S_{n*}>S^{+}_{n*}$, we can calculate $\ Tr({J_{n*}}) = a_{10}+b_{01}<0$, then $E_{(1,2)*}=(S_{(1,2)*},I_{(1,2)*})$ is stable;
    \item $r(m+1)-\beta m\leq 0$, we can find $\ Tr({J_{n*}}) = a_{10}+b_{01}<0$, then $E_{(1,2)*}=(S_{(1,2)*},I_{(1,2)*})$ is stable.
  \end{itemize}

\end{description}
\end{theorem}
\begin{proof}
We provide stability analysis by calculating the eigenvalues of Jacobian matrix of the model~(\ref{1.2}). Let
\begin{align}\label{1.13}
f(S,I) &= \frac{rS}{1+kI}+r\rho I-ra(S+\rho I)(S+I)-\frac{\beta S I}{S+I},\\
g(S,I) &= \frac{\beta S I}{S+I}-\mu I,
\end{align}
the Jacobian matrix~$J$~for model~(\ref{1.2}) is
\begin{equation}\label{1.14}
J = \left( {\begin{array}{*{20}{c}}
{\frac{{\partial f}}{{\partial S}}}&{\frac{{\partial f}}{{\partial I}}}\\
{\frac{{\partial g}}{{\partial S}}}&{\frac{{\partial g}}{{\partial I}}}
\end{array}} \right),
\end{equation}
where
\[\begin{array}{l}
\frac{{\partial f}}{{\partial S}} = \frac{r}{1+kI}-ra(2S+(\rho+1)I)-\frac{\beta I^2}{(S+I)^2},~\frac{{\partial g}}{{\partial S}} = \frac{\beta I^2}{(S+I)^2},\\
~\frac{{\partial f}}{{\partial I}} = -\frac{rkS}{(1+kI)^2}+r\rho-ra(2\rho I+(\rho+1)S)-\frac{\beta S^2}{(S+I)^2},
~\frac{{\partial g}}{{\partial I}} = \frac{\beta S^2}{(S+I)^2} - \mu.
\end{array}\]

Evaluating the Jacobian matrix for model~(\ref{1.2}) at $E_{0}=(0,0)$, we find
\begin{equation}\label{1.15}
{J_{\rm{0}}} = \left( {\begin{array}{*{20}{c}}
{ r }&{ r\rho}\\
0&{-\mu}
\end{array}} \right),
\end{equation}
the characteristic polynomial is
\begin{equation}\label{1.16}
H_0(\lambda ) = {\lambda ^2} - Tr({J_0})\lambda  + Det({J_0}),
\end{equation}
where $Det({J_0})<0$, so we can see the trivial equilibrium $E_{0}=(0,0)$ is a saddle point.

Given the Jacobian matrix for the model~(\ref{1.2}) evaluated at $E_{1}=(\frac{1}{a},0)$, we find
\begin{equation}\label{1.17}
{J_{\rm{1}}} = \left( {\begin{array}{*{20}{c}}
{ -r }&{-\frac{rk}{a}+r\rho-r(\rho+1)-\beta}\\
0&{\beta -\mu}
\end{array}} \right),
\end{equation}
and the characteristic polynomial is
\begin{equation}\label{1.18}
H_1(\lambda ) = {\lambda ^2} - Tr({J_1})\lambda  + Det({J_1}),
\end{equation}
so we can calculate if $\beta < \mu$, then $E_{1}=(\frac{1}{a},0)$ is a stable node; if $\beta > \mu$, $E_{1}=(\frac{1}{a},0)$ is a saddle point.

The Jacobian matrix for the model~(\ref{1.2}) evaluated at $E_{n*}=(S_{n*},I_{n*})$, $(n=1,2)$ is given by
\begin{equation}\label{1.19}
{J_{n*}} = \left( {\begin{array}{*{20}{c}}
{{a_{10}}}&{{a_{01}}}\\
{{b_{10}}}&{{b_{01}}}
\end{array}} \right),
\end{equation}
where
\[\begin{array}{l}
a_{10} = \frac{r}{1+kI_{n*}}-ra(2S_{n*}+(\rho+1)I_{n*})-\frac{\beta I_{n*}^2}{(S_{n*}+I_{n*})^2},~b_{10} = \frac{\beta I_{n*}^2}{(S_{n*}+I_{n*})^2},\\
~a_{01} = -\frac{rkS_{n*}}{(1+kI_{n*})^2}+r\rho-ra(2\rho I_{n*}+(\rho+1)S_{n*})-\frac{\beta S_{n*}^2}{(S_{n*}+I_{n*})^2},
~b_{01} = \frac{\beta S_{n*}^2}{(S_{n*}+I_{n*})^2} - \mu.
\end{array}\]
The characteristic polynomial is
\begin{equation}\label{1.20}
H_{n*}(\lambda ) = {\lambda ^2} - Tr({J_{n*}})\lambda  + Det({J_{n*}}),
\end{equation}
where
\begin{equation}\label{1.21}
\ Tr({J_{n*}}) = a_{10}+b_{01},\\
\ Det({J_{n*}}) =a_{10}b_{01}-a_{01}b_{10},
\end{equation}
and
\begin{align}\label{1.22}
Tr({J_{n*}}) &= \frac{1}{(1+m)(1+mkS_{n*})}(-ramk(1+m)(2+m(\rho+1))S_{n*}^2\nonumber\\
&+(-\beta km^2-ra(1+m)(2+m(\rho+1)))S_{n*}+r(m+1)-\beta m),
\end{align}
\begin{align}\label{1.23}
Det({J_{n*}}) &=\frac{(\beta-\mu)}{\beta}(\frac{rks(\beta-\mu)}{(1+kI_{n*})^2}+ra(S_{n*}(\rho(\beta-\mu)+\beta+\mu)+I_{n*}(\rho(2\beta-\mu)+\mu)\nonumber\\
&+(\beta-\mu)\mu-\frac{r\mu}{1+kI_{n*}}-r\rho(\beta-\mu)).
\end{align}
If $\ Det({J_{n*}}) =a_{10}b_{01}-a_{01}b_{10} < 0$, then $E_{(1,2)*}=(S_{(1,2)*},I_{(1,2)*})$ is saddle; If $\ Det({J_{n*}}) =a_{10}b_{01}-a_{01}b_{10} > 0$ and $\ Tr({J_{n*}}) = a_{10}+b_{01}<0$, then $E_{2*}=(S_{2*},I_{2*})$ is stable, otherwise, $E_{2*}=(S_{2*},I_{2*})$ is unstable. Next, we discuss the value sign of $\ Tr({J_{n*}}) $ respectively. First, we discuss the sign of $Tr({J_{n*}}) = \frac{1}{(1+m)(1+mkS_{n*})}(-ramk(1+m)(2+m(\rho+1))S_{n*}^2+(-\beta km^2-ra(1+m)(2+m(\rho+1)))S_{n*}+r(m+1)-\beta m)$. From the above, we can easily know that $S_{n*}$ is positive, $-ramk(1+m)(2+m(\rho+1))<0$ and $-\beta km^2-ra(1+m)(2+m(\rho+1))<0$. If $r(m+1)-\beta m\leq 0$, then $Tr({J_{n*}})$ must be less than $0$. If $r(m+1)-\beta m> 0$, then we regard $S_{n*}$ as the solution of $Tr({J_{n*}})=0$, Therefore, it is easy to know that the equation $-ramk(1+m)(2+m(\rho+1))S_{n*}^2+(-\beta km^2-ra(1+m)(2+m(\rho+1)))S_{n*}+r(m+1)-\beta m=0$ has two roots $S^{+}_{n*}, S^{-}_{n*}$, and $S^{+}_{n*}>0>S^{-}_{n*}$. Where
$$ S^{+}_{n*}= \frac{\beta km^2+ra(1+m)(2+m(\rho+1))-\sqrt{\Delta}}{-2ramk(1+m)(2+m(\rho+1))},$$
and
$$\Delta=(\beta km^2+ra(1+m)(2+m(\rho+1)))^2+4(ramk(1+m)(2+m(\rho+1)))(r(m+1)-\beta m).$$
So far, we can conclude that if $S_{n*}>S^{+}_{n*}$, then $Tr({J_{n*}})<0$.
The proof of the theorem is completed.
\end{proof}
\subsection{Example}

In this subsection, we will provide a numerical example to illustrate the case that the positive equilibrium $E_{2*}=(S_{2*},I_{2*})$ is stable. Here, the parameters are $d_1=d_2=0$, $k=0.01$, $\mu=0.55$, and $\rho=0.1$ and see Table~\ref{tab2} for others. Therefore, model~(\ref{1.2}) is in the following form:
\begin{equation}\label{Example}
\left\{\begin{array}{lll}
\frac{d S}{d t}=\frac{0.6 S}{1+0.01 I}+0.06 I-0.6(S+0.1 I)(S+I)-\frac{ S I}{S+I},  \\
\frac{d I}{d t}=\frac{S I}{S+I}-0.55 I. \\
\end{array}\right.
\end{equation}
Under this scenario, $\ Tr({J_{2*}}) =-0.1431<0$ and $\ Det({J_{2*}}) = 0.0452>0$. According to Theorem~\ref{theorem2}, the positive equilibrium $E_{2}*=(S_{2*},I_{2*})=(0.1680,0.1375)$ is stable. The phase trajectory diagram and time series diagram of the numerical example~(\ref{Example}) are shown in Fig.~\ref{fig1}.

\begin{figure}
\centering
\subfigure[]{
\includegraphics[width=7cm]{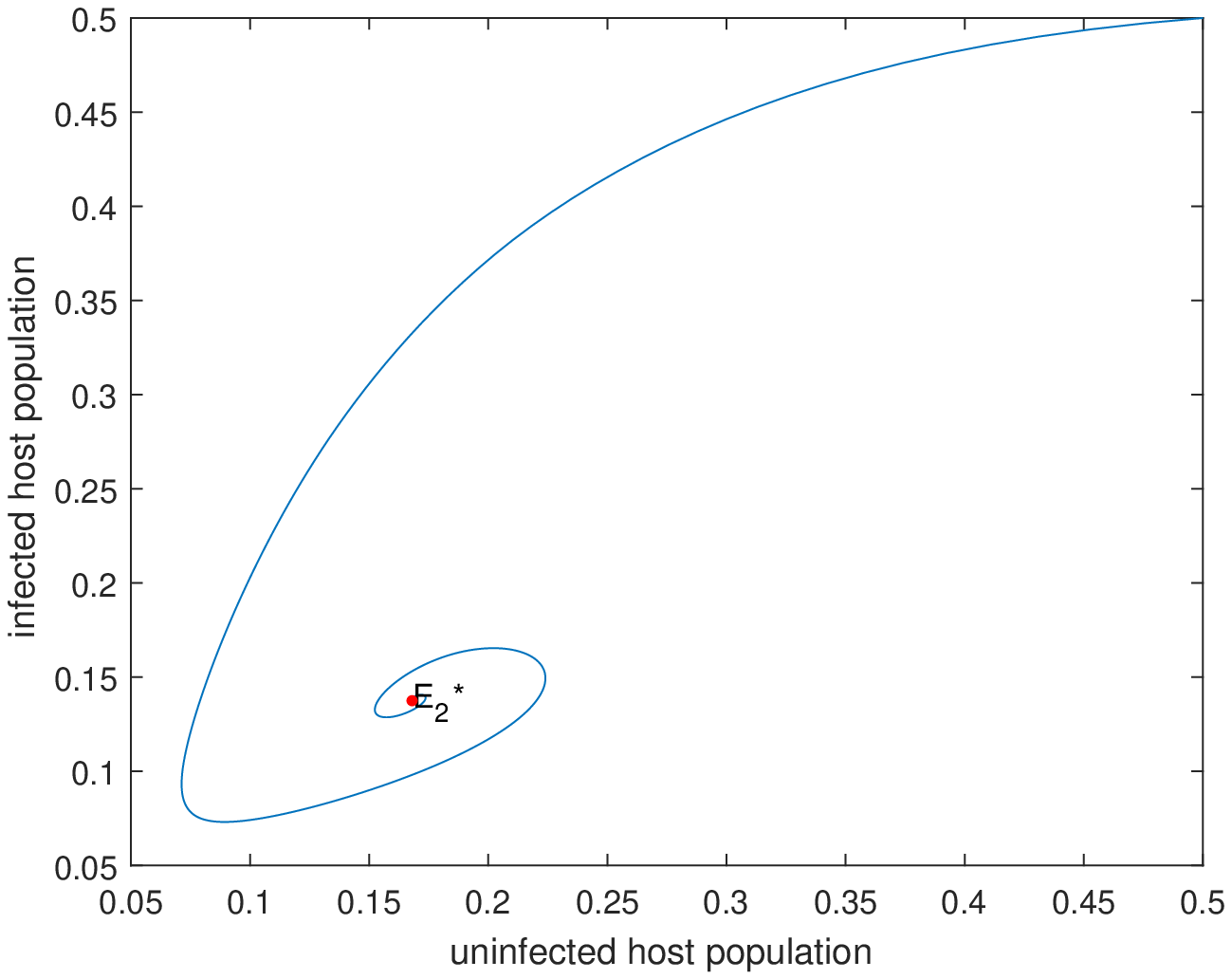}
}
\quad
\subfigure[]{
\includegraphics[width=7cm]{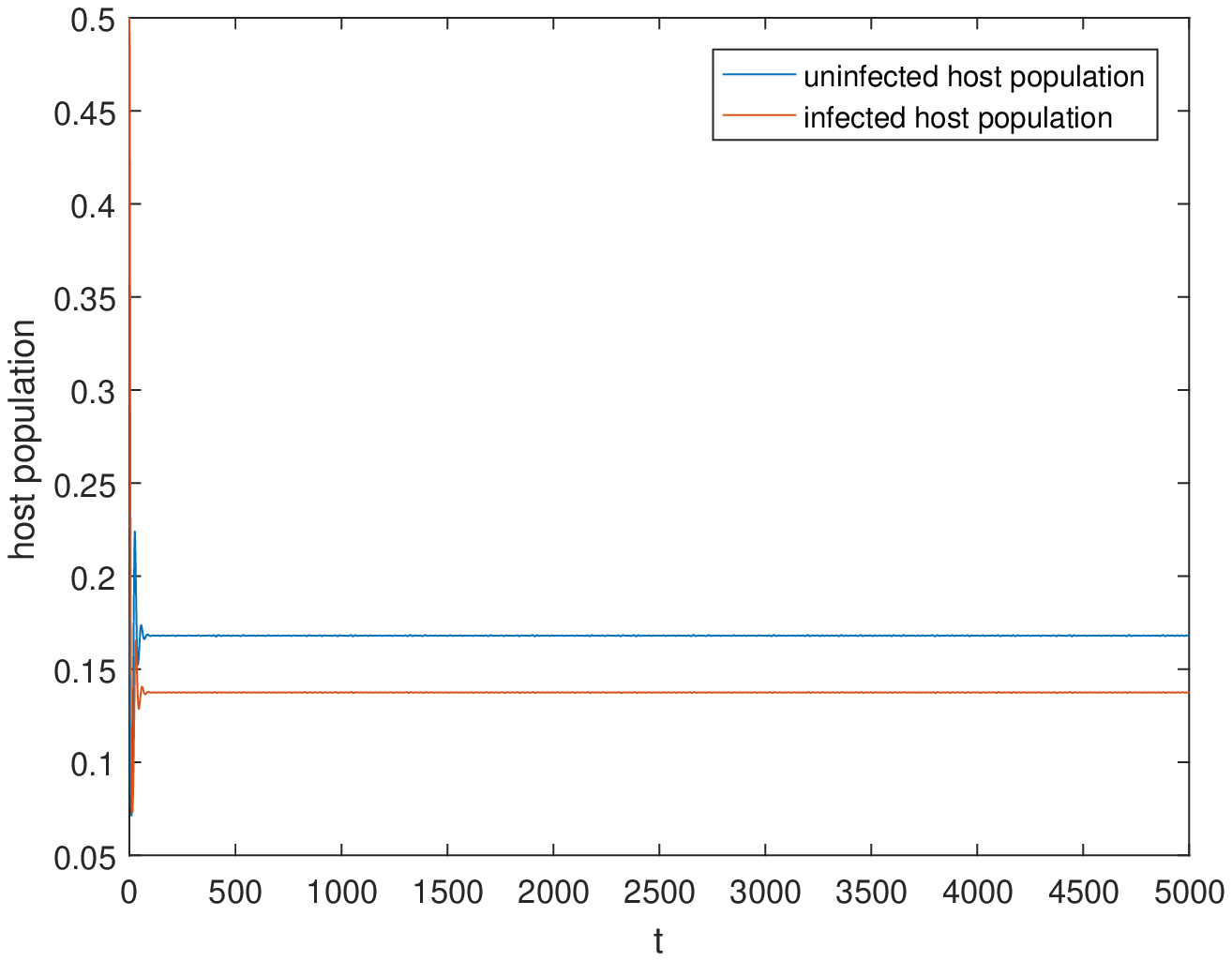}
}

\caption{Phase trajectory diagram of model ~(\ref{1.2}): $E_{2*}=(0.1680,0.1375)$ is stable}
\label{fig1}
\end{figure}

\section{Model with Diffusion}\label{section3}
\subsection{Turing instability}
The purpose of this subsection is to analyze the stability change caused by diffusion, i.e. Turing instability. Similar to references~\cite{ref1,ref5,ref19,ref25,ref26,ref27}, the linearization form of model~(\ref{1.2}) at positive equilibrium $E_{n*} \left(S_{n*}, I_{n*}\right)$ is as follows
\begin{equation}\label{3.1}
\begin{array}{l}
\frac{\partial u}{\partial t}=a_{10} u+a_{01} v-j^{2} d_{1} u, \\
\frac{\partial v}{\partial t}=b_{10} u+b_{01} v-j^{2} d_{2} v,
\end{array}
\end{equation}
where $u=S-S_{n*}, v=I-I_{n*}$, and $(u, v)$ are small perturbations around the equilibrium point $E_{n*} \left(S_{n*}, I_{n*}\right)$ and have the form
\begin{equation}\label{3.2}
\left(\begin{array}{l}
u \\
v
\end{array}\right)=\left(\begin{array}{c}
\epsilon \exp \left(i\mathbf{jr}+\lambda_{j} t\right) \\
\delta \exp \left(i\mathbf{jr}+\lambda_{j} t\right)
\end{array}\right),
\end{equation}
where $\epsilon\ll1$ and $\delta\ll1$, $j=|\mathbf{j}|$ is the wave number, $\mathbf{r}$ is the directional vector and $\lambda_{j}$ is the wave frequency. The characteristic equation of the linearized model~(\ref{3.1}) is given by
\begin{equation}\label{3.3}
(J^J_{n*}-\lambda_{j} I)\left(\begin{array}{l}
u \\
v
\end{array}\right)=0,
\end{equation}
where
\begin{equation}\label{3.4}
J^j_{n*}=\left(\begin{array}{cc}
a_{10}-d_{1} j^{2} & a_{01} \\
b_{10} & b_{01}-d_{2} j^{2}
\end{array}\right).
\end{equation}
Eq.~(\ref{3.3}) can be written as
\begin{equation}\label{3.5}
\lambda_{j}^{2}-\operatorname{Tr}\left(J^j_{n*}\right) \lambda_{j}+\operatorname{Det}\left(J^j_{n*}\right)=0,
\end{equation}
where
\begin{align} \label{3.6}
\operatorname{Tr}\left(J^j_{n*}\right) &=\operatorname{Tr}\left(J_{n*}\right)-\left(d_{1}+d_{2}\right) j^{2}=a_{10}+b_{01}-\left(d_{1}+d_{2}\right) j^{2}, \nonumber\\
\operatorname{Det}\left(J^j_{n*}\right) &=\left(a_{10}-d_{1} j^{2}\right)\left(b_{01}-d_{2} j^{2}\right)-a_{01} b_{10} \nonumber\\
 &=j^{4} d_{1} d_{2}-j^{2}\left(a_{10} d_{2}+b_{01} d_{1}\right)+\operatorname{det}\left(J_{n*}\right),
\end{align}
and $J_{n*}$ is given in Eq.~(\ref{1.19}). The roots of the characteristic Eq.~(\ref{3.4}) are:
\begin{equation}\label{3.7}
\lambda_{j}^{\pm}=\frac{\operatorname{Tr}\left(J^j_{n*}\right) \pm \sqrt{\left(\operatorname{Tr}\left(J^j_{n*}\right)\right)^{2}-4 \operatorname{Det}\left(J^j_{n*}\right)}}{2} .
\end{equation}
When the positive equilibrium of model~(\ref{1.2}) is stable without diffusion, $\operatorname{Tr}\left(J^j_{n*}\right)$ is negative and
\begin{equation}\label{3.8}
\operatorname{Det}\left(J^j_{n*}\right)>0.
\end{equation}
The diffusive model~(\ref{1.2}) will be locally asymptotically stable. We can find that the positive equilibrium is stable without diffusion $(j=0)$, but it may be unstable with diffusion $(j \neq 0)$, i.e., Turing bifurcation. Obviously, $\operatorname{Tr}\left(J^j_{n*}\right) < \operatorname{Tr}\left(J_{n*}\right)<0$, and therefore, the stability of positive equilibrium will change only when
\begin{align}\label{3.9}
\operatorname{Det}\left(J^j_{n*}\right) &=\left(a_{10}-d_{1} j^{2}\right)\left(b_{01}-d_{2} j^{2}\right)-a_{01} b_{10} \nonumber\\
 &=j^{4} d_{1} d_{2}-j^{2}\left(a_{10} d_{2}+b_{01} d_{1}\right)+\operatorname{det}\left(J_{n*}\right)
\end{align}
is negative.
The minimum of $\operatorname{Det}\left(J^j_{n*}\right)$ occurs at $j^{2}=j_{cr}^{2}$, where
\begin{equation}\label{3.10}
j_{c r}^{2}=\frac{a_{10} d_{2}+b_{01} d_{1}}{2 d_{1} d_{2}}>0.
\end{equation}
As $a_{10}+b_{01}<0,~j_{cr}^{2}$ is real and $d_{1}, d_{2}$ are always positive, we must have $a_{10} b_{01}<0$. Thus, a sufficient condition for instability is $\operatorname{Det}\left(j_{c r}^{2}\right)<0$, where
\begin{equation}\label{3.11}
\operatorname{Det}\left(j_{c r}^{2}\right)=\left(a_{10} b_{01}-a_{01} b_{10}\right)-\frac{\left(a_{10} d_{2}+b_{01} d_{1}\right)^{2}}{4 d_{1} d_{2}}.
\end{equation}
Therefore, the condition of Turing instability is as follows:
\begin{equation}\label{3.12}
a_{10} d_{2}+b_{01} d_{1}>2 \sqrt{d_{1} d_{2}} \sqrt{a_{10} b_{01}-a_{01} b_{10}}.
\end{equation}
Obviously, when $j^{2}\in(j_{-}^{2}, j_{+}^{2})$, $\operatorname{Det}\left(j_{c r}^{2}\right)<0$, where
\begin{equation}\label{3.13}
\begin{array}{l}
j_{\pm}^{2}=\frac{\left(a_{10} d_{2}+b_{01} d_{1}\right)\pm\sqrt{\left(a_{10} d_{2}+b_{01} d_{1}\right)^{2}-4 d_{1} d_{2} \operatorname{Det}\left(j_{c r}^{2}\right)}}{2 d_{1} d_{2}}.
\end{array}
\end{equation}

\subsection{Weakly nonlinear analysis}
It is pointed out that the amplitude equation is usually used to describe the evolution of dynamical system near bifurcation, showing a critical slowing down~\cite{ref23}. Note that when the control parameters are close to the threshold of Turing bifurcation, the eigenvalues related to critical modes are close to zero, that is, the critical mode is slow mode, then the whole dynamics can be attributed to the dynamic of active slow mode~\cite{ref19,ref20,ref24,ref25,ref26}. In this section, we will use the standard multiscale analysis to deduce the amplitude equation. We rewrite the transformed form of model~(\ref{1.2}) at the positive spatially homogeneous steady state $E_{n*} \left(S_{2*}, I_{2*}\right)$ as follows and denote by ${\left( {U,V} \right)^T}$ the perturbation solution ${\left( {U-{S_{2*}},V-{I_{2*}}} \right)^T}$ of model~(\ref{1.2}).
\begin{equation}\label{4.1}
\frac{{\partial X}}{{\partial t}} = LX + \mathbf{H},
\end{equation}
where $X = {(U,V)^T}$. The linear operator $L$ can be defined as:
\begin{equation}\label{4.2}
L = \left( {\begin{array}{*{20}{c}}
{a_{10}+d_1{\Delta}}&{a_{01}}\\
{{b_{10}}}&{{b_{01}}+d_2{\Delta}}
\end{array}} \right),
\end{equation}
and $\mathbf{H}$ be given by
\begin{equation}\label{4.3}
\mathbf{H} = \left( {\begin{array}{*{20}{c}}
{{a_{20}}{U^2} + {a_{11}}UV + {a_{02}}{V^2} + {a_{30}}{U^3} + {a_{21}}{U^2}V + {a_{12}}U{V^2} + {a_{03}}{V^3}+ o({\varepsilon ^3})}\\
{{b_{20}}{U^2} + {b_{11}}UV + {b_{02}}{V^2} + {b_{30}}{U^3} + {b_{21}}{U^2}V + {b_{12}}U{V^2} + {b_{03}}{V^3}+ o({\varepsilon ^3})}
\end{array}} \right),
\end{equation}
where\\
${a_{20}} = -ra+\frac{\beta I_{2*}^2}{(S_{2*}+I_{2*})^3}$,
${a_{11}} = -ra(\rho+1)-\frac{2\beta S_{2*} I_{2*}}{(S_{2*}+I_{2*})^3}-\frac{rk}{(1+kI_{2*})^2}$,\\
${a_{02}} = \frac{rk^2S_{2*}}{(1+kI_{2*})^3}-ra\rho-\frac{\beta S_{2*} I_{2*}}{(S_{2*}+I_{2*})^3}$,
${b_{20}} = -\frac{\beta I_{2*}^2}{(S_{2*}+I_{2*})^3}$,\\
${b_{11}} = \frac{2\beta S_{2*} I_{2*}}{(S_{2*}+I_{2*})^3}$,
${b_{02}} = -\frac{\beta S_{2*}^2}{(S_{2*}+I_{2*})^3}$,\\
${a_{30}} = -\frac{\beta I_{2*}^2}{(S_{2*}+I_{2*})^4}$,
${a_{21}} =  \frac{\beta I_{2*} (2 S_{2*}-I_{2*})}{3(S_{2*}+I_{2*})^4}$,\\
${a_{12}} = \frac{1}{3}(\frac{2\beta S_{2*} (2 I_{2*}-S_{2*})}{3(S_{2*}+I_{2*})^4} + \frac{rk^2}{(1+kI_{2*})^3)})$,
${a_{03}} = \frac{1}{3}(\frac{\beta S_{2*} (2 I_{2*}-S_{2*})}{(S_{2*}+I_{2*})^4} - \frac{3rk^3S_{2*}}{(1+kI_{2*})^4})$,\\
${b_{30}} = \frac{\beta S_{2*} I_{2*}^2}{(S_{2*}+I_{2*})^4}$,
${b_{21}} = \frac{\beta I_{2*} (I_{2*}-2S_{2*})}{3(S_{2*}+I_{2*})^4}$,\\
${b_{12}} = \frac{\beta S_{2*} (S_{2*}-2I_{2*})}{3(S_{2*}+I_{2*})^4}$,
and ${b_{03}} = \frac{\beta S_{2*}^2}{(S_{2*}+I_{2*})^4}$.\\
Then, we expand the control parameter $\mu$ near the Turing bifurcation threshold as follows
\begin{equation}\label{4.4}
{\mu _T} - \mu  = \varepsilon {\mu_1} + {\varepsilon ^2}{\mu_2} + {\varepsilon ^3}{\mu_3} + o({\varepsilon ^3}),
\end{equation}
where $\left| \varepsilon  \right| \ll 1$. Similar to Eq.~(\ref{4.4}), we expand the solution $X$, linear operator $L$ and the nonlinear term $\mathbf{H}$ into Taylor series at $\varepsilon  = 0$.
\begin{equation}\label{4.5}
X = \varepsilon \left( {\begin{array}{*{20}{c}}
{{U_1}}\\
{{V_1}}
\end{array}} \right) + {\varepsilon ^2}\left( {\begin{array}{*{20}{c}}
{{U_2}}\\
{{V_2}}
\end{array}} \right) + {\varepsilon ^3}\left( {\begin{array}{*{20}{c}}
{{U_3}}\\
{{V_3}}
\end{array}} \right) + o({\varepsilon ^3}),
\end{equation}
\begin{equation}\label{4.6}
\mathbf{H} = {\varepsilon ^2}{h_2} + {\varepsilon ^3}{h_3} + o({\varepsilon ^3}),
\end{equation}
\begin{equation}\label{4.7}
L = {L_T} + ({\mu _T} - \mu)M,
\end{equation}
where
\begin{equation}\label{4.8}
{h_2} = \left( {\begin{array}{*{20}{c}}
{h_2^1}\\
{h_2^2}
\end{array}} \right) = \left( {\begin{array}{*{20}{c}}
{a_{_{20}}^TU_{^1}^2 + a_{_{11}}^T{U_1}{V_1} + a_{_{02}}^TV_{^1}^2}\\
{b_{_{20}}^TU_{^1}^2 + b_{_{11}}^T{U_1}{V_1} + b_{_{02}}^TV_{^1}^2}
\end{array}} \right),
\end{equation}
and
\begin{align}\label{4.9}
{h_3} &= \left( {\begin{array}{*{20}{c}}
{h_3^1}\\
{h_3^2}
\end{array}} \right) \nonumber\\
&= \left( {\begin{array}{*{20}{c}}
{a_{_{30}}^T{U^3} + a_{_{21}}^T{U^2}V + a_{_{12}}^TU{V^2} + a_{_{03}}^T{V^3} + 2(a_{_{20}}^T{U_1}{U_2} + a_{_{02}}^T{V_1}{V_2}) + a_{_{11}}^T({U_1}{V_2} + {V_1}{U_2})}\\
{b_{_{30}}^T{U^3} + b_{_{21}}^T{U^2}V + b_{_{12}}^TU{V^2} + b_{_{03}}^T{V^3} + 2(b_{_{20}}^T{U_1}{U_2} + b_{_{02}}^T{V_1}{V_2}) + b_{_{11}}^T({U_1}{V_2} + {V_1}{U_2})}
\end{array}} \right)\nonumber\\
&- \left( {\begin{array}{*{20}{c}}
{{\mu _1}({a_{20}}^\prime U_{^1}^2 + {a_{11}}^\prime {U_1}{V_1} + {a_{02}}^\prime V_{^1}^2)}\\
{{\mu _1}({b_{20}}^\prime U_{^1}^2 + {b_{11}}^\prime {U_1}{V_1} + {b_{02}}^\prime V_{^1}^2)}
\end{array}} \right),
\end{align}
with $a_{i j}^{\prime}= \frac{d a_{i j}}{d \mu}$, $b_{i j}^{\prime}=\frac{d b_{i j}}{d \mu}$, ($i,j=0,1,2$). Notice that the linear operator
\begin{equation}\label{4.10}
L = {L_T} + ({\mu _T} - \mu)M,
\end{equation}
where
\begin{equation}\label{4.11}
{L_T} = {\left( {\begin{array}{*{20}{c}}
{a_{10}+d_1{\Delta}}&{a_{01}}\\
{{b_{10}}}&{{b_{01}}+d_2{\Delta}}
\end{array}} \right)_{\mu  = {\mu _T}}}
\end{equation}
and
\begin{equation}\label{4.12}
M = \left( {\begin{array}{*{20}{c}}
{{m_{11}}}&{{m_{12}}}\\
{{m_{21}}}&{{m_{22}}}
\end{array}} \right)
\end{equation}
with
${m_{11}} = \frac{da_{10}}{d\mu}$, ${m_{12}} = \frac{da_{01}}{d\mu}$, ${m_{21}} = \frac{db_{10}}{d\mu}$ and ${m_{22}} = \frac{db_{01}}{d\mu}$ at $U = {S_{2*}}$, $V = {I_{2*}}$.\\
Finally, the following multiple time scales are introduced
\begin{equation}\label{4.13}
\frac{\partial }{{\partial t}} = \varepsilon \frac{\partial }{{\partial {T_1}}} + {\varepsilon ^2}\frac{\partial }{{\partial {T_2}}} + o({\varepsilon ^2}).
\end{equation}
Substituting Eqs.~(\ref{4.2})$-$(\ref{4.13}) into Eq.~(\ref{4.1}) and  expanding it with respect to different orders of ${\varepsilon ^i},(i = 1,2,3)$,
\begin{equation}\label{4.14}
\begin{array}{l}
\varepsilon :{L_T}\left( {\begin{array}{*{20}{c}}
{{U_1}}\\
{{V_1}}
\end{array}} \right) = 0,\\
{\varepsilon ^2}:{L_T}\left( {\begin{array}{*{20}{c}}
{{U_2}}\\
{{V_2}}
\end{array}} \right) = \frac{\partial }{{\partial {T_1}}}\left( {\begin{array}{*{20}{c}}
{{U_1}}\\
{{V_1}}
\end{array}} \right) - {\mu _1}M\left( {\begin{array}{*{20}{c}}
{{U_1}}\\
{{V_1}}
\end{array}} \right) - {h_2},\\
{\varepsilon ^3}:{L_T}\left( {\begin{array}{*{20}{c}}
{{U_3}}\\
{{V_3}}
\end{array}} \right) = \frac{\partial }{{\partial {T_1}}}\left( {\begin{array}{*{20}{c}}
{{U_2}}\\
{{V_2}}
\end{array}} \right) + \frac{\partial }{{\partial {T_2}}}\left( {\begin{array}{*{20}{c}}
{{U_1}}\\
{{V_1}}
\end{array}} \right) - {\mu _1}M\left( {\begin{array}{*{20}{c}}
{{U_2}}\\
{{V_2}}
\end{array}} \right) - {\mu _2}M\left( {\begin{array}{*{20}{c}}
{{U_1}}\\
{{V_1}}
\end{array}} \right) - {h_3}.
\end{array}
\end{equation}
Next, we find the amplitude equation by solving Eq.~(\ref{4.14}). Since $L_T$ has an eigenvector associated with the zero eigenvalue
${\left( {f,1} \right)^T}$ with $f=\frac{a_{10} d_2/d_1-b_{01}}{2 b_{10}}$. The general solution of the first component of Eq.~(\ref{4.14}) can be obtained
\begin{equation}\label{4.15}
\left( {\begin{array}{*{20}{c}}
{{U_1}}\\
{{V_1}}
\end{array}} \right) = \left( {\begin{array}{*{20}{c}}
f\\
1
\end{array}} \right)\left( {\sum\limits_{j = 1}^3 {{W_j}{e^{i{\mathbf{k}_j} \cdot r}} + c.c.} } \right),
\end{equation}
where ${W_j}$ is the amplitude of the mode ${e^{i{\mathbf{k}_j} \cdot r}}$. The second component of Eq.~(\ref{4.14}) is nonhomogeneous, the adjoint operator of ${L_T}$ is $L_{_T}^*$, and it has the following zero eigenvectors form
\begin{equation}\label{4.16}
\left( {\begin{array}{*{20}{c}}
1\\
g
\end{array}} \right){e^{i{\mathbf{k}_j} \cdot r}} + c.c.,\begin{array}{*{20}{c}}
{}&{j = 1,2,3},
\end{array}
\end{equation}
where $g=\frac{b_{01}-a_{10} d_2/d_1}{2 b_{10}d_2/d_1}$. Let
\begin{equation}\label{4.17}
\left( {\begin{array}{*{20}{c}}
{{F_U}}\\
{{F_V}}
\end{array}} \right) = \frac{\partial }{{\partial {T_1}}}\left( {\begin{array}{*{20}{c}}
{{U_1}}\\
{{V_1}}
\end{array}} \right) - {\mu _1}\left( {\begin{array}{*{20}{c}}
{{m_{11}}{U_1} + {m_{12}}{V_1}}\\
{{m_{21}}{U_1} + {m_{22}}{V_1}}
\end{array}} \right) - \left( {\begin{array}{*{20}{c}}
{h_2^1}\\
{h_2^2}
\end{array}} \right).
\end{equation}
Then, with the help of Fredholm solvability condition
\begin{equation}\label{4.18}
\left( {1,g} \right)\left( {\begin{array}{*{20}{c}}
{F_U^j}\\
{F_V^j}
\end{array}} \right) = 0,
\end{equation}
where $F_U^j$ and $F_V^j$ are the coefficients of ${e^{i{\mathbf{k}_j} \cdot r}}$ in ${F_U}$ and ${F_V}$, respectively. It follows after some routine calculation that for ${j_l} = 1,2,3$ and ${j_l} \ne {l_m}$ if $l \ne m$
\begin{equation}\label{4.19}
\left( {f + g} \right)\frac{{\partial {W_{j1}}}}{{\partial {T_1}}} = {\mu _1}{h_3}{W_{j1}} - 2({h_1} + g{h_2}){\bar W_{j2}}{\bar W_{j3}},
\end{equation}
where
\[{h_1} =  - ({f^2}a_{20}^T + fa_{11}^T + a_{02}^T),\]
\[{h_2} =  - ({f^2}b_{20}^T + fb_{11}^T + b_{02}^T),\]
\[{h_3} = f{m_{11}} + {m_{12}} + g(f{m_{21}} + {m_{22}}).\]
Note that the forms of $U_1$ and $V_1$ are given by Eq.~(\ref{4.15}). We have a particular solution for the second component of Eq.~(\ref{4.14}) as follows:
\begin{equation}\label{4.20}
\begin{array}{l}
\left( {\begin{array}{*{20}{c}}
{{U_2}}\\
{{V_2}}
\end{array}} \right) = \left( {\begin{array}{*{20}{c}}
{{{\bar U}_0}}\\
{{{\bar V}_0}}
\end{array}} \right) + \sum\limits_{j = 1}^3 {\left( {\begin{array}{*{20}{c}}
{{{\bar U}_j}}\\
{{{\bar V}_j}}
\end{array}} \right)} {e^{i{\mathbf{k}_j} \cdot r}} + \sum\limits_{j = 1}^3 {\left( {\begin{array}{*{20}{c}}
{{{\bar U}_{jj}}}\\
{{{\bar V}_{jj}}}
\end{array}} \right)} {e^{i2{\mathbf{k}_j} \cdot r}}\\
 + \left( {\begin{array}{*{20}{c}}
{{{\bar U}_{12}}}\\
{{{\bar V}_{12}}}
\end{array}} \right){e^{i({\mathbf{k}_1} - {\mathbf{k}_2}) \cdot r}} + \left( {\begin{array}{*{20}{c}}
{{{\bar U}_{23}}}\\
{{{\bar V}_{23}}}
\end{array}} \right){e^{i({\mathbf{k}_2} - {\mathbf{k}_3}) \cdot r}}\\
 + \left( {\begin{array}{*{20}{c}}
{{{\bar U}_{31}}}\\
{{{\bar V}_{31}}}
\end{array}} \right){e^{i({\mathbf{k}_3} - {\mathbf{k}_1}) \cdot r}} + c.c.
\end{array}
\end{equation}
with the coefficients given below at ${\mu _T} = \mu$\\
\begin{equation}\label{4.21}
\begin{array}{l}
\left( {\begin{array}{*{20}{c}}
{{{\bar U}_0}}\\
{{{\bar V}_0}}
\end{array}} \right) = \left( {\begin{array}{*{20}{c}}
{\frac{{2({b_{01}}{h_1} - {a_{01}}{h_2})}}{{{\Delta _0}}}}\\
{\frac{{2({a_{10}}{h_2} - {b_{10}}{h_1})}}{{{\Delta _0}}}}
\end{array}} \right)\sum\limits_{j = 1}^3 {{{\left| {{W_j}} \right|}^2}} \\
 \equiv \left( {\begin{array}{*{20}{c}}
{{z_{U0}}}\\
{{z_{V0}}}
\end{array}} \right)\sum\limits_{j = 1}^3 {{{\left| {{W_j}} \right|}^2}} ,{{\bar U}_j} = f{{\bar V}_{j,}}\left( {\begin{array}{*{20}{c}}
{{X_{jj}}}\\
{{Y_{jj}}}
\end{array}} \right) \equiv \left( {\begin{array}{*{20}{c}}
{{z_{U1}}}\\
{{z_{V1}}}
\end{array}} \right)W_j^2\\
 = \frac{1}{{({a_{10}} - 4{d_1}j_{cr}^2)({b_{01}} - 4{d_2}j_{cr}^2) - {a_{01}}{b_{10}}}} \times \left( {\begin{array}{*{20}{c}}
{({b_{01}} - 4{d_2}j_{cr}^2){h_1} - {a_{01}}{h_2}}\\
{({a_{10}} - 4{d_1}j_{cr}^2){h_2} - {b_{10}}{h_1}}
\end{array}} \right)W_j^2
\end{array}
\end{equation}
and
\begin{equation}\label{4.22}
\begin{array}{l}
\left( {\begin{array}{*{20}{c}}
{{X_{jk}}}\\
{{Y_{jk}}}
\end{array}} \right) \equiv \left( {\begin{array}{*{20}{c}}
{{z_{U2}}}\\
{{z_{V2}}}
\end{array}} \right){W_j}{{\bar W}_k}\\
 = \frac{1}{{({a_{10}} - 3{d_1}j_{cr}^2)({b_{01}} - 3{d_2}j_{cr}^2) - {a_{01}}{b_{10}}}} \times \left( {\begin{array}{*{20}{c}}
{({b_{01}} - 3{d_2}j_{cr}^2){h_1} - {a_{01}}{h_2}}\\
{({a_{10}} - 3{d_1}j_{cr}^2){h_2} - {b_{10}}{h_1}}
\end{array}} \right){W_j}{{\bar W}_k}.
\end{array}
\end{equation}
For the third component of Eq.~(\ref{4.14}), we apply Fredholm solvability condition again, and get for $j = 1$
\begin{align}\label{4.23}
(f + g)\left( {\frac{{\partial {V_j}}}{{\partial {T_1}}} + \frac{{\partial {W_j}}}{{\partial {T_2}}}} \right) &= {h_3}({\mu _1}{V_j} + {\mu _2}{W_j}) + {h_4}{{\bar W}_l}{{\bar W}_m} + H({{\bar V}_l}{{\bar W}_m} + {{\bar V}_m}{{\bar W}_l})\nonumber\\
 &- ({G_1}{\left| {{W_1}} \right|^2} + {G_2}({\left| {{W_2}} \right|^2} + {\left| {{W_3}} \right|^2})){W_j}
\end{align}
with
\begin{equation}\label{4.24}
{h_4} =  - 2{\mu _1}({a_{20}}^\prime {f^2} + {a_{11}}^\prime f + {a_{02}}^\prime  + g({b_{20}}^\prime {f^2} + {a_{11}}^\prime f + {b_{02}}^\prime )),
\end{equation}
\begin{equation}\label{4.25}
H =  - 2({h_1} + g{h_2}),
\end{equation}
\begin{align}\label{4.26}
{G_1} =  &- (3{a_{30}}{f^3} + 2{a_{11}}f{z_{V0}} + {a_{11}}f{z_{V1}} + 4{a_{20}}f{z_{U0}} + 2{a_{20}}f{z_{U1}}\nonumber\\
 &+ 3{a_{21}}{f^2} + 4{a_{02}}{z_{V0}} + 2{a_{02}}{z_{V1}} + 2{a_{11}}{z_{U0}}\nonumber\\
 &+ {a_{11}}{z_{U1}} + 3{a_{12}}f + 3{a_{03}})\nonumber\\
 &- g({b_{30}}{f^3} + 2{b_{11}}f{z_{V0}} + {b_{11}}f{z_{V1}} + 4{b_{20}}f{z_{U0}} + 2{b_{20}}f{z_{U1}}\nonumber\\
 &+ 3{b_{21}}{f^2} + 4{b_{02}}{z_{V0}} + 2{b_{02}}{z_{V1}} + 2{b_{11}}{z_{U0}}\nonumber\\
 &+ {b_{11}}{z_{U1}} + 3{b_{12}}f + 3{b_{03}}),
\end{align}
and
\begin{align}\label{4.27}
{G_2} =  &- (6{a_{30}}{f^3} + 2{a_{11}}f{z_{V0}} + {a_{11}}f{z_{V2}} + 4{a_{20}}f{z_{U0}} + 2{a_{20}}f{z_{U2}}\nonumber\\
 &+ 6{a_{21}}{f^2} + 4{a_{02}}{z_{V0}} + 2{a_{02}}{z_{V2}} + 2{a_{11}}{z_{U0}}\nonumber\\
 &+ {a_{11}}{z_{U2}} + 6{a_{12}}f + 6{a_{03}})\nonumber\\
 &- g(6{b_{30}}{f^3} + 2{b_{11}}f{z_{V0}} + {b_{11}}f{z_{V2}} + 4{b_{20}}f{z_{U0}} + 2{b_{20}}f{z_{U2}}\nonumber\\
 &+ 6{b_{21}}{f^2} + 4{b_{02}}{z_{V0}} + 2{b_{02}}{z_{V2}} + 2{b_{11}}{z_{U0}}\nonumber\\
 &+ {b_{11}}{z_{U2}} + 6{b_{12}}f + 6{b_{03}}).
\end{align}
The amplitude equation Eq.~(\ref{4.28}) of amplitude $A_j$ is given as follows, by combining Eq.~(\ref{4.19}) with Eq.~(\ref{4.23})
\begin{equation}\label{4.28}
{\tau _0}\frac{{\partial {A_j}}}{{\partial t}} = \mu {A_j} + h{{\bar A}_l}{{\bar A}_m}
 - ({g_1}{\left| {{A_1}} \right|^2} + {g_2}({\left| {{A_2}} \right|^2} + {\left| {{A_3}} \right|^2})){A_j}
\end{equation}
where
\begin{equation}\label{4.29}
{\tau _0} = \frac{{f + g}}{{{\mu _T}[f{m_{11}} + {m_{12}} + g(f{m_{21}} + {m_{22}})]}},
\end{equation}
\begin{equation}\label{4.30}
\sigma = \frac{{{\mu _T} - \mu }}{{{\mu _T}}},
\end{equation}
\begin{equation}\label{4.31}
h = \frac{H}{{{\mu _T}[f{m_{11}} + {m_{12}} + g(f{m_{21}} + {m_{22}})]}},
\end{equation}
\begin{equation}\label{4.32}
{g_i} = \frac{{{G_i}}}{{{\mu _T}[f{m_{11}} + {m_{12}} + g(f{m_{21}} + {m_{22}})]}}.
\end{equation}
It is worth noting that Eq.~(\ref{4.28}) is in complex form. According to the method of reference~\cite{ref21}, for the convenience of discussion, we convert it into real form with the help of ${A_j} = {\rho _j}\exp (i{\varphi _j})$, as follow:
\begin{equation}\label{4.33}
\left\{ {\begin{array}{*{20}{c}}
{{\tau _0}\frac{{\partial \varphi }}{{\partial t}} =  - h\frac{{\rho _1^2\rho _2^2 + \rho _1^2\rho _3^2 + \rho _2^2\rho _3^2}}{{{\rho _1}{\rho _2}{\rho _3}}}\sin \varphi },\\
{{\tau _0}\frac{{\partial {\rho _1}}}{{\partial t}} =  \sigma {\rho _1} + h{\rho _2}{\rho _3}\cos \varphi  - {g_1}\rho _1^3 - {g_2}(\rho _3^2 + \rho _2^2){\rho _1}},\\
{{\tau _0}\frac{{\partial {\rho _2}}}{{\partial t}} =  \sigma {\rho _2} + h{\rho _1}{\rho _3}\cos \varphi  - {g_1}\rho _2^3 - {g_2}(\rho _3^2 + \rho _1^2){\rho _2}},\\
{{\tau _0}\frac{{\partial {\rho _3}}}{{\partial t}} =  \sigma {\rho _3} + h{\rho _2}{\rho _1}\cos \varphi  - {g_1}\rho _3^3 - {g_2}(\rho _1^2 + \rho _2^2){\rho _3}},
\end{array}} \right.
\end{equation}
where ${\rho _j}$ are the real amplitudes and ${\varphi _j}$ are the phase angles, and $\varphi  = {\varphi _1} + {\varphi _2} + {\varphi _3}$. Since we only focus on the stable steady state and notice the fact that $h{\rho _i} \ne 0$, according to first equation of~(\ref{4.33}), we obtain $\varphi  = 0$ or $\pi $. In addition, we know that ${\tau _0} > 0$, which means that when $h > 0$, the state corresponding to $\varphi  = 0$ is stable, while when $h < 0$, the state corresponding to $\varphi  = \pi $ is stable. After that, model of amplitude equations~(\ref{4.33}) becomes
\begin{equation}\label{4.34}
\left\{ {\begin{array}{*{20}{c}}
{{\tau _0}\frac{{\partial {\rho _1}}}{{\partial t}} = \sigma {\rho _1} + \left| h \right|{\rho _2}{\rho _3} - {g_1}\rho _1^3 - {g_2}(\rho _3^2 + \rho _2^2){\rho _1}},\\
{{\tau _0}\frac{{\partial {\rho _2}}}{{\partial t}} = \sigma {\rho _2} + \left| h \right|{\rho _1}{\rho _3} - {g_1}\rho _2^3 - {g_2}(\rho _3^2 + \rho _1^2){\rho _2}},\\
{{\tau _0}\frac{{\partial {\rho _3}}}{{\partial t}} = \sigma {\rho _3} + \left| h \right|{\rho _2}{\rho _1} - {g_1}\rho _3^3 - {g_2}(\rho _1^2 + \rho _2^2){\rho _3}}.
\end{array}} \right.
\end{equation}
The amplitude equations are usually valid only when the control parameter $\mu$ is in Turing space. It is not difficult to see that the above ordinary differential equations~(\ref{4.34}) have five equilibrium points, corresponding to five steady states ~\cite{ref19,ref20,ref21,ref22,ref24,ref25,ref26}. Considering the symmetry of the model, we have
\begin{itemize}
  \item Model~(\ref{4.34}) always makes an equilibrium ${E_0} = (0,0,0)$ is stable for $\sigma  < {\sigma _2} = 0$ and unstable for $\sigma  > {\sigma _2}$.
  \item Model~(\ref{4.34}) has an equilibrium ${E_s} = (\sqrt {\frac{\sigma }{{{g_1}}}},0,0)$ corresponding to stripe patterns, which is stable for $\mu  > {\sigma _3} = \frac{{{h^2}{g_1}}}{{{{({g_2} - {g_1})}^2}}}$ and unstable for $\sigma  > {\sigma _3}$.
  \item  Model~(\ref{4.34}) has an equilibrium ${E_h} = ({\rho _{h}^{1}}^\pm,{\rho _{h}^{2}}^\pm,{\rho _{h}^{3}}^\pm)$ corresponding to hexagon patterns, with $\varphi  = 0$ or $\varphi  = \pi$, and ${\rho _{h}^{(1,2,3)}}^+ = \frac{{\left| h \right| + \sqrt {{h^2} + 4({g_1} + 2{g_2})\sigma } }}{{2({g_1} + 2{g_2})}}$ is stable for $\sigma  < {\sigma _4} = \frac{{{h^2}(2{g_1} + {g_2})}}{{{{({g_2} - {g_1})}^2}}}$, ${\rho _{h}^{(1,2,3)}}^-  = \frac{{\left| h \right| - \sqrt {{h^2} + 4({g_1} + 2{g_2})\sigma } }}{{2({g_1} + 2{g_2})}}$ is unstable. Where
\[{\rho _{h}^{1}}^\pm ={\rho _{h}^{2}}^\pm = {\rho _{h}^{3}}^\pm =\frac{{\left| h \right| \pm \sqrt {{h^2} + 4({g_1} + 2{g_2})\sigma } }}{{2({g_1} + 2{g_2})}}.\]
  \item Model~(\ref{4.34}) has an equilibrium ${E_m} = ({\rho _{m}^1},{\rho _{m}^2},{\rho _{m}^3})$ corresponding to mixed patterns, with ${g_1} > {g_2}$, $\sigma > {g_1}\rho _1^2$ and which is unstable. Where
\[{\rho _{m}^1} = \frac{{\left| h \right|}}{{{g_2} - {g_1}}},{\rho _{m}^2} = {\rho _{m}^3} = \sqrt {\frac{{\sigma  - {g_1}\rho _1^2}}{{{g_2} + {g_1}}}} .\]
\end{itemize}

\section{Numerical simulations}\label{section4}
In this section, we simulate the spatial model~(\ref{1.2}) by using the two-dimensional positive bounded domain $\Omega \subseteq \mathrm{R}_{+}^{2}$ with the initial condition of a positive random number between $0$ and $1$ with zero-flux boundary. We discretize the time and space, respectively, and distribute the space in the grid of $L_{x}=L_{y}=20$, where the space step is $\Delta h=0.2$, and the time step is $\Delta t=0.01$. The reaction terms are discretized by Euler scheme. The standard five-point explicit finite difference scheme is adopted for the Laplace operator (i.e., diffusion term)~\cite{ref1,ref24,ref26,ref27,ref34}, as follows:
\begin{align}
\Delta_{\Delta h} S_{i, j}^{n}&=\frac{S_{i+1, j}^{n}+S_{i-1, j}^{n}+S_{i, j+1}^{n}+S_{i, j-1}^{n}-4 S_{i, j}^{n}}{\Delta h^{2}},\\
\Delta_{\Delta h} I_{i, j}^{n}&=\frac{I_{i+1, j}^{n}+I_{i-1, j}^{n}+I_{i, j+1}^{n}+I_{i, j-1}^{n}-4 I_{i, j}^{n}}{\Delta h^{2}}.
\end{align}
The parameter values are shown in Table~\ref{tab2}.
\begin{table*}[ht]
    \centering
    \caption{Description of parameters and their fixed values for model~(\ref{1.2}).}
\begin{tabular}{lll}
\hline\hline Parameter & Description & Value \\
\hline$r$ & maximum birth rate of the hosts & $0.60$~\cite{ref1}\\
$k$ & level of fear which drives anti-infected host behavior of uninfected host & control variable \\
$\rho$ & reducing reproduction ability of infected hosts & control variable\\
$1/a$ & carrying capacity & $1$~\cite{ref1,ref3} \\
$\beta$ & disease transmission rate & $1$~\cite{ref1,ref3}\\
$\mu$ & death rate of the infected host population & control variable\\
$d_1$ & diffusion coefficients of $S$ & control variable \\
$d_2$ & diffusion coefficients of $I$ & control variable\\
\hline\hline
\end{tabular}
    \label{tab2}
\end{table*}

\subsection{Pattern formation by death rate}
In the numerical simulation, we found different shapes of patterns, and found that the distributions of uninfected host $S$ and infected host $I$ are always the same type. We, therefore, only consider the distribution of uninfected hosts $S$ for the sake of brevity. The main factors affecting Turing instability were death rate, fear effect, reducing reproduction ability of infected hosts and diffusion coefficient. Therefore, the idea of our numerical experiment is to consider the influence of different parameters on the root of the characteristic equation~(\ref{3.5}), and to simulate the population distribution under each parameter, that is, pattern formation. First, we plan to simulate the effects of different mortality rates on population distribution. From Fig.~\ref{fig2} and Fig.~\ref{fig3}, we set the diffusion coefficients of $S$ and $I$ are $d_1=0.01$ and $d_2=0.25$~\cite{ref1}, natural death rate $\mu$ to $\mu=0.4$, $\mu=0.51$, $\mu=0.54$, $\mu=0.55$, $\mu=0.56$ and $\mu=0.57$, other parameters are $k=0.4$, $\rho=0.1$~\cite{ref1} and shown in Table~\ref{tab2}. The corresponding positive equilibrium are $E_{2*}=(0.0434, 0.0650)$, $E_{2*}=(0.1110, 0.1066)$, $E_{2*}=(0.1364, 0.1162)$, $E_{2*}=(0.1456, 0.1191)$, $E_{2*}=(0.1551, 0.1219)$, $E_{2*}=(0.1650, 0.1245)$. Fig.~\ref{fig2} shows that maximum real part of the roots of Eq.~(\ref{3.5}) and $Det(J^j_{n*})$ against $j$ for different $\mu$ taken from the Turing region, respectively. we observe the change of the pattern shape, as presented in Fig.~\ref{fig3}. The results show that natural death rate can control the growth of pattern, following: hot spots (Fig.~\ref{fig3}(a)) $\rightarrow$ hot spots-stripes (Fig.~\ref{fig3}(b)) $\rightarrow$ hot stripes (Fig.~\ref{fig3}(c)) $\rightarrow$ cold stripes (Fig.~\ref{fig3}(d)) $\rightarrow$ cold spots-stripes (Fig.~\ref{fig3}(e)) $\rightarrow$ cold spots (Fig.~\ref{fig3}(f)).
\begin{figure}
\centering
\subfigure[]{
\includegraphics[width=7cm]{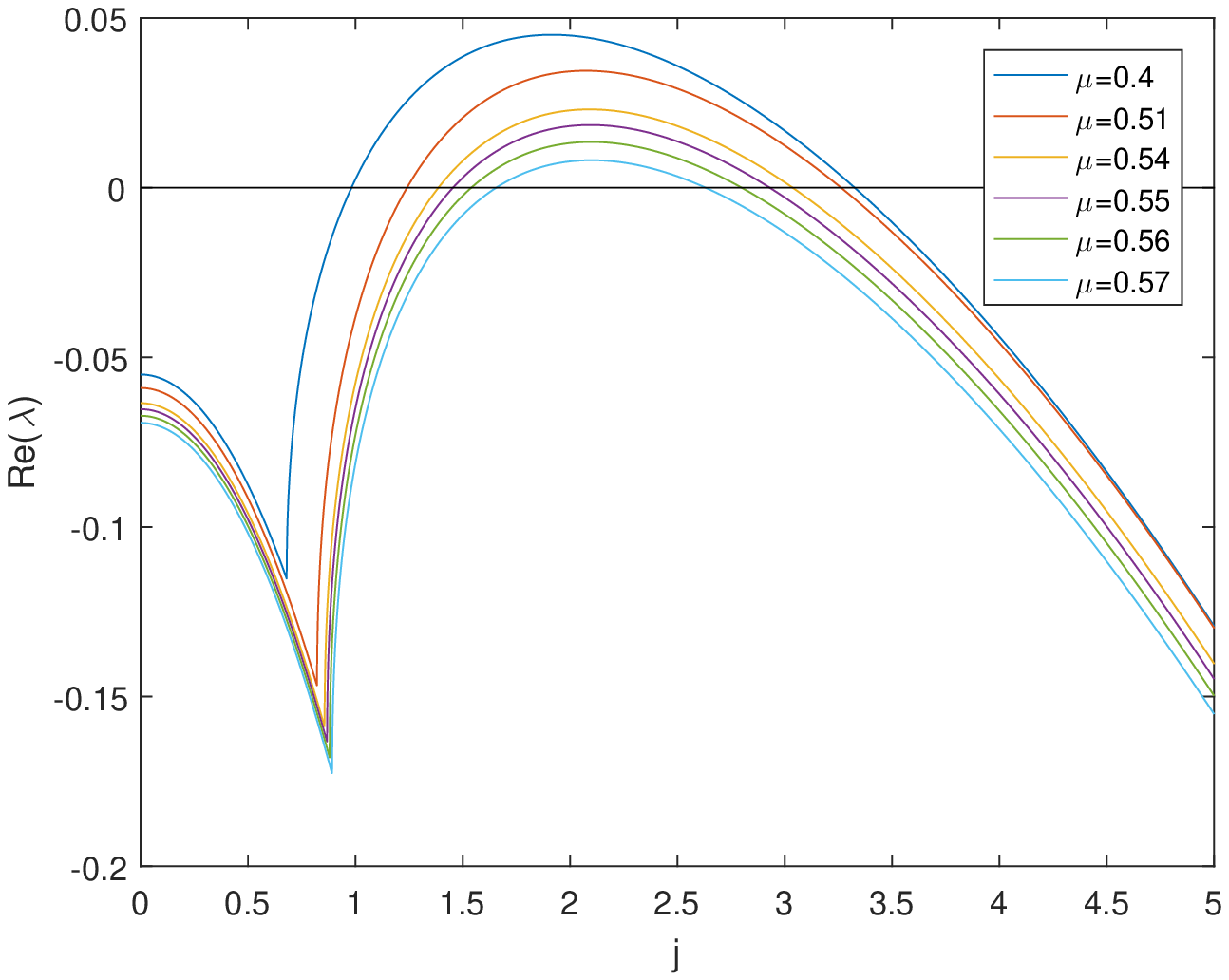}
}
\quad
\subfigure[]{
\includegraphics[width=7cm]{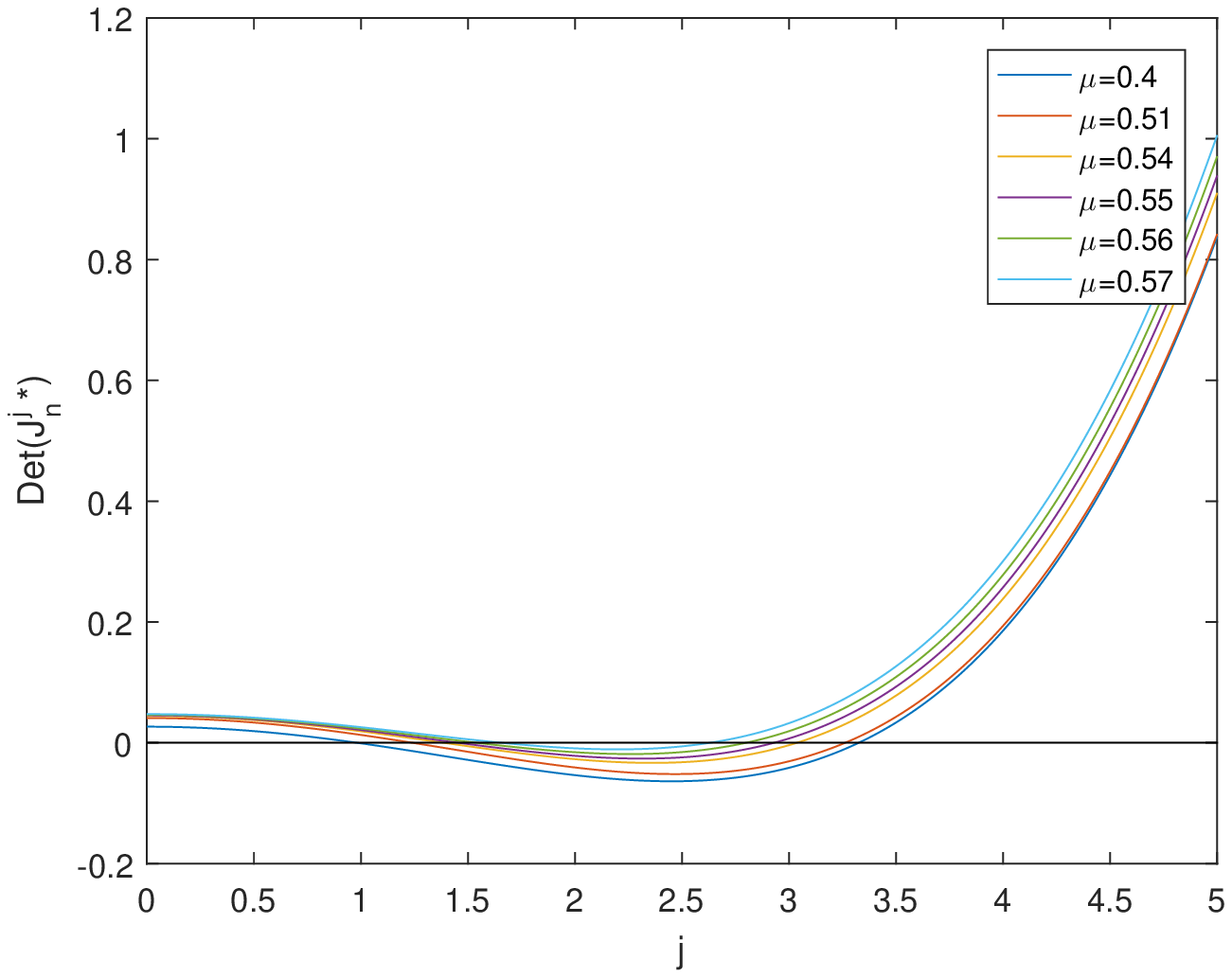}
}

\caption{Plots of (a) the maximum real part of the roots of Eq.~(\ref{3.5}) and (b) $Det(J^j_{n*})$ against $j$ for different $\mu$ taken from the Turing region. Other parameters are set to $d_{1}=0.01$, $d_{2}=0.25$, $\rho=0.1$ and $k=0.4$.}
\label{fig2}
\end{figure}

\begin{figure}
\centering
\subfigure[$\mu=0.4$, $t=2000$]{
\includegraphics[width=4cm]{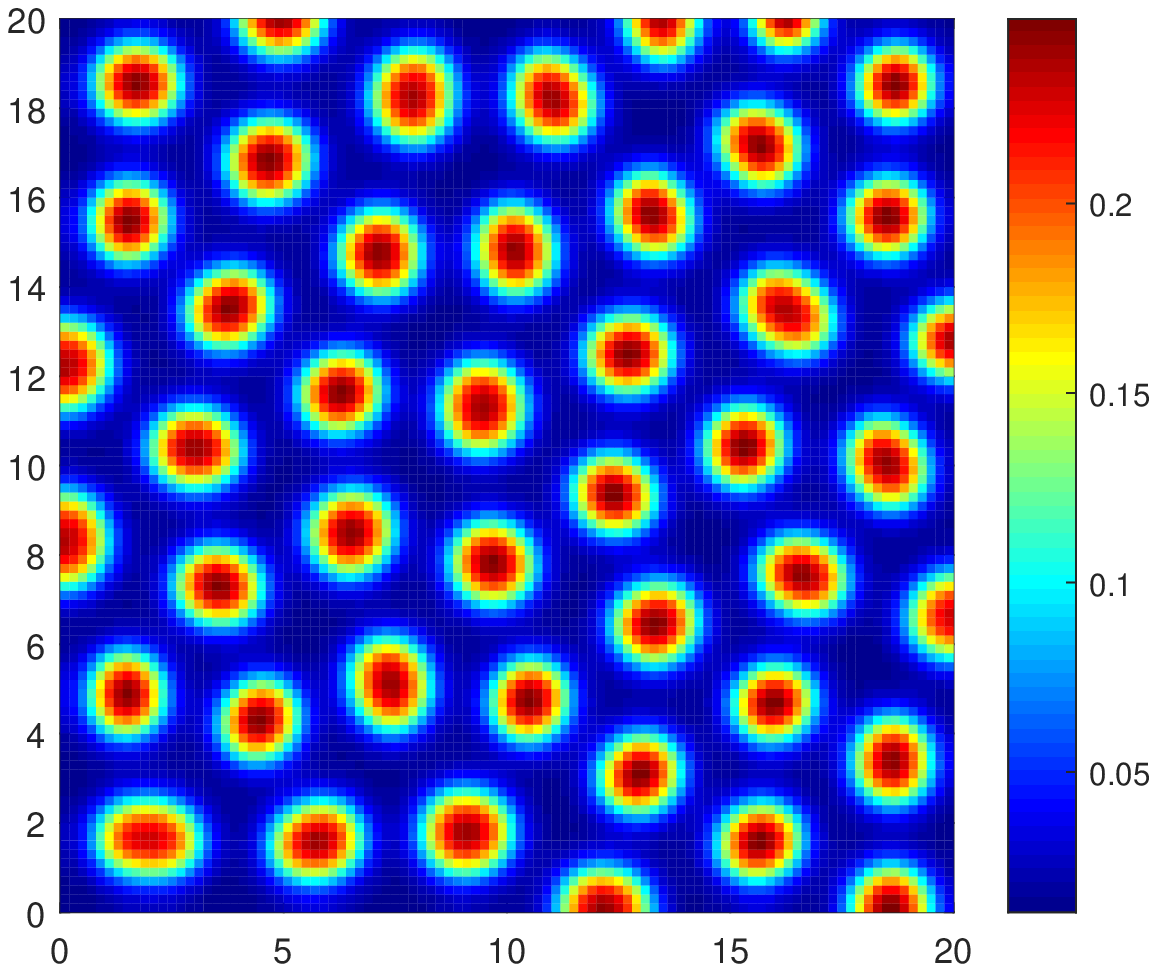}
}
\quad
\subfigure[$\mu=0.51$, $t=2000$]{
\includegraphics[width=4cm]{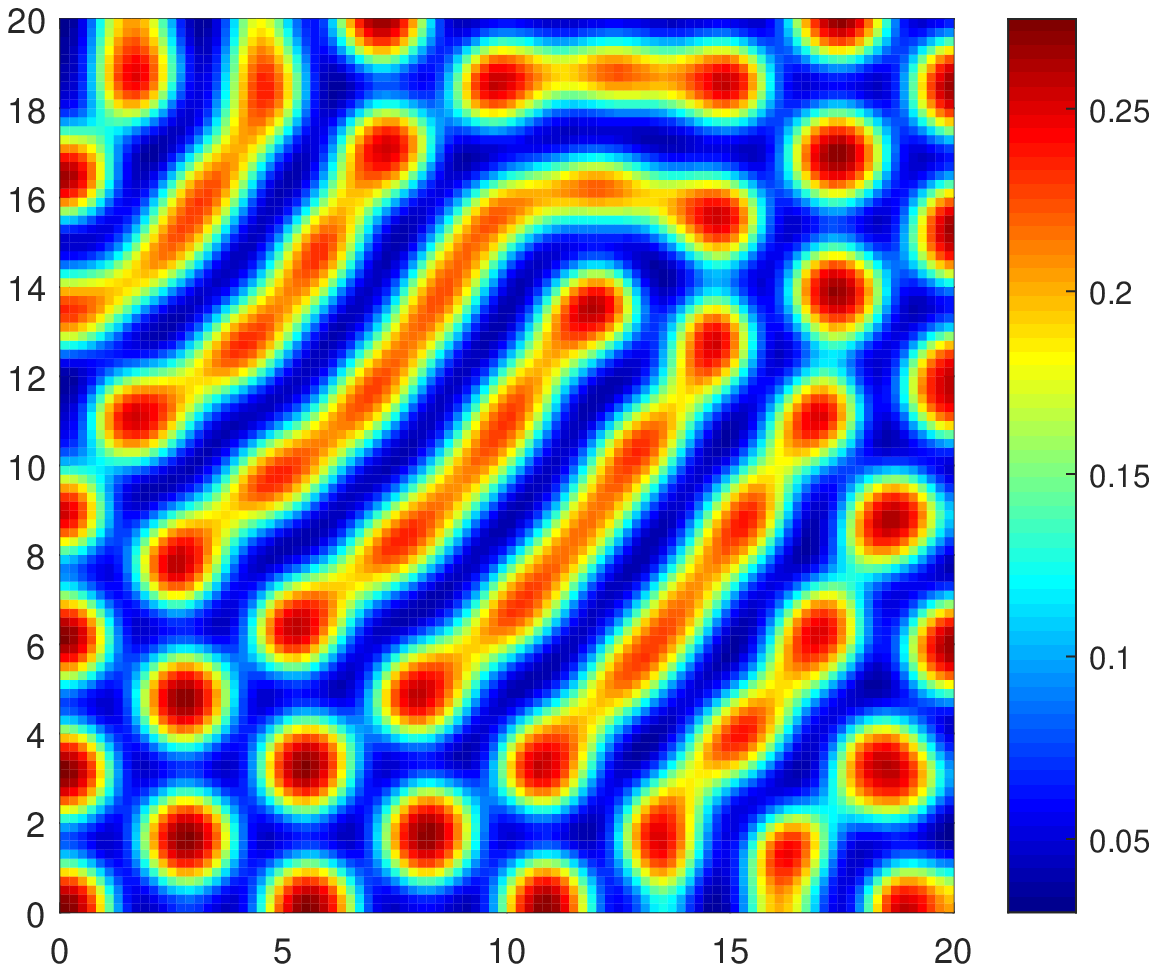}
}
\quad
\subfigure[$\mu=0.54$, $t=2000$]{
\includegraphics[width=4cm]{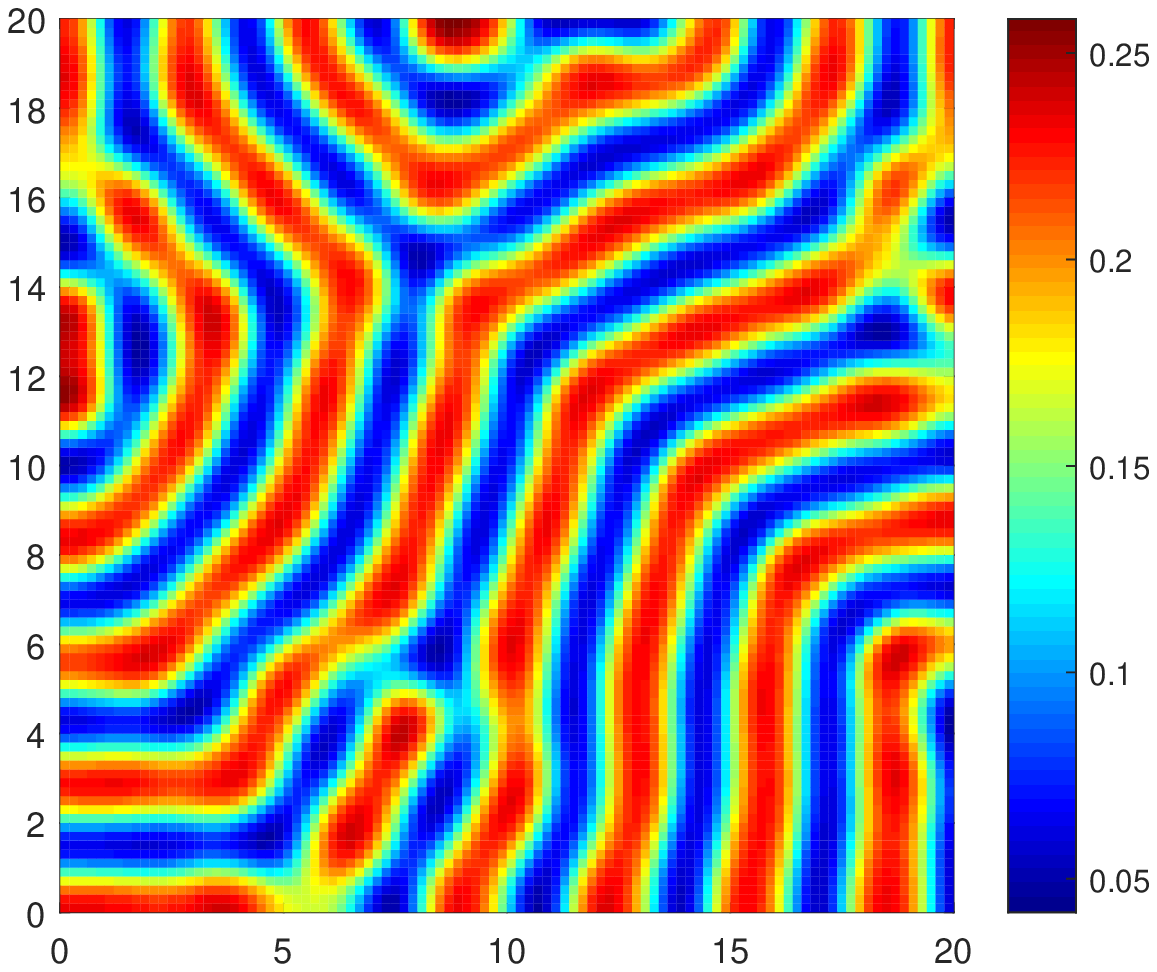}
}
\quad

\subfigure[$\mu=0.55$, $t=2000$]{
\includegraphics[width=4cm]{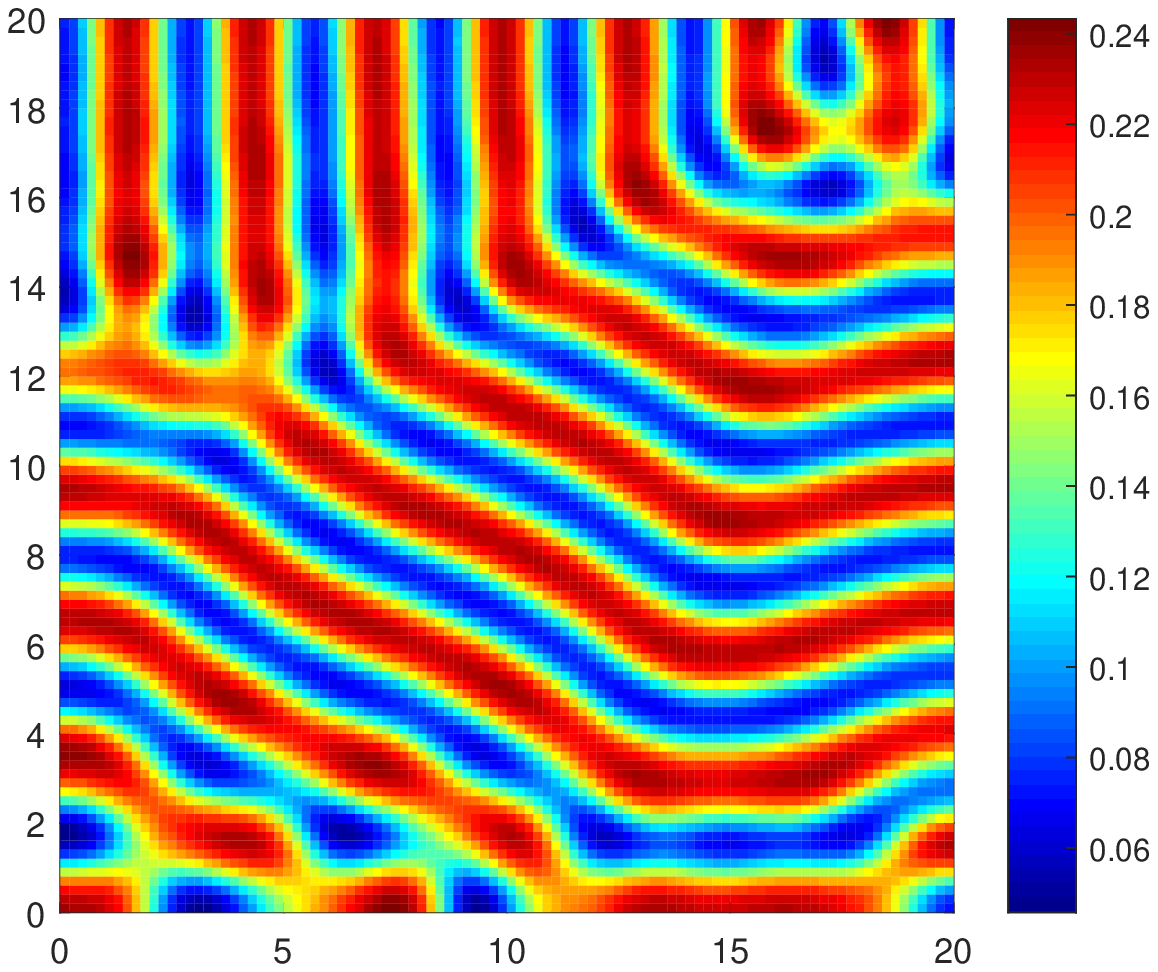}
}
\quad
\subfigure[$\mu=0.56$, $t=2000$]{
\includegraphics[width=4cm]{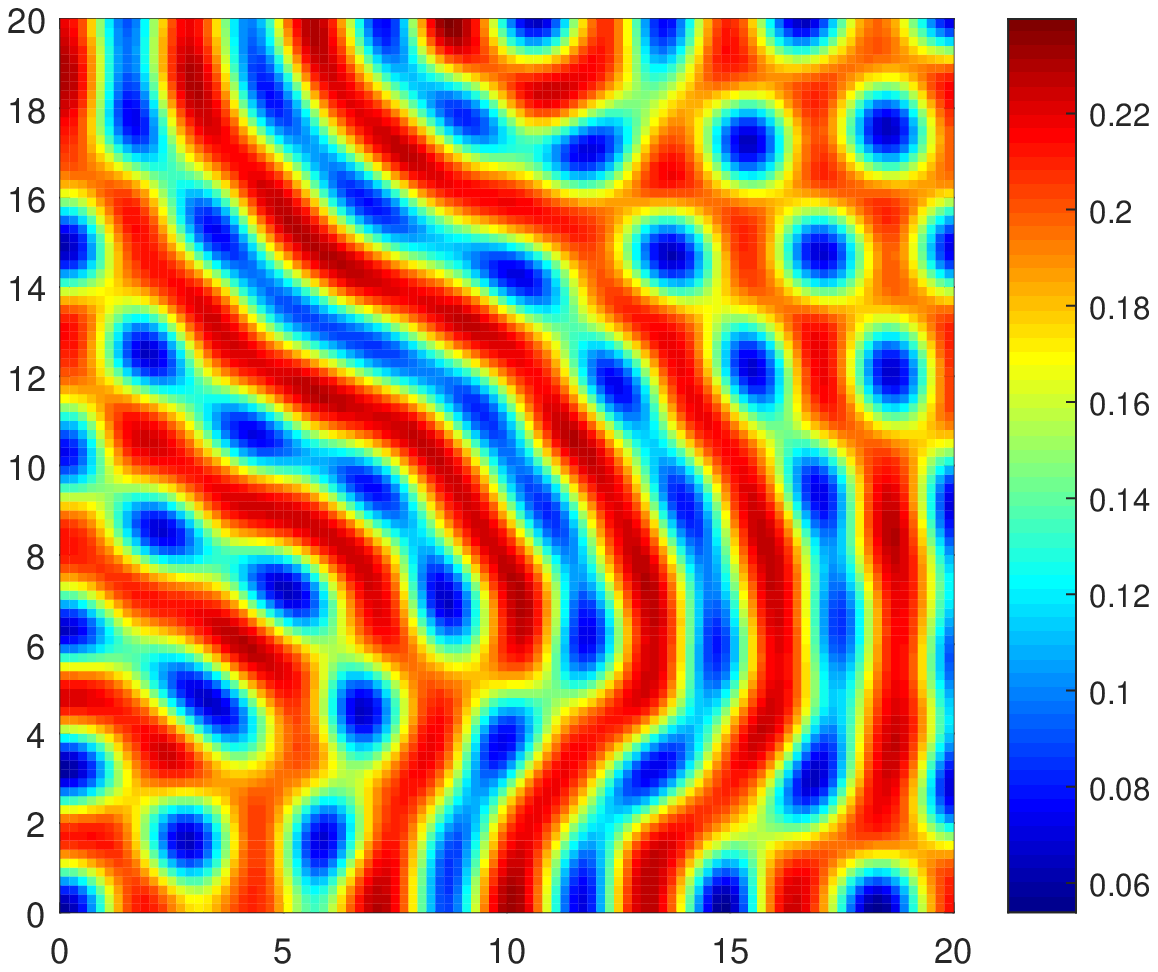}
}
\quad
\subfigure[$\mu=0.57$, $t=2000$]{
\includegraphics[width=4cm]{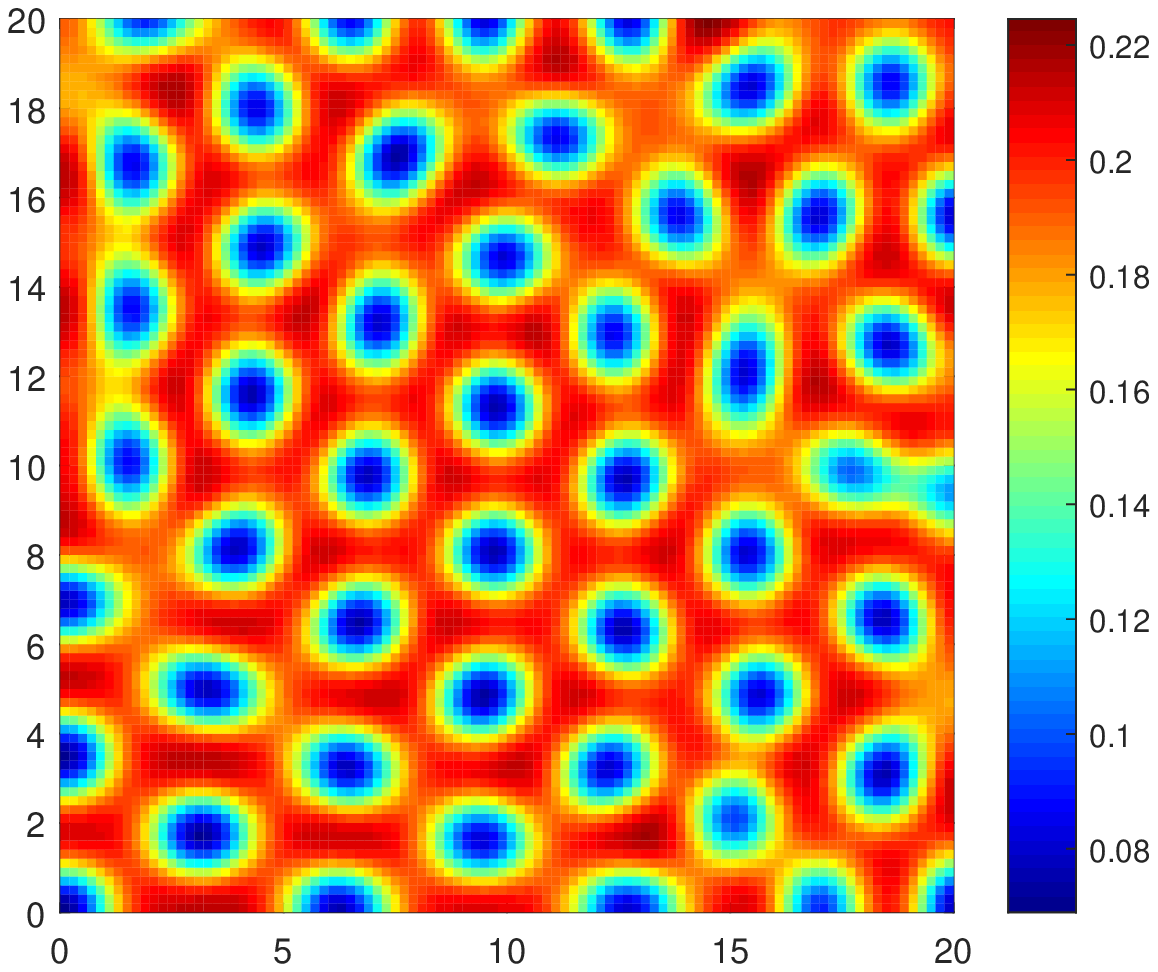}
}
\caption{Natural death rate of infected hosts: (a) $\mu=0.4$, (b) $\mu=0.51$, (c) $\mu=0.54$, (d) $\mu=0.55$, (e) $\mu=0.56$, (f) $\mu=0.57$. Natural death rate control the growth of pattern, following: hot spots $\rightarrow$ hot spots-stripes $\rightarrow$ hot stripes $\rightarrow$ cold stripes $\rightarrow$ cold spots-stripes $\rightarrow$ cold spots for $d_{1}=0.01$, $d_{2}=0.25$, $\rho=0.1$ and $k=0.4$. }
\label{fig3}
\end{figure}

\subsection{Pattern formation by fear effect}
In this subsection, we plan to observe the distribution of the population under different degrees of fear. From Fig.~\ref{fig4} and Fig.~\ref{fig5}, we set the diffusion coefficients of $S$ and $I$ are $d_1=0.01$ and $d_2=0.25$~\cite{ref1}, fear effect $k$ to $k=0.01$, $k=0.25$, $k=0.42$, $k=0.5$, $k=1.8$ and $k=3$, other parameters are $\mu=0.55$, $\rho=0.1$~\cite{ref1} and shown in Table~\ref{tab2}. The corresponding positive equilibrium are $E_{2*}=(0.1680, 0.1375)$, $E_{2*}=(0.1532, 0.1254)$, $E_{2*}=(0.1446, 0.1183)$, $E_{2*}=(0.1410, 0.1153)$, $E_{2*}=(0.1022, 0.0836)$, $E_{2*}=(0.0828, 0.0678)$. Fig.~\ref{fig4} shows that maximum real part of the roots of Eq.~(\ref{3.5}) and $Det(J^j_{n*})$ against $j$ for different $k$ taken from the Turing region. Respectively, we observed the change of the pattern shape. The results showed that fear effect can control the growth of pattern: cold spots (Fig.~\ref{fig5}(a)) $\rightarrow$ cold spots-stripes (Fig.~\ref{fig5}(b)) $\rightarrow$ cold stripes (Fig.~\ref{fig5}(c)) $\rightarrow$ hot stripes (Fig.~\ref{fig5}(d)) $\rightarrow$ hot spots-stripes (Fig.~\ref{fig5}(e)) $\rightarrow$ hot spots (Fig.~\ref{fig5}(f)).

\begin{figure}
\centering
\subfigure[]{
\includegraphics[width=7cm]{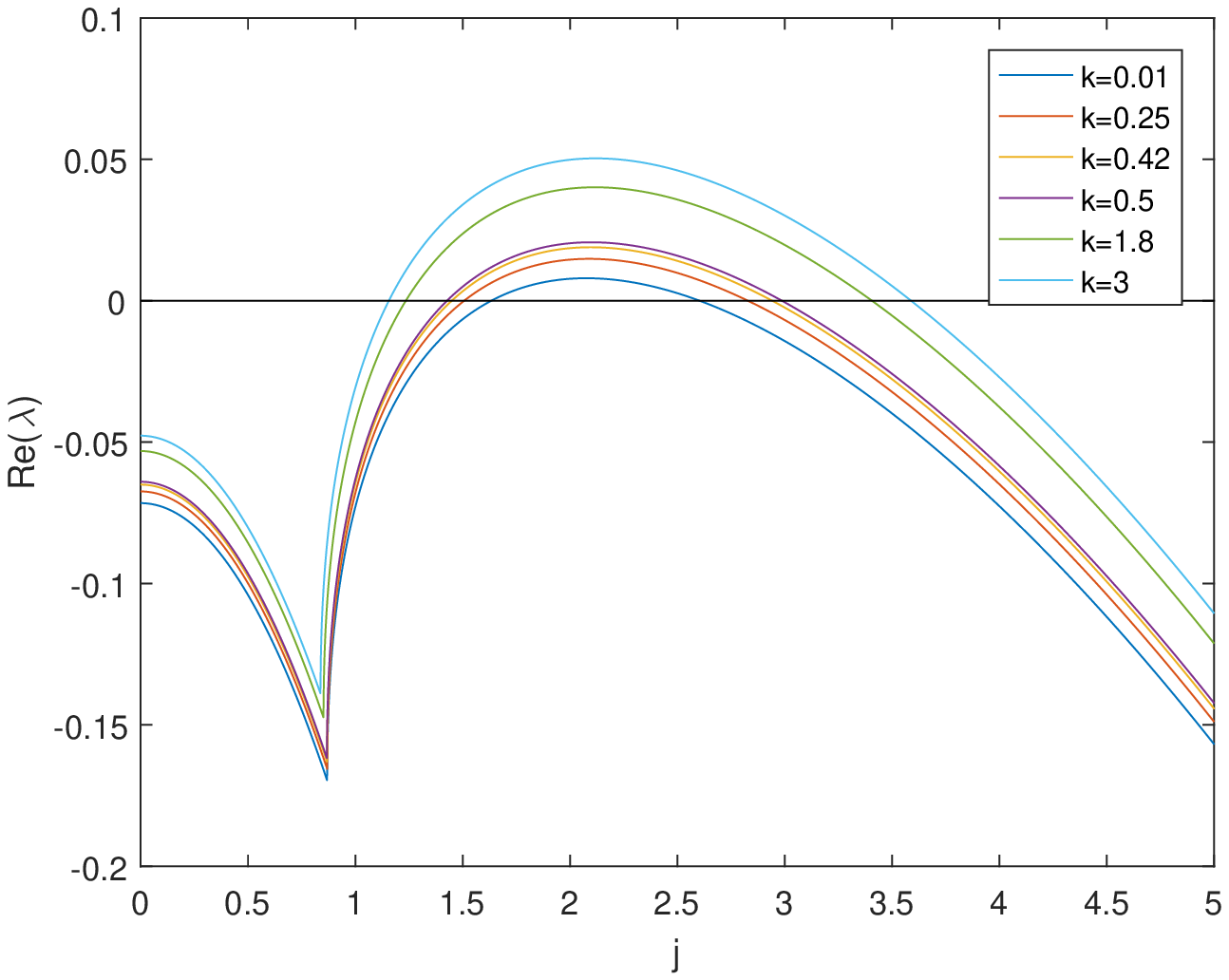}
}
\quad
\subfigure[]{
\includegraphics[width=7cm]{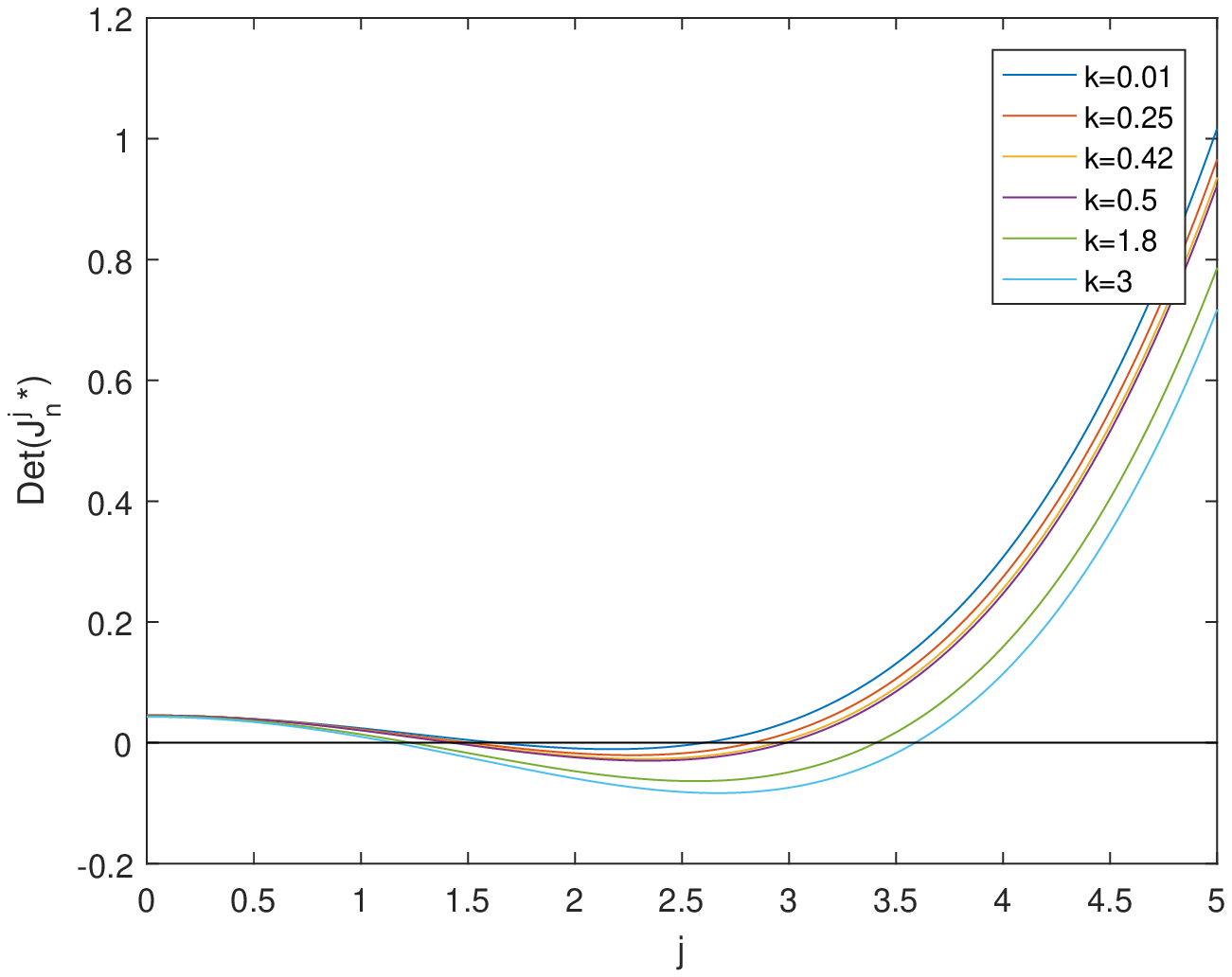}
}

\caption{Plots of (a) the maximum real part of the roots of Eq.~(\ref{3.5}) and (b) $Det(J^j_{n*})$ against $j$ for different $k$ taken from the Turing region. Other parameters are set to $d_{1}=0.01$, $d_{2}=0.25$, $\rho=0.1$ and $\mu=0.55$~\cite{ref1}.}
\label{fig4}
\end{figure}

\begin{figure}
\centering
\subfigure[$k=0.01$, $t=1000$]{
\includegraphics[width=4cm]{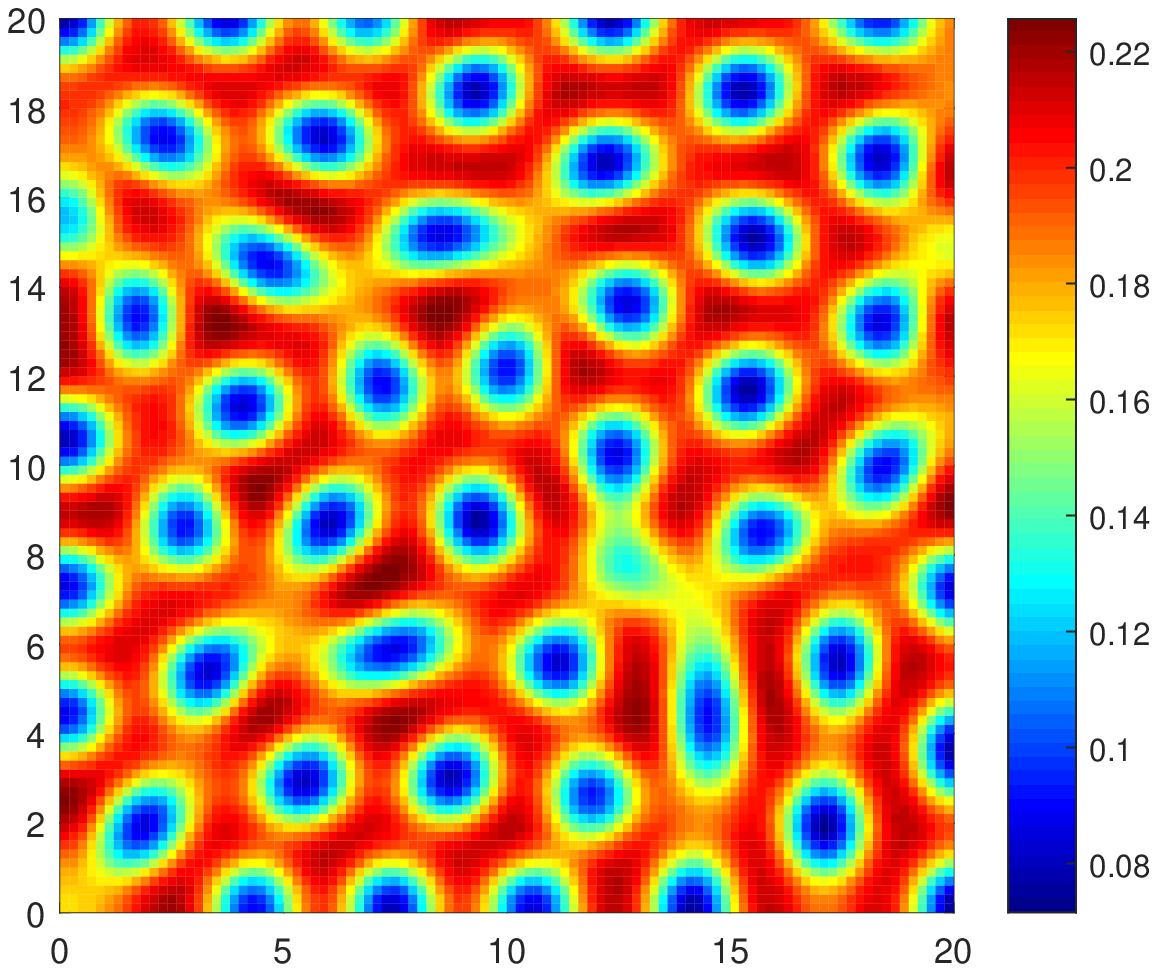}
}
\quad
\subfigure[$k=0.25$, $t=1000$]{
\includegraphics[width=4cm]{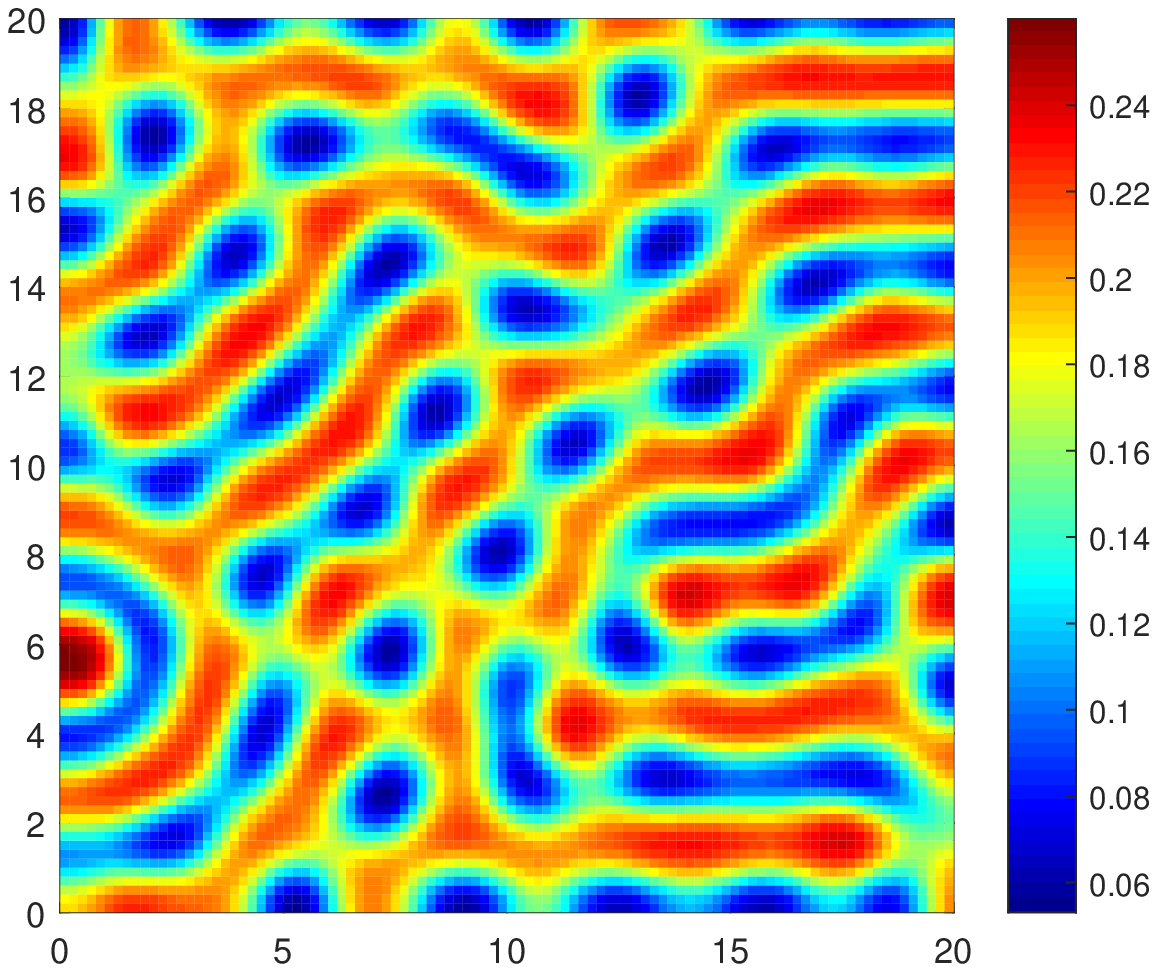}
}
\quad
\subfigure[$k=0.42$, $t=1000$]{
\includegraphics[width=4cm]{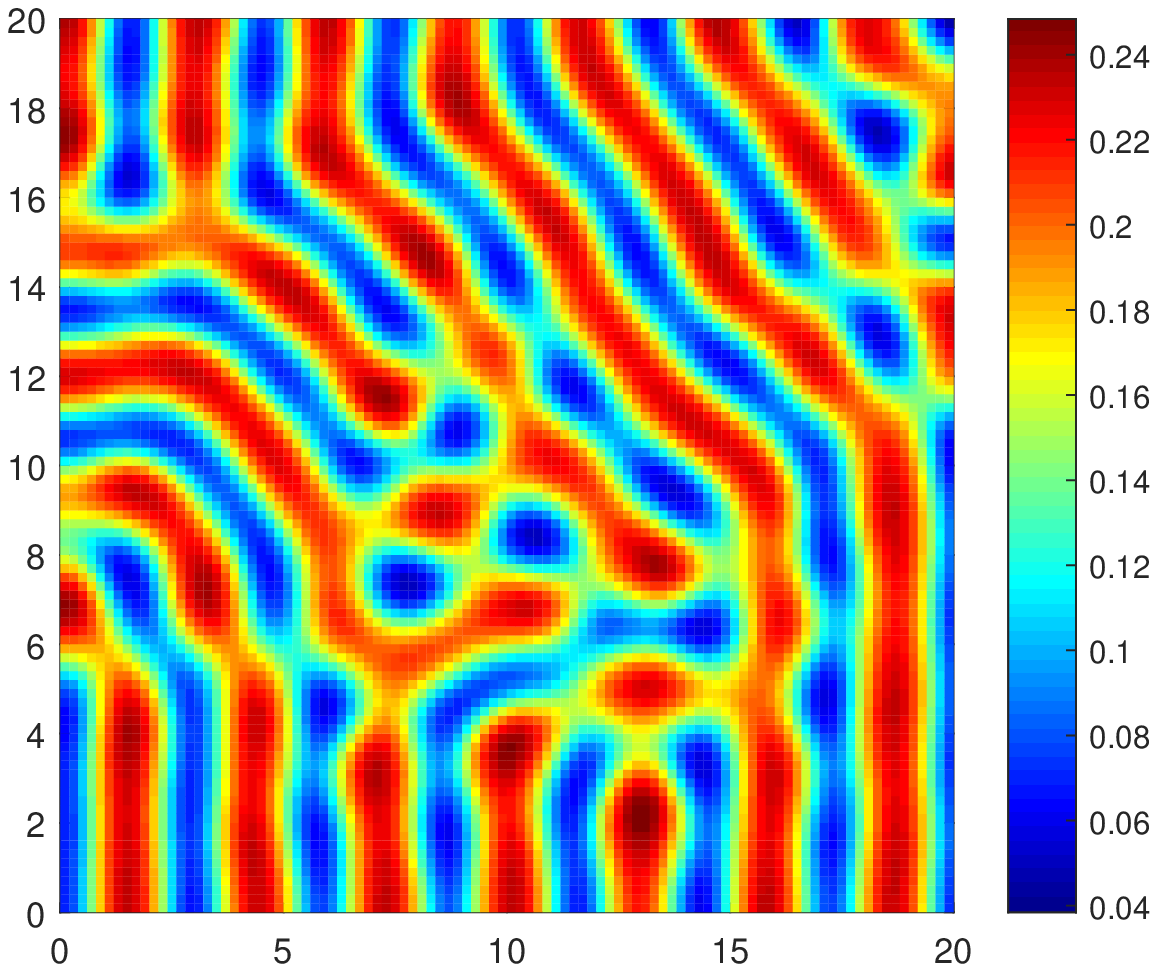}
}
\quad

\subfigure[$k=0.5$, $t=1000$]{
\includegraphics[width=4cm]{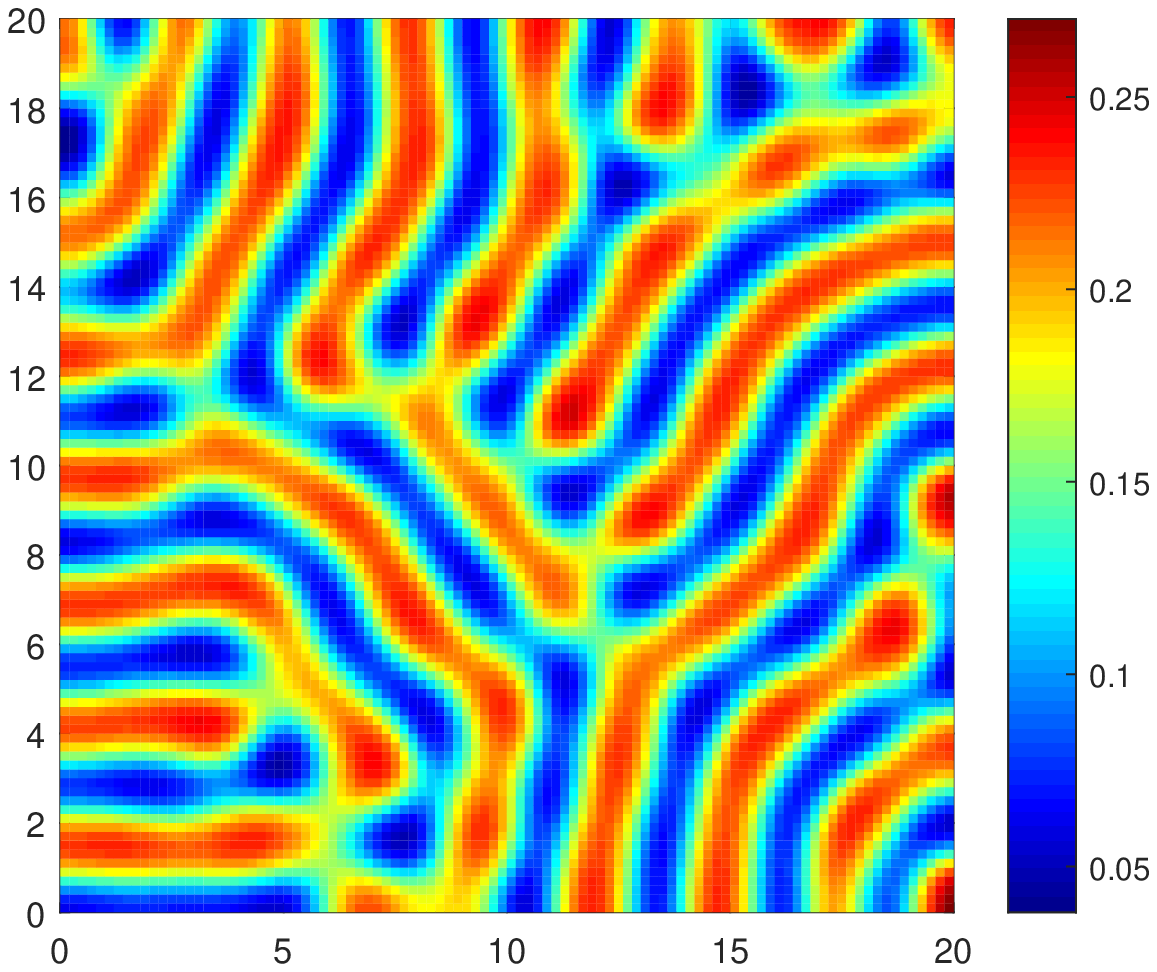}
}
\quad
\subfigure[$k=1.8$, $t=1000$]{
\includegraphics[width=4cm]{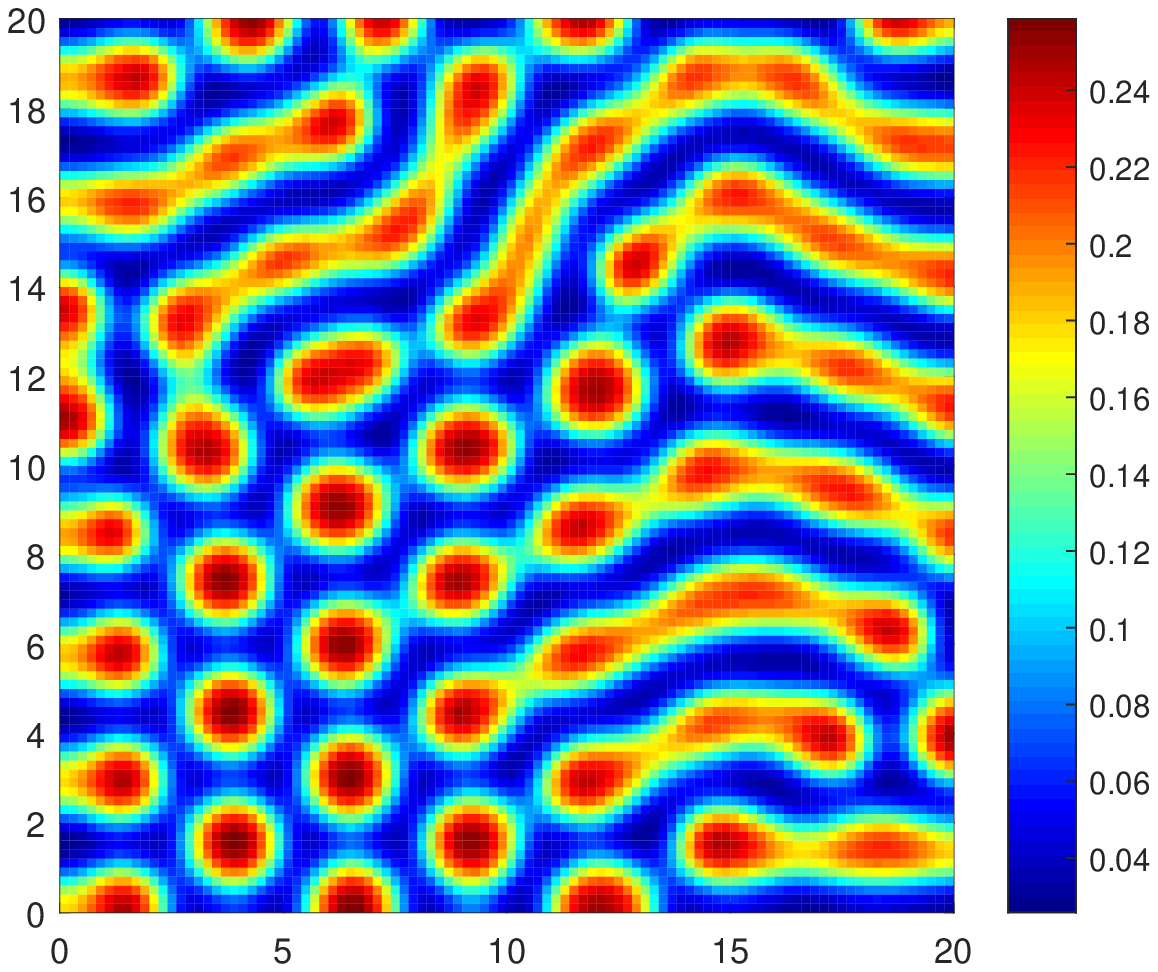}
}
\quad
\subfigure[$k=3$, $t=1000$]{
\includegraphics[width=4cm]{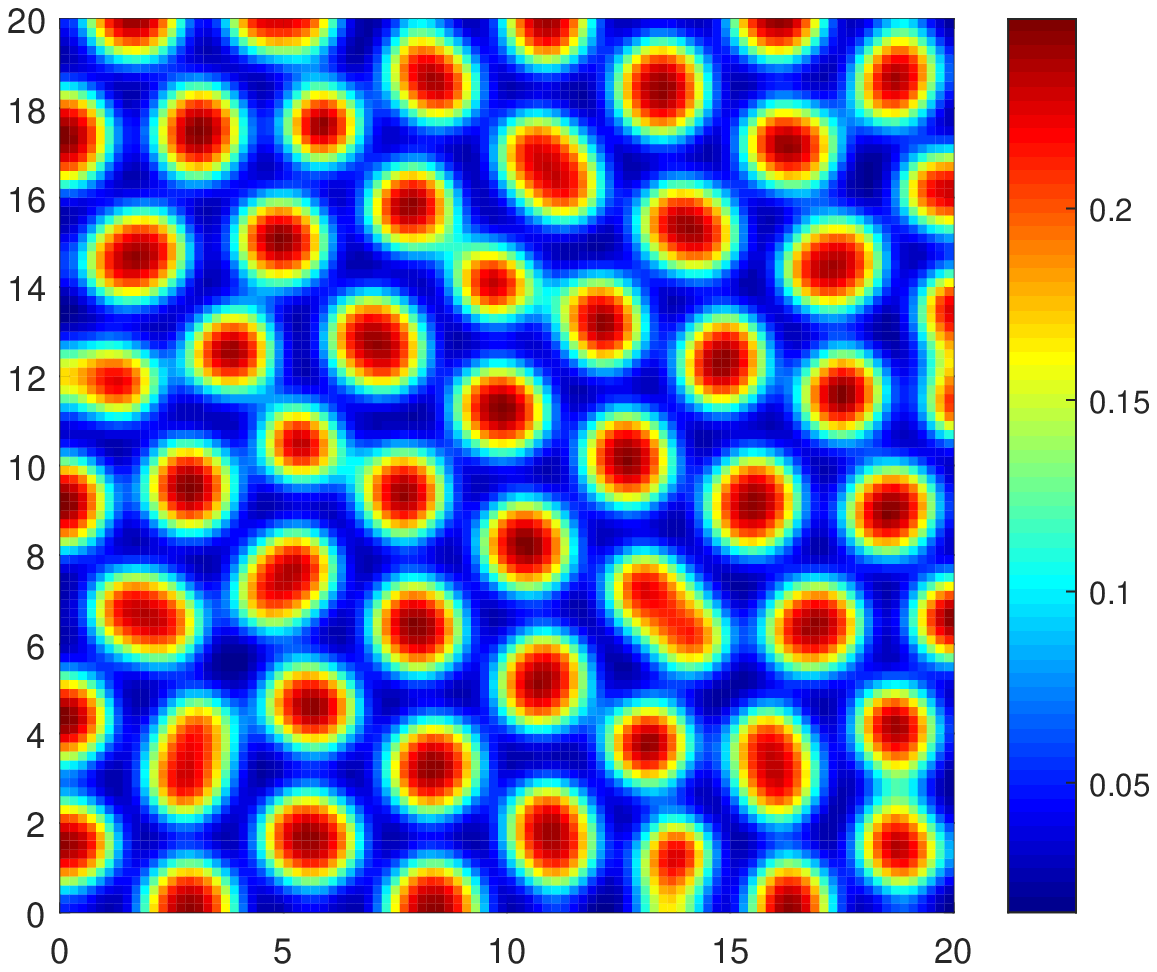}
}
\caption{Fear effect: (a) $k=0.01$, (b) $k=0.25$, (c) $k=0.42$, (d) $k=0.5$, (e) $k=1.8$,  (f) $k=3$. Fear effect controls the growth of pattern: Cold spots $\rightarrow$ cold spots-stripes $\rightarrow$ cold stripes $\rightarrow$ hot stripes $\rightarrow$ hot spots-stripes $\rightarrow$ hot spots for $d_{1}=0.01$, $d_{2}=0.25$, $\rho=0.1$ and $\mu=0.55$. }
\label{fig5}
\end{figure}

\subsection{Pattern formation by reducing reproduction ability of infected hosts }
Then, we plan to observe the effect of reducing reproduction ability of infected hosts on population distribution. From Fig.~\ref{fig6} and Fig.~\ref{fig7}, we set the fear effect is $k=0.01$, diffusion coefficients of $S$ is $d_1=0.01$~\cite{ref1}, diffusion coefficients of $I$ is $d_2=0.25$~\cite{ref1} and reducing reproduction ability of infected hosts $\rho=0.085$ and $\rho=0.06$, other parameters are $\mu=0.55$~\cite{ref1} and shown in Table~\ref{tab2}. Fig.~\ref{fig6} shows that maximum real part of the roots of Eq.~(\ref{3.5}) and $Det(J^j_{n*})$ against $j$ for different $\rho$ taken from the Turing region. Respectively, we observed the change of the pattern shape. The results showed that reducing reproduction ability of infected hosts can control the growth of pattern: cold spots-stripes patten (Fig.~\ref{fig7}(a))$\rightarrow$ cold stripes patten (Fig.~\ref{fig7}(b)).

\begin{figure}
\centering
\subfigure[]{
\includegraphics[width=7cm]{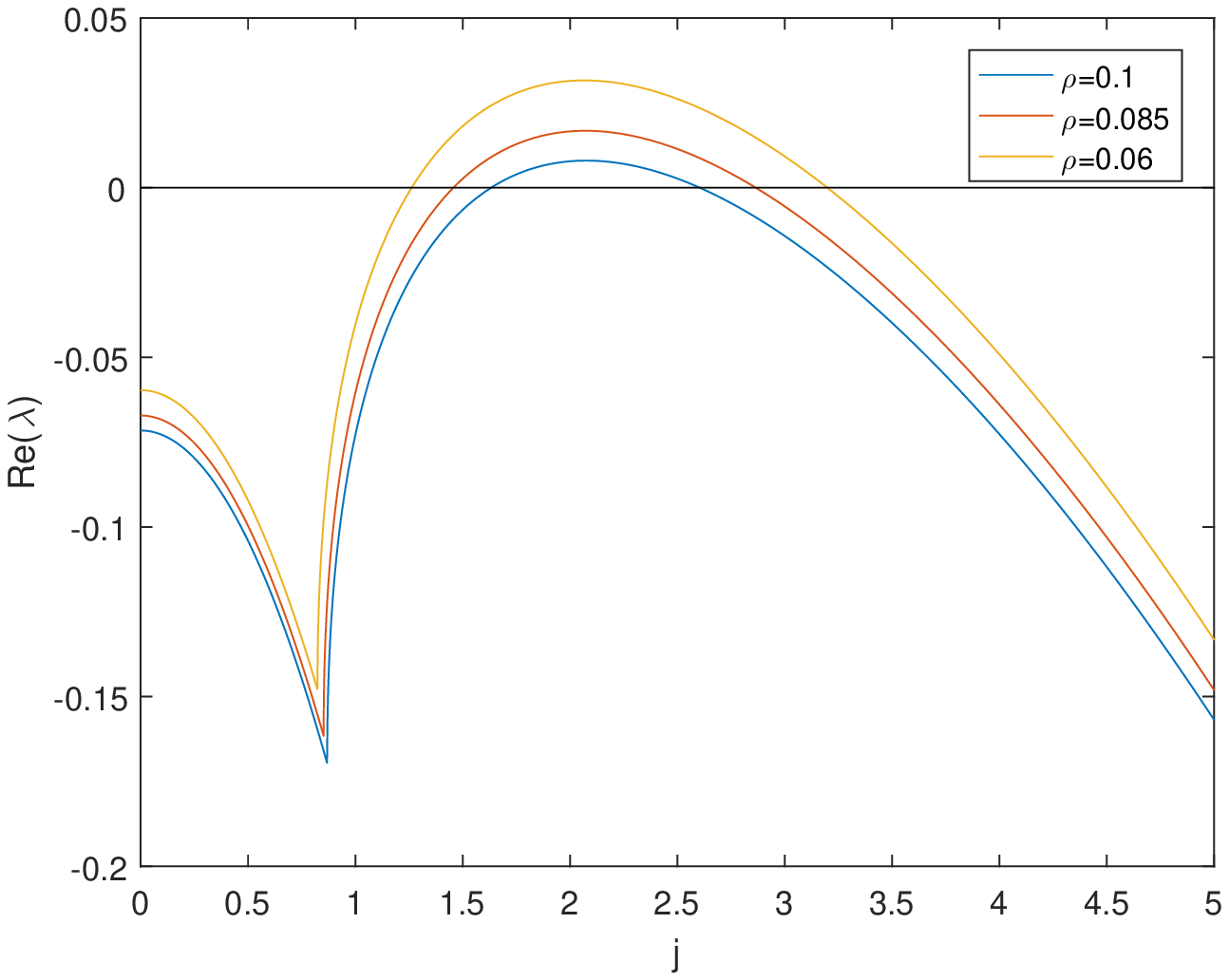}
}
\quad
\subfigure[]{
\includegraphics[width=7cm]{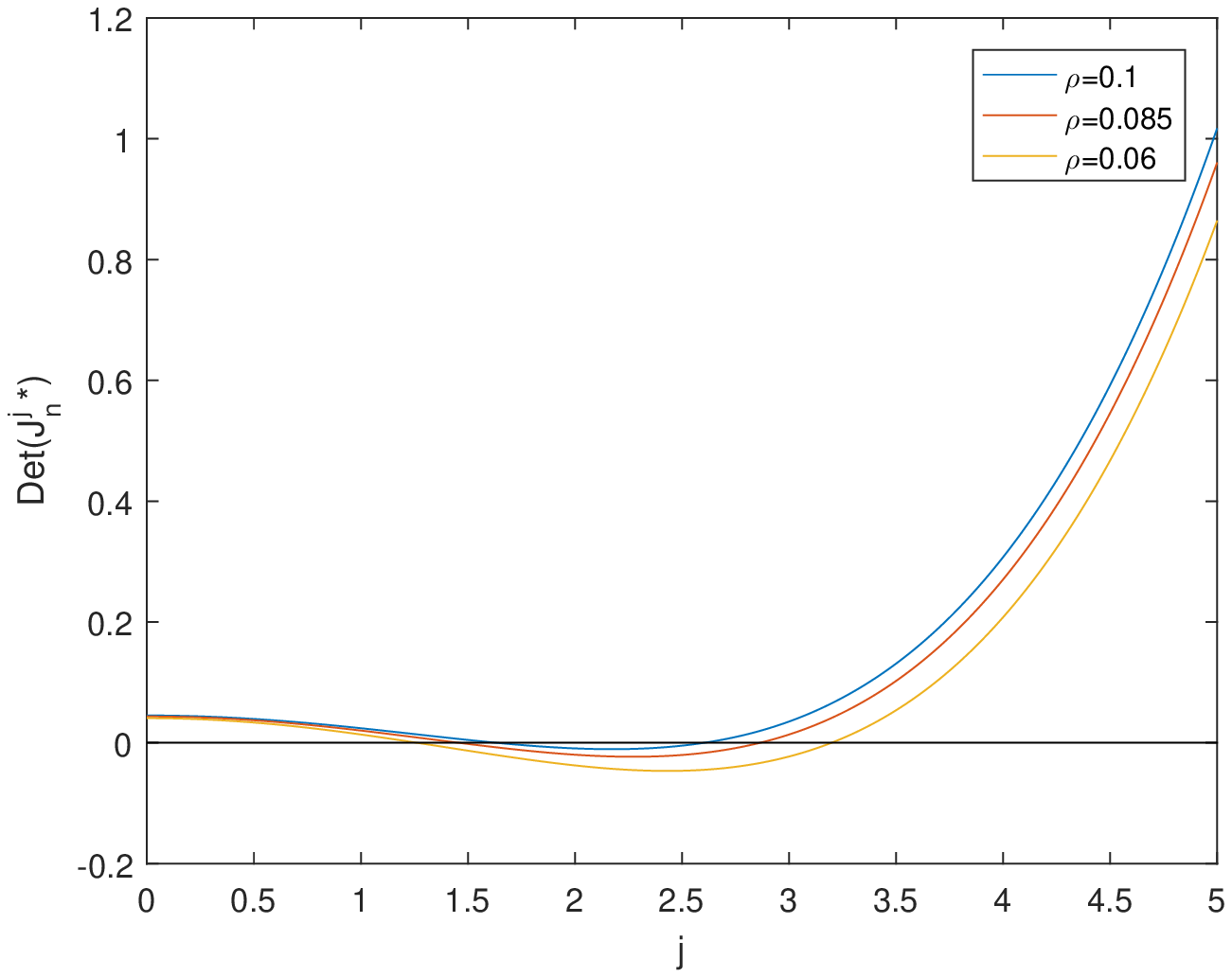}
}

\caption{Plots of (a) the maximum real part of the roots of Eq.~(\ref{3.5}) and (b) $Det(J^j_{n*})$ against $j$ for different $\rho$ taken from the Turing region. Other parameters are set to $d_{1}=0.01$, $d_{2}=0.25$, $\mu=0.55$ and $k=0.01$.}
\label{fig6}
\end{figure}

\begin{figure}
\centering
\subfigure[$\rho=0.085$, $t=1000$]{
\includegraphics[width=7cm]{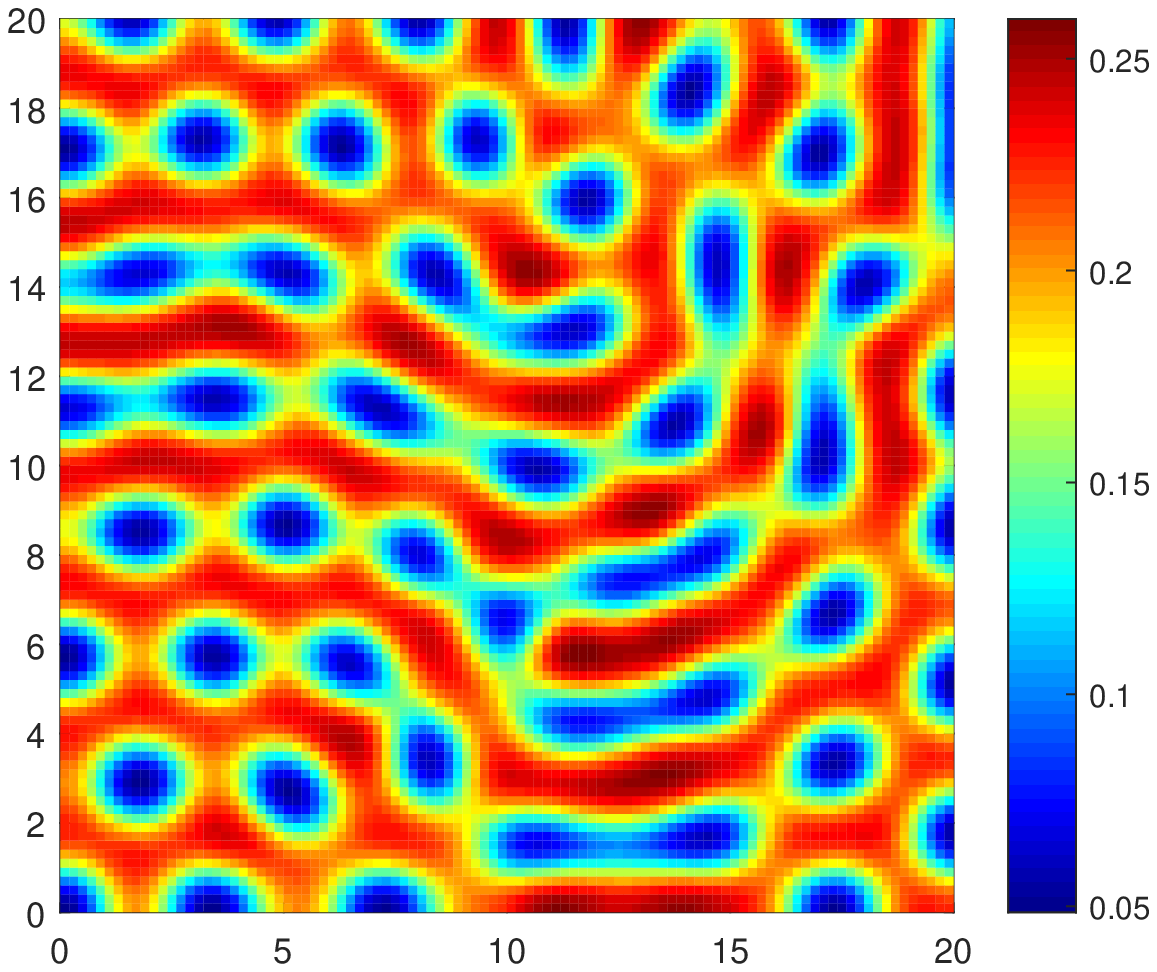}
}
\quad
\subfigure[$\rho=0.06$, $t=1000$]{
\includegraphics[width=7cm]{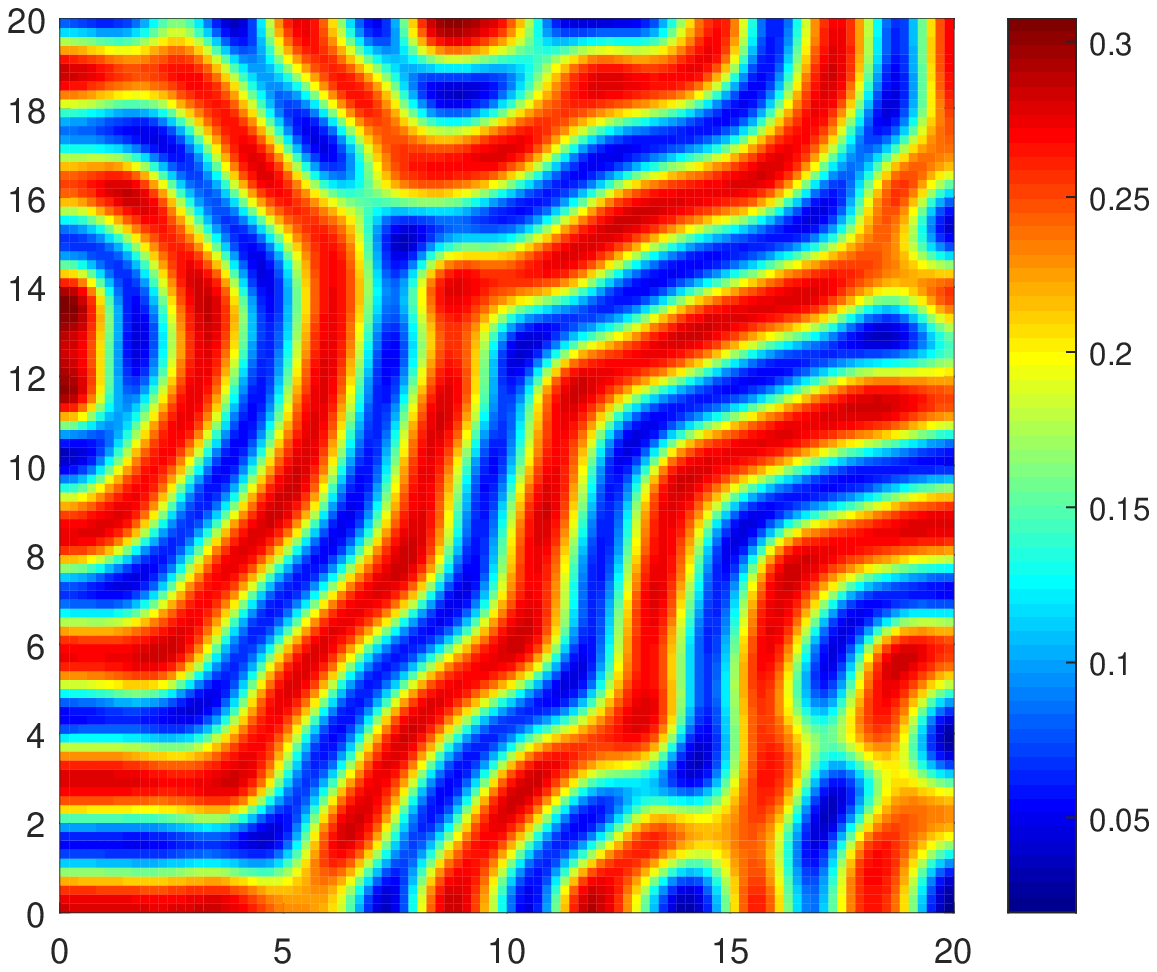}
}
\caption{Reducing reproduction ability of infected hosts: (a) $\rho=0.085$, (b) $\rho=0.06$. Reducing reproduction ability of infected hosts controls the growth of pattern: cold spots-stripes, stripes patten formation for $d_{1}=0.01$, $d_{2}=0.25$, $\mu=0.55$ and $k=0.01$.}
\label{fig7}
\end{figure}

\subsection{Pattern formation by diffusion}
Finally, we give the influence of the self diffusion coefficients of the susceptible and infected on the population distribution. From Fig.~\ref{fig8} and Fig.~\ref{fig9}, we set the fear effect is $k=0.01$, reducing reproduction ability of infected hosts $\rho=0.1$~\cite{ref1}, diffusion coefficients of $S$ is $d_1=0.01$~\cite{ref1} and diffusion coefficients of $I$ to $d_2=0.32$ and $d_2=3$, other parameters are $\mu=0.55$~\cite{ref1} and shown in Table~\ref{tab2}. Fig.~\ref{fig8} shows that maximum real part of the roots of Eq.~(\ref{3.5}) and $Det(J^j_{n*})$ against $j$ for different $d_{2}$ taken from the Turing region. Respectively, we observed the change of the pattern shape. The results showed that self diffusion coefficient of the susceptible can control the growth of pattern: cold spots-stripes patten (Fig.~\ref{fig9}(a))$\rightarrow$ cold stripes patten. (Fig.~\ref{fig9}(b)).
\begin{figure}
\centering
\subfigure[]{
\includegraphics[width=7cm]{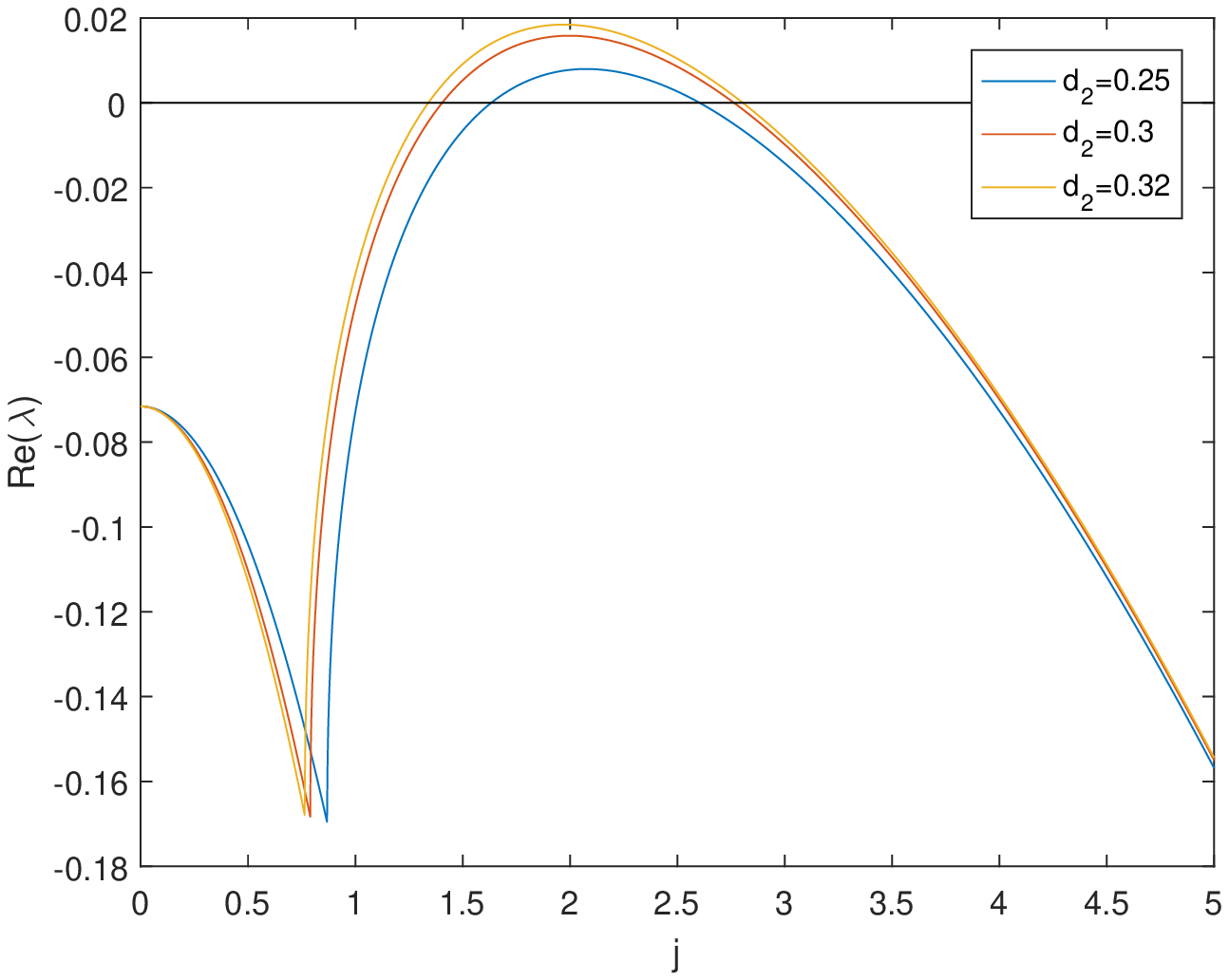}
}
\quad
\subfigure[]{
\includegraphics[width=7cm]{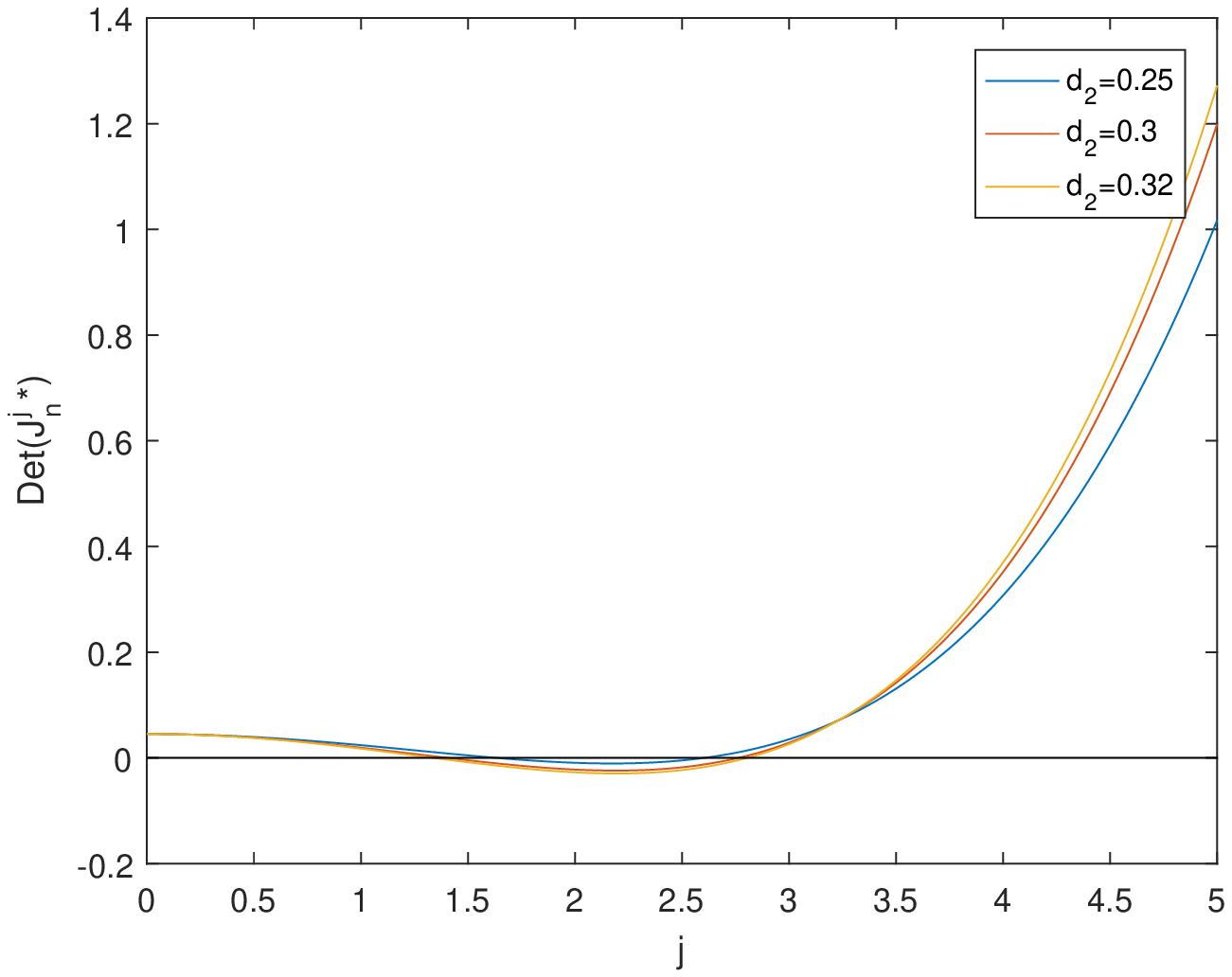}
}

\caption{Plots of (a) the maximum real part of the roots of Eq.~(\ref{3.5}) and (b) $Det(J^j_{n*})$ against $j$ for different $d_{2}$ taken from the Turing region. Other parameters are set to $d_{1}=0.01$, $k=0.01$, $\rho=0.1$ and $\mu=0.55$.}
\label{fig8}
\end{figure}

\begin{figure}
\centering
\subfigure[$d_2=0.32$, $t=1000$]{
\includegraphics[width=7cm]{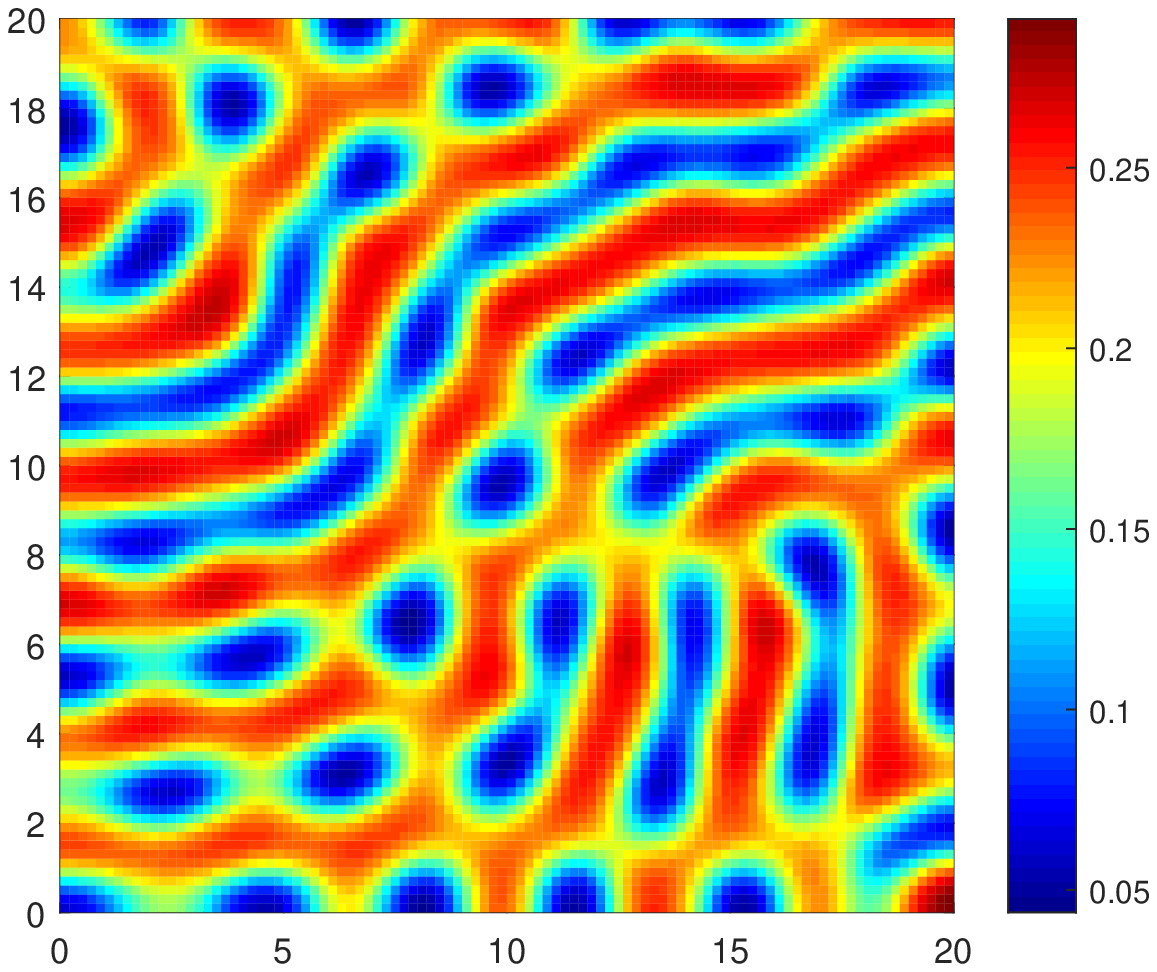}
}
\quad
\subfigure[$d_2=3$, $t=1000$]{
\includegraphics[width=7cm]{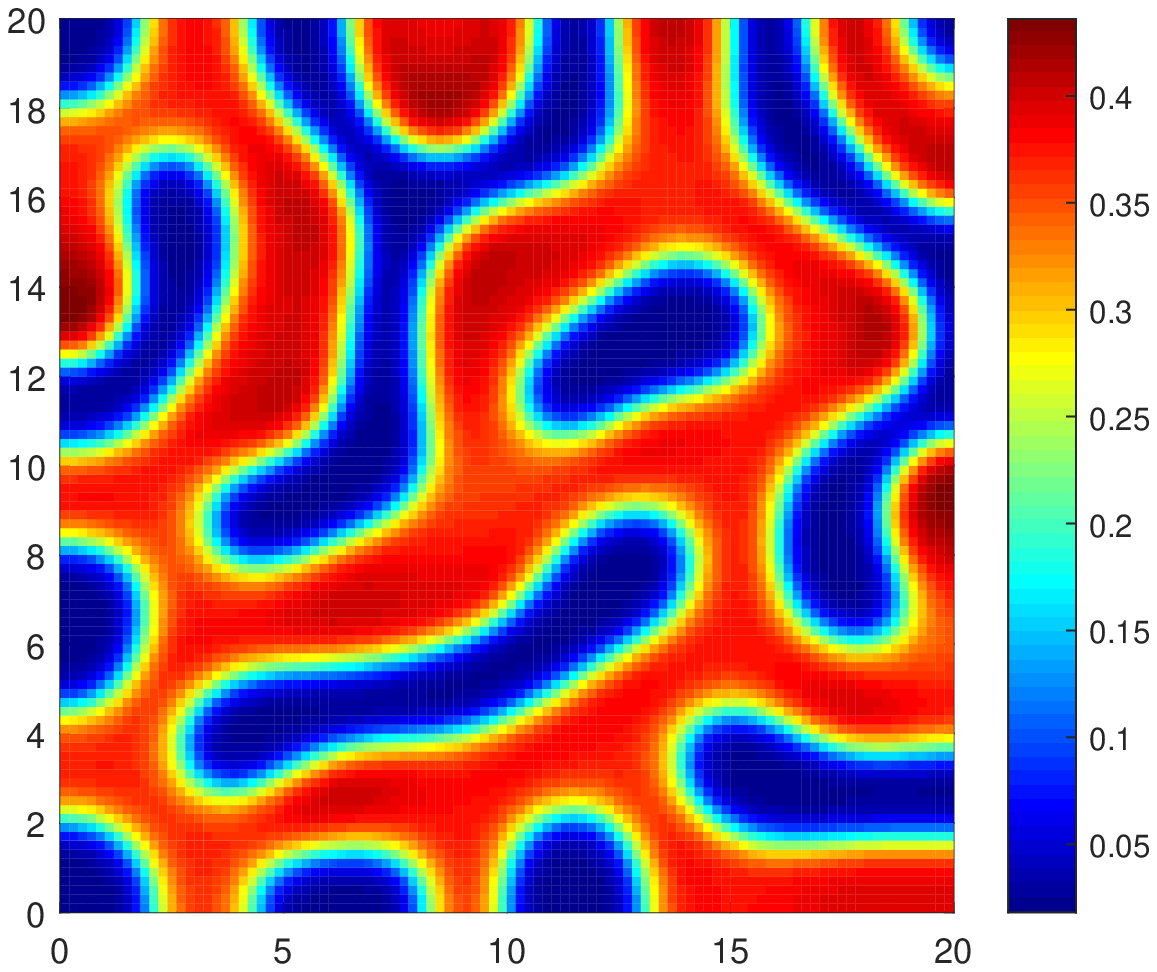}
}
\caption{Diffusion coefficient: (a) $d_2=0.32$, (b) $d_2=3$. Self diffusion coefficient of the susceptible controls the growth of pattern: cold spots-stripes, stripes patten formation for $d_{1}=0.01$, $k=0.01$, $\rho=0.1$ and $\mu=0.55$.}
\label{fig9}
\end{figure}

From Fig.~\ref{fig10} and Fig.~\ref{fig11}, we set the fear effect is $k=0.01$, reducing reproduction ability of infected hosts $\rho=0.1$~\cite{ref1}, diffusion coefficients of $I$ is $d_2=0.25$~\cite{ref1} and diffusion coefficients of $S$ to $d_1=0.008$ and $d_1=0.005$, other parameters are $\mu=0.55$~\cite{ref1} and shown in Table~\ref{tab2}. Fig.~\ref{fig10} shows that maximum real part of the roots of Eq.~(\ref{3.5}) and $Det(J^j_{n*})$ against $j$ for different $d_{1}$ taken from the Turing region. Respectively, we observed the change of the pattern shape. The results showed that self diffusion coefficient of the infected can control the growth of pattern: cold spots-stripes patten (Fig.~\ref{fig11}(a))$\rightarrow$ cold stripes patten (Fig.~\ref{fig11}(b)).
\begin{figure}
\centering
\subfigure[]{
\includegraphics[width=7cm]{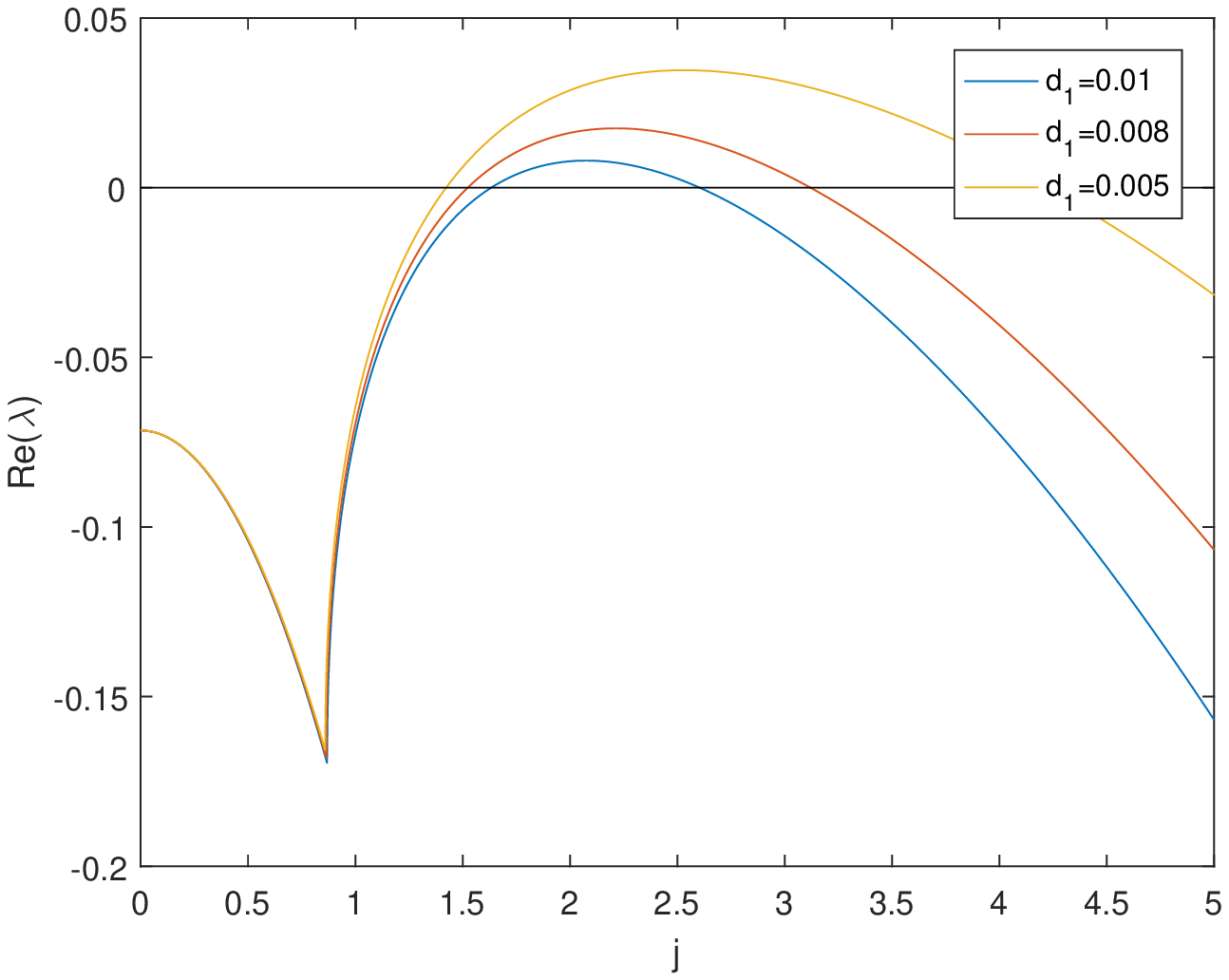}
}
\quad
\subfigure[]{
\includegraphics[width=7cm]{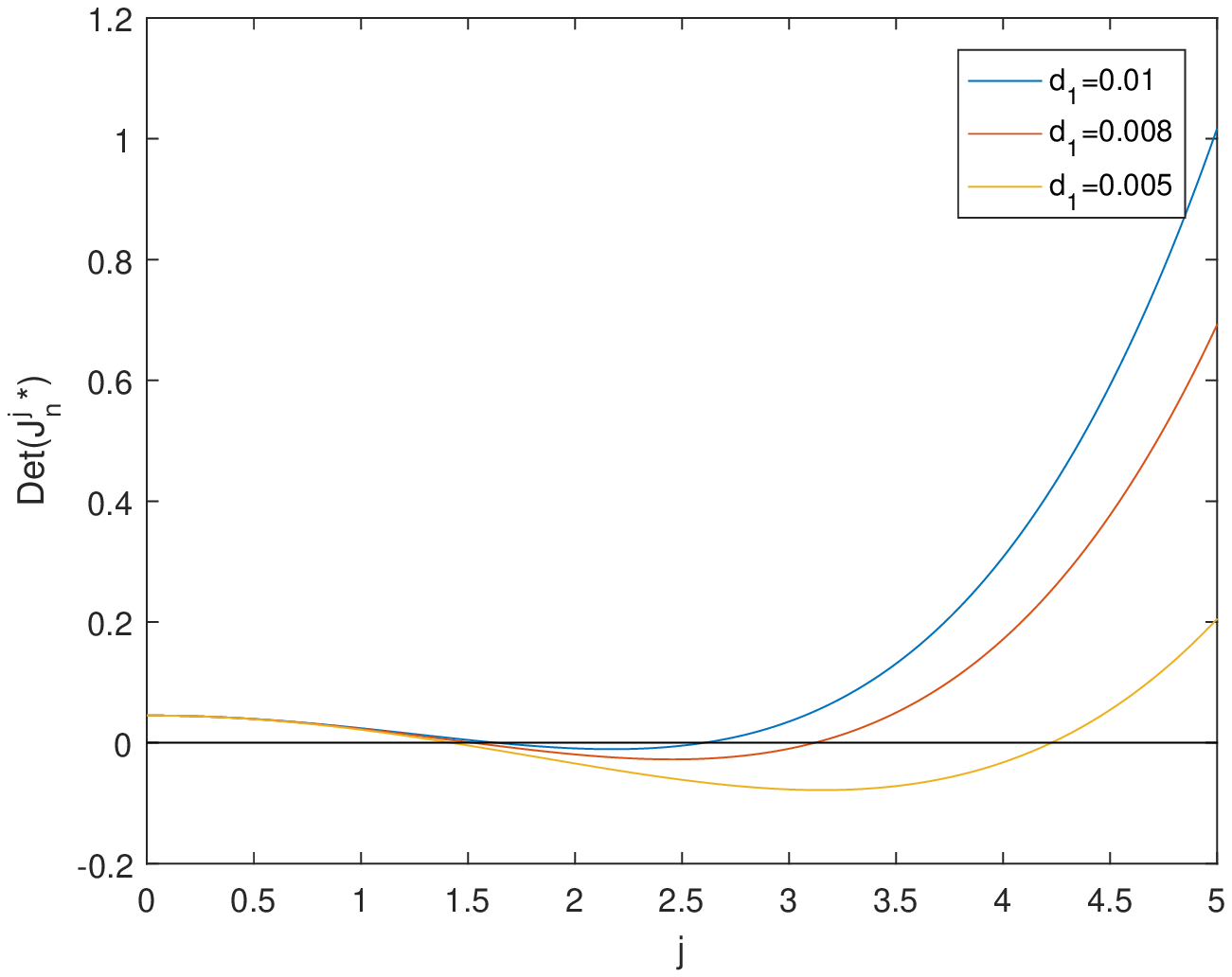}
}

\caption{Plots of (a) the maximum real part of the roots of Eq.~(\ref{3.5}) and (b) $Det(J^j_{n*})$ against $j$ for different $d_{1}$ taken from the Turing region. Other parameters are set to $d_{2}=0.25$, $k=0.01$, $\rho=0.1$ and $\mu=0.55$.}
\label{fig10}
\end{figure}

\begin{figure}
\centering
\subfigure[$d_1=0.008$, $t=1000$]{
\includegraphics[width=7cm]{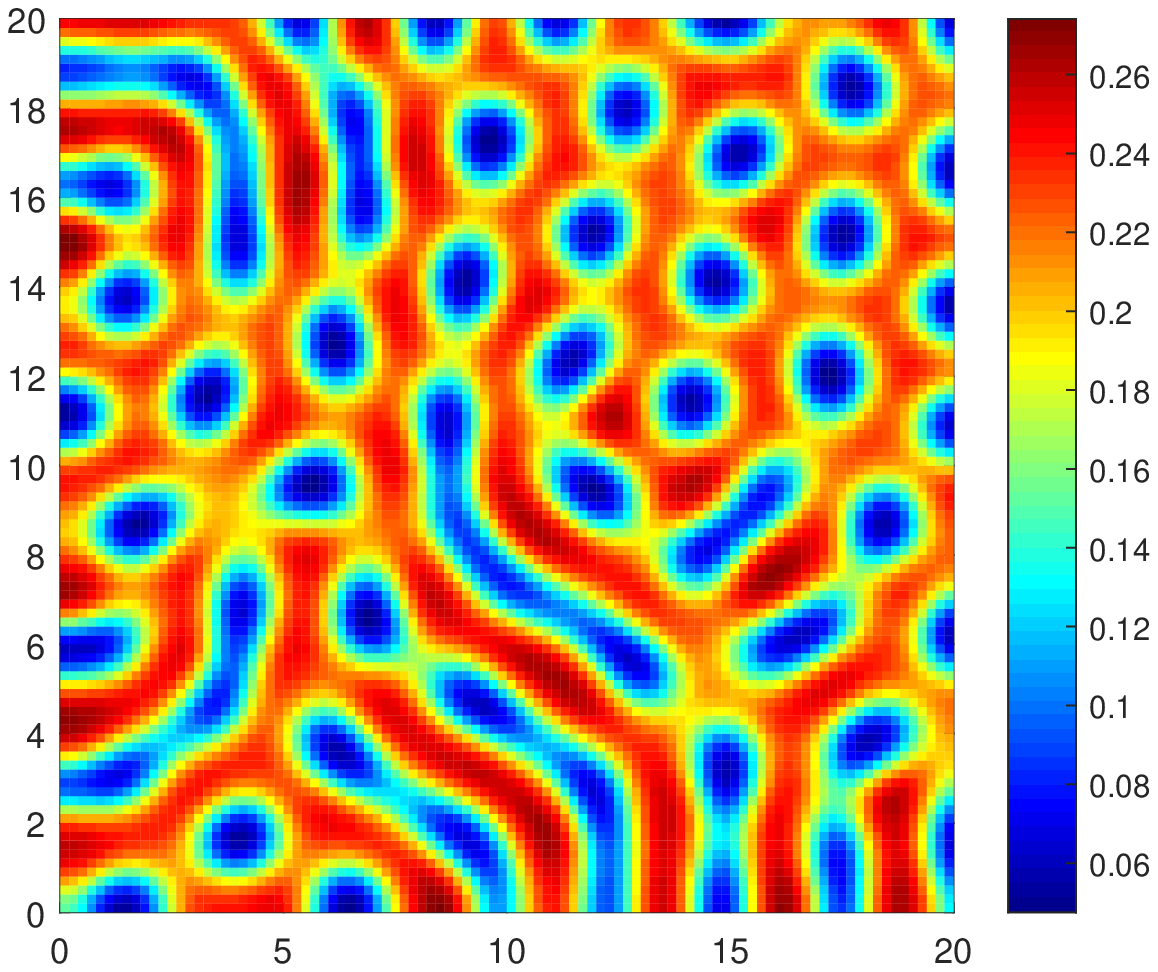}
}
\quad
\subfigure[$d_1=0.005$, $t=1000$]{
\includegraphics[width=7cm]{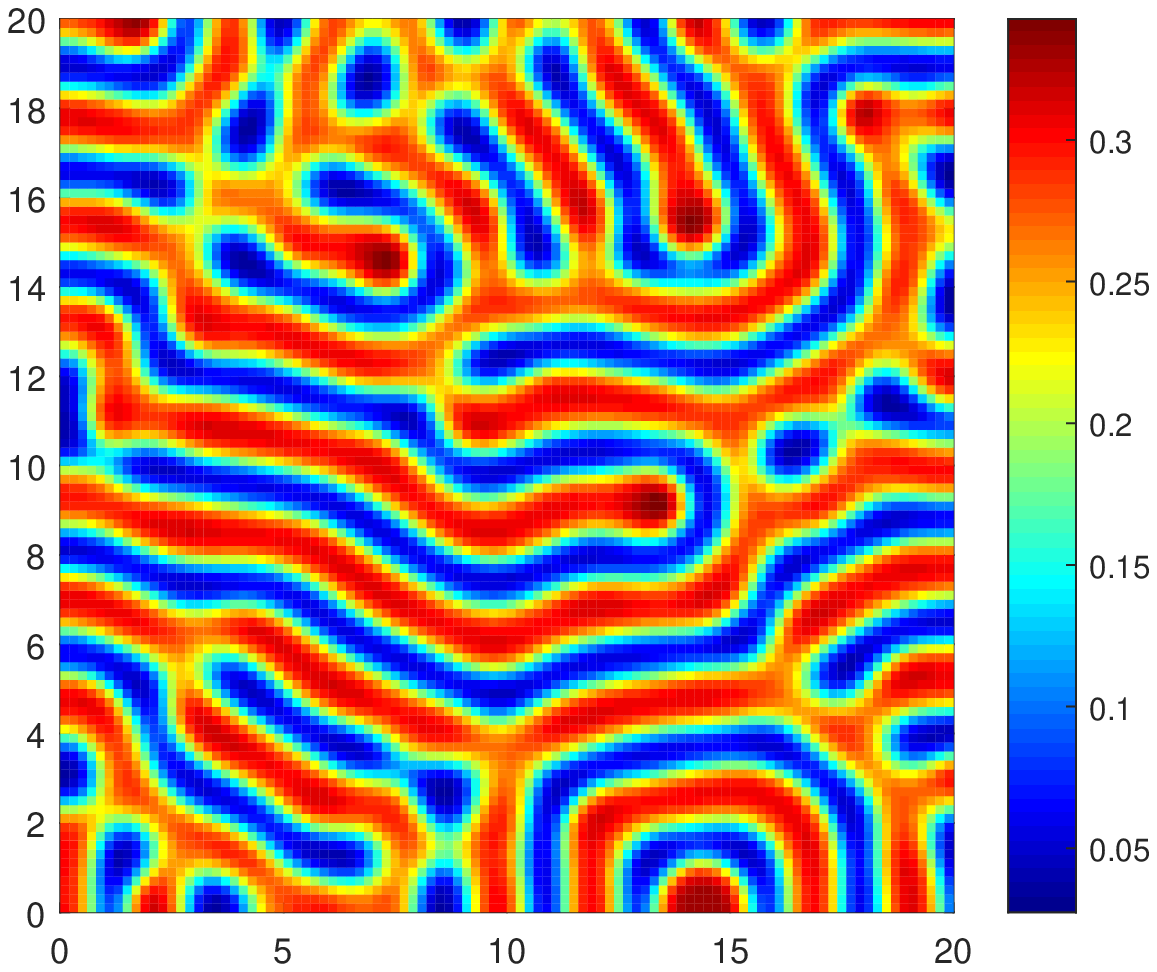}
}
\caption{Diffusion coefficient: (a) $d_1=0.008$, (b) $d_1=0.005$. Self diffusion coefficient of the infected controls the growth of pattern: cold spots-stripes, stripes patten formation for $d_{2}=0.25$, $k=0.01$, $\rho=0.1$ and $\mu=0.55$.}
\label{fig11}
\end{figure}

\section{Conclusions}\label{section5}
In conclusion, we study the pattern formation of a host parasite model induced by fear effect. The conditions of Turing instability are obtained through theoretical analysis. With the help of numerical simulation, the influence of various parameters on the pattern formation is explored by selecting different control parameters (natural mortality $\mu$, fear effect $k$, reducing reproduction ability $\rho$, diffusion coefficient $d_1, d_2$). Through the observation of the pattern growth, we find some interesting phenomena. The results show that under the effect of fear, the pattern growth mechanism is always unified, growing in the form of cold spots $\rightarrow$ cold spots-stripes $\rightarrow$ cold stripes $\rightarrow$ hot stripes $\rightarrow$ hot spots-stripes $\rightarrow$ hot spots, which is also consistent with the results of literature~\cite{ref5}. Moreover, we also find that the change of pattern growth caused by natural mortality was opposite to that of fear factor.

The dynamic phenomena of ODE predator-prey system considering fear effect include: promoting stability, periodic solution (limit cycle), making chaotic state become stable and so on~\cite{ref2,ref4,ref6,ref7,ref10,ref28,ref29,ref30,ref31,ref32,ref33}. However, researches on the change of spatial population dynamics caused by fear effect are still not general. Through the description of~\cite{ref5,ref35}, we can find that the evolution of pattern induced by fear effect often follows the rule that with the increase of fear, the growth order of pattern is: cold spots $\rightarrow$ cold spots-stripes $\rightarrow$ cold stripes $\rightarrow$ hot stripes $\rightarrow$ hot spots-stripes $\rightarrow$ hot spots. As for whether such a conclusion is general, we need to continue to explore in our future work.
\begin{acknowledgments}
This work was supported by the Innovative Research Project of Shenzhen under Project No. KQJSCX20180328165509766, Nature Science Foundation of Guangdong Province under Project No. 2020A1515010812 and 2021A1515011594.
\end{acknowledgments}

\section*{DATA AVAILABILITY}
Data sharing is not applicable to this article as no new data were created or analyzed in this study.
\small
\bibliography{aipsamp1}

\providecommand{\noopsort}[1]{}\providecommand{\singleletter}[1]{#1}%
\begin{thebibliography}{37}%
\makeatletter
\providecommand \@ifxundefined [1]{%
 \@ifx{#1\undefined}
}%
\providecommand \@ifnum [1]{%
 \ifnum #1\expandafter \@firstoftwo
 \else \expandafter \@secondoftwo
 \fi
}%
\providecommand \@ifx [1]{%
 \ifx #1\expandafter \@firstoftwo
 \else \expandafter \@secondoftwo
 \fi
}%
\providecommand \natexlab [1]{#1}%
\providecommand \enquote  [1]{``#1''}%
\providecommand \bibnamefont  [1]{#1}%
\providecommand \bibfnamefont [1]{#1}%
\providecommand \citenamefont [1]{#1}%
\providecommand \href@noop [0]{\@secondoftwo}%
\providecommand \href [0]{\begingroup \@sanitize@url \@href}%
\providecommand \@href[1]{\@@startlink{#1}\@@href}%
\providecommand \@@href[1]{\endgroup#1\@@endlink}%
\providecommand \@sanitize@url [0]{\catcode `\\12\catcode `\$12\catcode
  `\&12\catcode `\#12\catcode `\^12\catcode `\_12\catcode `\%12\relax}%
\providecommand \@@startlink[1]{}%
\providecommand \@@endlink[0]{}%
\providecommand \url  [0]{\begingroup\@sanitize@url \@url }%
\providecommand \@url [1]{\endgroup\@href {#1}{\urlprefix }}%
\providecommand \urlprefix  [0]{URL }%
\providecommand \Eprint [0]{\href }%
\providecommand \doibase [0]{http://dx.doi.org/}%
\providecommand \selectlanguage [0]{\@gobble}%
\providecommand \bibinfo  [0]{\@secondoftwo}%
\providecommand \bibfield  [0]{\@secondoftwo}%
\providecommand \translation [1]{[#1]}%
\providecommand \BibitemOpen [0]{}%
\providecommand \bibitemStop [0]{}%
\providecommand \bibitemNoStop [0]{.\EOS\space}%
\providecommand \EOS [0]{\spacefactor3000\relax}%
\providecommand \BibitemShut  [1]{\csname bibitem#1\endcsname}%
\let\auto@bib@innerbib\@empty
\bibitem [{\citenamefont {Wang}, \citenamefont {Zanette},\ and\ \citenamefont
  {Zou}(2016)}]{ref2}%
  \BibitemOpen
  \bibfield  {author} {\bibinfo {author} {\bibfnamefont {X.}~\bibnamefont
  {Wang}}, \bibinfo {author} {\bibfnamefont {L.}~\bibnamefont {Zanette}}, \
  and\ \bibinfo {author} {\bibfnamefont {X.}~\bibnamefont {Zou}},\ }\bibfield
  {title} {\enquote {\bibinfo {title} {Modelling the fear effect in
  predator--prey interactions},}\ }\href@noop {} {\bibfield  {journal}
  {\bibinfo  {journal} {Journal of mathematical biology}\ }\textbf {\bibinfo
  {volume} {73}},\ \bibinfo {pages} {1179--1204} (\bibinfo {year}
  {2016})}\BibitemShut {NoStop}%
\bibitem [{\citenamefont {Hwang}\ and\ \citenamefont {Kuang}(2003)}]{ref3}%
  \BibitemOpen
  \bibfield  {author} {\bibinfo {author} {\bibfnamefont {T.-W.}\ \bibnamefont
  {Hwang}}\ and\ \bibinfo {author} {\bibfnamefont {Y.}~\bibnamefont {Kuang}},\
  }\bibfield  {title} {\enquote {\bibinfo {title} {Deterministic extinction
  effect of parasites on host populations},}\ }\href@noop {} {\bibfield
  {journal} {\bibinfo  {journal} {Journal of mathematical biology}\ }\textbf
  {\bibinfo {volume} {46}},\ \bibinfo {pages} {17--30} (\bibinfo {year}
  {2003})}\BibitemShut {NoStop}%
\bibitem [{\citenamefont {Ebert}, \citenamefont {Lipsitch},\ and\ \citenamefont
  {Mangin}(2000)}]{ref9}%
  \BibitemOpen
  \bibfield  {author} {\bibinfo {author} {\bibfnamefont {D.}~\bibnamefont
  {Ebert}}, \bibinfo {author} {\bibfnamefont {M.}~\bibnamefont {Lipsitch}}, \
  and\ \bibinfo {author} {\bibfnamefont {K.~L.}\ \bibnamefont {Mangin}},\
  }\bibfield  {title} {\enquote {\bibinfo {title} {The effect of parasites on
  host population density and extinction: experimental epidemiology with
  daphnia and six microparasites},}\ }\href@noop {} {\bibfield  {journal}
  {\bibinfo  {journal} {The American Naturalist}\ }\textbf {\bibinfo {volume}
  {156}},\ \bibinfo {pages} {459--477} (\bibinfo {year} {2000})}\BibitemShut
  {NoStop}%
\bibitem [{\citenamefont {Hethcote}(2000)}]{ref13}%
  \BibitemOpen
  \bibfield  {author} {\bibinfo {author} {\bibfnamefont {H.~W.}\ \bibnamefont
  {Hethcote}},\ }\bibfield  {title} {\enquote {\bibinfo {title} {The
  mathematics of infectious diseases},}\ }\href@noop {} {\bibfield  {journal}
  {\bibinfo  {journal} {SIAM review}\ }\textbf {\bibinfo {volume} {42}},\
  \bibinfo {pages} {599--653} (\bibinfo {year} {2000})}\BibitemShut {NoStop}%
\bibitem [{\citenamefont {Hwang}\ and\ \citenamefont {Kuang}(2005)}]{ref14}%
  \BibitemOpen
  \bibfield  {author} {\bibinfo {author} {\bibfnamefont {T.-W.}\ \bibnamefont
  {Hwang}}\ and\ \bibinfo {author} {\bibfnamefont {Y.}~\bibnamefont {Kuang}},\
  }\bibfield  {title} {\enquote {\bibinfo {title} {Host extinction dynamics in
  a simple parasite-host interaction model},}\ }\href@noop {} {\bibfield
  {journal} {\bibinfo  {journal} {Mathematical Biosciences \& Engineering}\
  }\textbf {\bibinfo {volume} {2}},\ \bibinfo {pages} {743} (\bibinfo {year}
  {2005})}\BibitemShut {NoStop}%
\bibitem [{\citenamefont {Holmes}\ \emph {et~al.}(1994)\citenamefont {Holmes},
  \citenamefont {Lewis}, \citenamefont {Banks},\ and\ \citenamefont
  {Veit}}]{ref15}%
  \BibitemOpen
  \bibfield  {author} {\bibinfo {author} {\bibfnamefont {E.~E.}\ \bibnamefont
  {Holmes}}, \bibinfo {author} {\bibfnamefont {M.~A.}\ \bibnamefont {Lewis}},
  \bibinfo {author} {\bibfnamefont {J.}~\bibnamefont {Banks}}, \ and\ \bibinfo
  {author} {\bibfnamefont {R.}~\bibnamefont {Veit}},\ }\bibfield  {title}
  {\enquote {\bibinfo {title} {Partial differential equations in ecology:
  spatial interactions and population dynamics},}\ }\href@noop {} {\bibfield
  {journal} {\bibinfo  {journal} {Ecology}\ }\textbf {\bibinfo {volume} {75}},\
  \bibinfo {pages} {17--29} (\bibinfo {year} {1994})}\BibitemShut {NoStop}%
\bibitem [{\citenamefont {Neuhauser}(2001)}]{ref16}%
  \BibitemOpen
  \bibfield  {author} {\bibinfo {author} {\bibfnamefont {C.}~\bibnamefont
  {Neuhauser}},\ }\bibfield  {title} {\enquote {\bibinfo {title} {Mathematical
  challenges in spatial ecology},}\ }\href@noop {} {\bibfield  {journal}
  {\bibinfo  {journal} {Notices of the AMS}\ }\textbf {\bibinfo {volume}
  {48}},\ \bibinfo {pages} {1304--1314} (\bibinfo {year} {2001})}\BibitemShut
  {NoStop}%
\bibitem [{\citenamefont {Okubo}\ and\ \citenamefont {Levin}(2001)}]{ref17}%
  \BibitemOpen
  \bibfield  {author} {\bibinfo {author} {\bibfnamefont {A.}~\bibnamefont
  {Okubo}}\ and\ \bibinfo {author} {\bibfnamefont {S.~A.}\ \bibnamefont
  {Levin}},\ }\href@noop {} {\emph {\bibinfo {title} {Diffusion and ecological
  problems: modern perspectives}}},\ Vol.~\bibinfo {volume} {14}\ (\bibinfo
  {publisher} {Springer},\ \bibinfo {year} {2001})\BibitemShut {NoStop}%
\bibitem [{\citenamefont {Wang}\ \emph {et~al.}(2018)\citenamefont {Wang},
  \citenamefont {Gao}, \citenamefont {Cai}, \citenamefont {Shi},\ and\
  \citenamefont {Fu}}]{ref18}%
  \BibitemOpen
  \bibfield  {author} {\bibinfo {author} {\bibfnamefont {W.}~\bibnamefont
  {Wang}}, \bibinfo {author} {\bibfnamefont {X.}~\bibnamefont {Gao}}, \bibinfo
  {author} {\bibfnamefont {Y.}~\bibnamefont {Cai}}, \bibinfo {author}
  {\bibfnamefont {H.}~\bibnamefont {Shi}}, \ and\ \bibinfo {author}
  {\bibfnamefont {S.}~\bibnamefont {Fu}},\ }\bibfield  {title} {\enquote
  {\bibinfo {title} {Turing patterns in a diffusive epidemic model with
  saturated infection force},}\ }\href@noop {} {\bibfield  {journal} {\bibinfo
  {journal} {Journal of the Franklin Institute}\ }\textbf {\bibinfo {volume}
  {355}},\ \bibinfo {pages} {7226--7245} (\bibinfo {year} {2018})}\BibitemShut
  {NoStop}%
\bibitem [{\citenamefont {Zhang}\ \emph
  {et~al.}(2019{\natexlab{a}})\citenamefont {Zhang}, \citenamefont {Cai},
  \citenamefont {Wang},\ and\ \citenamefont {Wang}}]{ref1}%
  \BibitemOpen
  \bibfield  {author} {\bibinfo {author} {\bibfnamefont {B.}~\bibnamefont
  {Zhang}}, \bibinfo {author} {\bibfnamefont {Y.}~\bibnamefont {Cai}}, \bibinfo
  {author} {\bibfnamefont {B.}~\bibnamefont {Wang}}, \ and\ \bibinfo {author}
  {\bibfnamefont {W.}~\bibnamefont {Wang}},\ }\bibfield  {title} {\enquote
  {\bibinfo {title} {Pattern formation in a reaction--diffusion parasite--host
  model},}\ }\href@noop {} {\bibfield  {journal} {\bibinfo  {journal} {Physica
  A: Statistical Mechanics and its Applications}\ }\textbf {\bibinfo {volume}
  {525}},\ \bibinfo {pages} {732--740} (\bibinfo {year}
  {2019}{\natexlab{a}})}\BibitemShut {NoStop}%
\bibitem [{\citenamefont {Cai}\ and\ \citenamefont {Wang}(2015)}]{ref8}%
  \BibitemOpen
  \bibfield  {author} {\bibinfo {author} {\bibfnamefont {Y.}~\bibnamefont
  {Cai}}\ and\ \bibinfo {author} {\bibfnamefont {W.}~\bibnamefont {Wang}},\
  }\bibfield  {title} {\enquote {\bibinfo {title} {Dynamics of a parasite-host
  epidemiological model in spatial heterogeneous environment},}\ }\href@noop {}
  {\bibfield  {journal} {\bibinfo  {journal} {Discrete \& Continuous Dynamical
  Systems-B}\ }\textbf {\bibinfo {volume} {20}},\ \bibinfo {pages} {989}
  (\bibinfo {year} {2015})}\BibitemShut {NoStop}%
\bibitem [{\citenamefont {Zanette}\ \emph {et~al.}(2011)\citenamefont
  {Zanette}, \citenamefont {White}, \citenamefont {Allen},\ and\ \citenamefont
  {Clinchy}}]{ref11}%
  \BibitemOpen
  \bibfield  {author} {\bibinfo {author} {\bibfnamefont {L.~Y.}\ \bibnamefont
  {Zanette}}, \bibinfo {author} {\bibfnamefont {A.~F.}\ \bibnamefont {White}},
  \bibinfo {author} {\bibfnamefont {M.~C.}\ \bibnamefont {Allen}}, \ and\
  \bibinfo {author} {\bibfnamefont {M.}~\bibnamefont {Clinchy}},\ }\bibfield
  {title} {\enquote {\bibinfo {title} {Perceived predation risk reduces the
  number of offspring songbirds produce per year},}\ }\href@noop {} {\bibfield
  {journal} {\bibinfo  {journal} {Science}\ }\textbf {\bibinfo {volume}
  {334}},\ \bibinfo {pages} {1398--1401} (\bibinfo {year} {2011})}\BibitemShut
  {NoStop}%
\bibitem [{\citenamefont {Abbey-Lee}, \citenamefont {Mathot},\ and\
  \citenamefont {Dingemanse}(2016)}]{ref12}%
  \BibitemOpen
  \bibfield  {author} {\bibinfo {author} {\bibfnamefont {R.~N.}\ \bibnamefont
  {Abbey-Lee}}, \bibinfo {author} {\bibfnamefont {K.~J.}\ \bibnamefont
  {Mathot}}, \ and\ \bibinfo {author} {\bibfnamefont {N.~J.}\ \bibnamefont
  {Dingemanse}},\ }\bibfield  {title} {\enquote {\bibinfo {title} {Behavioral
  and morphological responses to perceived predation risk: a field experiment
  in passerines},}\ }\href@noop {} {\bibfield  {journal} {\bibinfo  {journal}
  {Behavioral Ecology}\ }\textbf {\bibinfo {volume} {27}},\ \bibinfo {pages}
  {857--864} (\bibinfo {year} {2016})}\BibitemShut {NoStop}%
\bibitem [{\citenamefont {Panday}\ \emph {et~al.}(2019)\citenamefont {Panday},
  \citenamefont {Pal}, \citenamefont {Samanta},\ and\ \citenamefont
  {Chattopadhyay}}]{ref10}%
  \BibitemOpen
  \bibfield  {author} {\bibinfo {author} {\bibfnamefont {P.}~\bibnamefont
  {Panday}}, \bibinfo {author} {\bibfnamefont {N.}~\bibnamefont {Pal}},
  \bibinfo {author} {\bibfnamefont {S.}~\bibnamefont {Samanta}}, \ and\
  \bibinfo {author} {\bibfnamefont {J.}~\bibnamefont {Chattopadhyay}},\
  }\bibfield  {title} {\enquote {\bibinfo {title} {A three species food chain
  model with fear induced trophic cascade},}\ }\href@noop {} {\bibfield
  {journal} {\bibinfo  {journal} {International Journal of Applied and
  Computational Mathematics}\ }\textbf {\bibinfo {volume} {5}},\ \bibinfo
  {pages} {1--26} (\bibinfo {year} {2019})}\BibitemShut {NoStop}%
\bibitem [{\citenamefont {Mandal}\ \emph {et~al.}(2020)\citenamefont {Mandal},
  \citenamefont {Jana}, \citenamefont {Nandi},\ and\ \citenamefont
  {Kar}}]{ref36}%
  \BibitemOpen
  \bibfield  {author} {\bibinfo {author} {\bibfnamefont {M.}~\bibnamefont
  {Mandal}}, \bibinfo {author} {\bibfnamefont {S.}~\bibnamefont {Jana}},
  \bibinfo {author} {\bibfnamefont {S.~K.}\ \bibnamefont {Nandi}}, \ and\
  \bibinfo {author} {\bibfnamefont {T.~K.}\ \bibnamefont {Kar}},\ }\bibfield
  {title} {\enquote {\bibinfo {title} {Modelling and control of a
  fractional-order epidemic model with fear effect},}\ }\href@noop {}
  {\bibfield  {journal} {\bibinfo  {journal} {Energy, ecology \& environment}\
  ,\ \bibinfo {pages} {1--12}} (\bibinfo {year} {2020})}\BibitemShut {NoStop}%
\bibitem [{\citenamefont {Sasmal}\ and\ \citenamefont {Takeuchi}(2020)}]{ref4}%
  \BibitemOpen
  \bibfield  {author} {\bibinfo {author} {\bibfnamefont {S.~K.}\ \bibnamefont
  {Sasmal}}\ and\ \bibinfo {author} {\bibfnamefont {Y.}~\bibnamefont
  {Takeuchi}},\ }\bibfield  {title} {\enquote {\bibinfo {title} {Dynamics of a
  predator-prey system with fear and group defense},}\ }\href@noop {}
  {\bibfield  {journal} {\bibinfo  {journal} {Journal of Mathematical Analysis
  and Applications}\ }\textbf {\bibinfo {volume} {481}},\ \bibinfo {pages}
  {123471} (\bibinfo {year} {2020})}\BibitemShut {NoStop}%
\bibitem [{\citenamefont {Chakraborty}, \citenamefont {Baek},\ and\
  \citenamefont {Bairagi}(2021)}]{ref5}%
  \BibitemOpen
  \bibfield  {author} {\bibinfo {author} {\bibfnamefont {B.}~\bibnamefont
  {Chakraborty}}, \bibinfo {author} {\bibfnamefont {H.}~\bibnamefont {Baek}}, \
  and\ \bibinfo {author} {\bibfnamefont {N.}~\bibnamefont {Bairagi}},\
  }\bibfield  {title} {\enquote {\bibinfo {title} {Diffusion-induced regular
  and chaotic patterns in a ratio-dependent predator--prey model with fear
  factor and prey refuge},}\ }\href@noop {} {\bibfield  {journal} {\bibinfo
  {journal} {Chaos: An Interdisciplinary Journal of Nonlinear Science}\
  }\textbf {\bibinfo {volume} {31}},\ \bibinfo {pages} {033128} (\bibinfo
  {year} {2021})}\BibitemShut {NoStop}%
\bibitem [{\citenamefont {Wang}\ \emph {et~al.}(2019)\citenamefont {Wang},
  \citenamefont {Cai}, \citenamefont {Fu},\ and\ \citenamefont {Wang}}]{ref6}%
  \BibitemOpen
  \bibfield  {author} {\bibinfo {author} {\bibfnamefont {J.}~\bibnamefont
  {Wang}}, \bibinfo {author} {\bibfnamefont {Y.}~\bibnamefont {Cai}}, \bibinfo
  {author} {\bibfnamefont {S.}~\bibnamefont {Fu}}, \ and\ \bibinfo {author}
  {\bibfnamefont {W.}~\bibnamefont {Wang}},\ }\bibfield  {title} {\enquote
  {\bibinfo {title} {The effect of the fear factor on the dynamics of a
  predator-prey model incorporating the prey refuge},}\ }\href@noop {}
  {\bibfield  {journal} {\bibinfo  {journal} {Chaos: An Interdisciplinary
  Journal of Nonlinear Science}\ }\textbf {\bibinfo {volume} {29}},\ \bibinfo
  {pages} {083109} (\bibinfo {year} {2019})}\BibitemShut {NoStop}%
\bibitem [{\citenamefont {Wang}\ \emph {et~al.}(2020)\citenamefont {Wang},
  \citenamefont {Tan}, \citenamefont {Cai},\ and\ \citenamefont {Wang}}]{ref7}%
  \BibitemOpen
  \bibfield  {author} {\bibinfo {author} {\bibfnamefont {X.}~\bibnamefont
  {Wang}}, \bibinfo {author} {\bibfnamefont {Y.}~\bibnamefont {Tan}}, \bibinfo
  {author} {\bibfnamefont {Y.}~\bibnamefont {Cai}}, \ and\ \bibinfo {author}
  {\bibfnamefont {W.}~\bibnamefont {Wang}},\ }\bibfield  {title} {\enquote
  {\bibinfo {title} {Impact of the fear effect on the stability and bifurcation
  of a leslie--gower predator--prey model},}\ }\href@noop {} {\bibfield
  {journal} {\bibinfo  {journal} {International Journal of Bifurcation and
  Chaos}\ }\textbf {\bibinfo {volume} {30}},\ \bibinfo {pages} {2050210}
  (\bibinfo {year} {2020})}\BibitemShut {NoStop}%
\bibitem [{\citenamefont {Francesca~Carfora}\ and\ \citenamefont
  {Torcicollo}(2020)}]{ref26}%
  \BibitemOpen
  \bibfield  {author} {\bibinfo {author} {\bibfnamefont {M.}~\bibnamefont
  {Francesca~Carfora}}\ and\ \bibinfo {author} {\bibfnamefont {I.}~\bibnamefont
  {Torcicollo}},\ }\bibfield  {title} {\enquote {\bibinfo {title}
  {Cross-diffusion-driven instability in a predator-prey system with fear and
  group defense},}\ }\href@noop {} {\bibfield  {journal} {\bibinfo  {journal}
  {Mathematics}\ }\textbf {\bibinfo {volume} {8}},\ \bibinfo {pages} {1244}
  (\bibinfo {year} {2020})}\BibitemShut {NoStop}%
\bibitem [{\citenamefont {Panday}\ \emph {et~al.}(2018)\citenamefont {Panday},
  \citenamefont {Pal}, \citenamefont {Samanta},\ and\ \citenamefont
  {Chattopadhyay}}]{ref28}%
  \BibitemOpen
  \bibfield  {author} {\bibinfo {author} {\bibfnamefont {P.}~\bibnamefont
  {Panday}}, \bibinfo {author} {\bibfnamefont {N.}~\bibnamefont {Pal}},
  \bibinfo {author} {\bibfnamefont {S.}~\bibnamefont {Samanta}}, \ and\
  \bibinfo {author} {\bibfnamefont {J.}~\bibnamefont {Chattopadhyay}},\
  }\bibfield  {title} {\enquote {\bibinfo {title} {Stability and bifurcation
  analysis of a three-species food chain model with fear},}\ }\href@noop {}
  {\bibfield  {journal} {\bibinfo  {journal} {International Journal of
  Bifurcation and Chaos}\ }\textbf {\bibinfo {volume} {28}},\ \bibinfo {pages}
  {1850009} (\bibinfo {year} {2018})}\BibitemShut {NoStop}%
\bibitem [{\citenamefont {Cong}, \citenamefont {Fan},\ and\ \citenamefont
  {Zou}(2021)}]{ref29}%
  \BibitemOpen
  \bibfield  {author} {\bibinfo {author} {\bibfnamefont {P.}~\bibnamefont
  {Cong}}, \bibinfo {author} {\bibfnamefont {M.}~\bibnamefont {Fan}}, \ and\
  \bibinfo {author} {\bibfnamefont {X.}~\bibnamefont {Zou}},\ }\bibfield
  {title} {\enquote {\bibinfo {title} {Dynamics of a three-species food chain
  model with fear effect},}\ }\href@noop {} {\bibfield  {journal} {\bibinfo
  {journal} {Communications in Nonlinear Science and Numerical Simulation}\
  }\textbf {\bibinfo {volume} {99}},\ \bibinfo {pages} {105809} (\bibinfo
  {year} {2021})}\BibitemShut {NoStop}%
\bibitem [{\citenamefont {Qiao}\ \emph {et~al.}(2019)\citenamefont {Qiao},
  \citenamefont {Cai}, \citenamefont {Fu},\ and\ \citenamefont {Wang}}]{ref30}%
  \BibitemOpen
  \bibfield  {author} {\bibinfo {author} {\bibfnamefont {T.}~\bibnamefont
  {Qiao}}, \bibinfo {author} {\bibfnamefont {Y.}~\bibnamefont {Cai}}, \bibinfo
  {author} {\bibfnamefont {S.}~\bibnamefont {Fu}}, \ and\ \bibinfo {author}
  {\bibfnamefont {W.}~\bibnamefont {Wang}},\ }\bibfield  {title} {\enquote
  {\bibinfo {title} {Stability and hopf bifurcation in a predator-prey model
  with the cost of anti-predator behaviors},}\ }\href@noop {} {\bibfield
  {journal} {\bibinfo  {journal} {International Journal of Bifurcation and
  Chaos}\ }\textbf {\bibinfo {volume} {29}} (\bibinfo {year}
  {2019})}\BibitemShut {NoStop}%
\bibitem [{\citenamefont {Sasmal}(2018)}]{ref31}%
  \BibitemOpen
  \bibfield  {author} {\bibinfo {author} {\bibfnamefont {S.~K.}\ \bibnamefont
  {Sasmal}},\ }\bibfield  {title} {\enquote {\bibinfo {title} {Population
  dynamics with multiple allee effects induced by fear factors - a mathematical
  study on prey-predator interactions},}\ }\href@noop {} {\bibfield  {journal}
  {\bibinfo  {journal} {Applied Mathematical Modelling}\ }\textbf {\bibinfo
  {volume} {64}},\ \bibinfo {pages} {1--14} (\bibinfo {year}
  {2018})}\BibitemShut {NoStop}%
\bibitem [{\citenamefont {Wang}\ and\ \citenamefont {Zou}(2017)}]{ref32}%
  \BibitemOpen
  \bibfield  {author} {\bibinfo {author} {\bibfnamefont {X.}~\bibnamefont
  {Wang}}\ and\ \bibinfo {author} {\bibfnamefont {X.}~\bibnamefont {Zou}},\
  }\bibfield  {title} {\enquote {\bibinfo {title} {Modeling the fear effect in
  predator-prey interactions with adaptive avoidance of predators},}\
  }\href@noop {} {\bibfield  {journal} {\bibinfo  {journal} {Bulletin of
  Mathematical Biology}\ }\textbf {\bibinfo {volume} {79}},\ \bibinfo {pages}
  {1325--1359} (\bibinfo {year} {2017})}\BibitemShut {NoStop}%
\bibitem [{\citenamefont {Zhang}\ \emph
  {et~al.}(2019{\natexlab{b}})\citenamefont {Zhang}, \citenamefont {Cai},
  \citenamefont {Fu},\ and\ \citenamefont {Wang}}]{ref33}%
  \BibitemOpen
  \bibfield  {author} {\bibinfo {author} {\bibfnamefont {H.}~\bibnamefont
  {Zhang}}, \bibinfo {author} {\bibfnamefont {Y.}~\bibnamefont {Cai}}, \bibinfo
  {author} {\bibfnamefont {S.}~\bibnamefont {Fu}}, \ and\ \bibinfo {author}
  {\bibfnamefont {W.}~\bibnamefont {Wang}},\ }\bibfield  {title} {\enquote
  {\bibinfo {title} {Impact of the fear effect in a prey-predator model
  incorporating a prey refuge},}\ }\href@noop {} {\bibfield  {journal}
  {\bibinfo  {journal} {Applied Mathematics and Computation}\ }\textbf
  {\bibinfo {volume} {356}},\ \bibinfo {pages} {328--337} (\bibinfo {year}
  {2019}{\natexlab{b}})}\BibitemShut {NoStop}%
\bibitem [{\citenamefont {Han}, \citenamefont {Guin},\ and\ \citenamefont
  {Dai}(2020)}]{ref35}%
  \BibitemOpen
  \bibfield  {author} {\bibinfo {author} {\bibfnamefont {R.}~\bibnamefont
  {Han}}, \bibinfo {author} {\bibfnamefont {L.~N.}\ \bibnamefont {Guin}}, \
  and\ \bibinfo {author} {\bibfnamefont {B.}~\bibnamefont {Dai}},\ }\bibfield
  {title} {\enquote {\bibinfo {title} {Cross-diffusion-driven pattern formation
  and selection in a modified leslie-gower predator-prey model with fear
  effect},}\ }\href@noop {} {\bibfield  {journal} {\bibinfo  {journal} {Journal
  of Biological Systems}\ }\textbf {\bibinfo {volume} {28}},\ \bibinfo {pages}
  {27--64} (\bibinfo {year} {2020})}\BibitemShut {NoStop}%
\bibitem [{\citenamefont {Ye}\ and\ \citenamefont {Zhao}(2021)}]{ref37}%
  \BibitemOpen
  \bibfield  {author} {\bibinfo {author} {\bibfnamefont {Y.}~\bibnamefont
  {Ye}}\ and\ \bibinfo {author} {\bibfnamefont {Y.}~\bibnamefont {Zhao}},\
  }\bibfield  {title} {\enquote {\bibinfo {title} {Bifurcation analysis of a
  delay-induced predator-prey model with allee effect and prey group
  defense},}\ }\href@noop {} {\bibfield  {journal} {\bibinfo  {journal}
  {International Journal of Bifurcation and Chaos}\ }\textbf {\bibinfo {volume}
  {31}},\ \bibinfo {pages} {2150158} (\bibinfo {year} {2021})}\BibitemShut
  {NoStop}%
\bibitem [{\citenamefont {Zhang}\ \emph {et~al.}(2014)\citenamefont {Zhang},
  \citenamefont {Xing}, \citenamefont {Zang},\ and\ \citenamefont
  {Han}}]{ref19}%
  \BibitemOpen
  \bibfield  {author} {\bibinfo {author} {\bibfnamefont {T.}~\bibnamefont
  {Zhang}}, \bibinfo {author} {\bibfnamefont {Y.}~\bibnamefont {Xing}},
  \bibinfo {author} {\bibfnamefont {H.}~\bibnamefont {Zang}}, \ and\ \bibinfo
  {author} {\bibfnamefont {M.}~\bibnamefont {Han}},\ }\bibfield  {title}
  {\enquote {\bibinfo {title} {Spatio-temporal dynamics of a reaction-diffusion
  system for a predator--prey model with hyperbolic mortality},}\ }\href@noop
  {} {\bibfield  {journal} {\bibinfo  {journal} {Nonlinear Dynamics}\ }\textbf
  {\bibinfo {volume} {78}},\ \bibinfo {pages} {265--277} (\bibinfo {year}
  {2014})}\BibitemShut {NoStop}%
\bibitem [{\citenamefont {Jana}, \citenamefont {Batabyal},\ and\ \citenamefont
  {Lakshmanan}(2020)}]{ref25}%
  \BibitemOpen
  \bibfield  {author} {\bibinfo {author} {\bibfnamefont {D.}~\bibnamefont
  {Jana}}, \bibinfo {author} {\bibfnamefont {S.}~\bibnamefont {Batabyal}}, \
  and\ \bibinfo {author} {\bibfnamefont {M.}~\bibnamefont {Lakshmanan}},\
  }\bibfield  {title} {\enquote {\bibinfo {title} {Self-diffusion-driven
  pattern formation in prey-predator system with complex habitat under fear
  effect},}\ }\href@noop {} {\bibfield  {journal} {\bibinfo  {journal} {The
  European Physical Journal Plus}\ }\textbf {\bibinfo {volume} {135}} (\bibinfo
  {year} {2020})}\BibitemShut {NoStop}%
\bibitem [{\citenamefont {Capone}\ \emph {et~al.}(2019)\citenamefont {Capone},
  \citenamefont {Carfora}, \citenamefont {De~Luca},\ and\ \citenamefont
  {Torcicollo}}]{ref27}%
  \BibitemOpen
  \bibfield  {author} {\bibinfo {author} {\bibfnamefont {F.}~\bibnamefont
  {Capone}}, \bibinfo {author} {\bibfnamefont {M.~F.}\ \bibnamefont {Carfora}},
  \bibinfo {author} {\bibfnamefont {R.}~\bibnamefont {De~Luca}}, \ and\
  \bibinfo {author} {\bibfnamefont {I.}~\bibnamefont {Torcicollo}},\ }\bibfield
   {title} {\enquote {\bibinfo {title} {Turing patterns in a reaction-diffusion
  system modeling hunting cooperation},}\ }\href@noop {} {\bibfield  {journal}
  {\bibinfo  {journal} {Mathematics and Computers in Simulation}\ }\textbf
  {\bibinfo {volume} {165}},\ \bibinfo {pages} {172--180} (\bibinfo {year}
  {2019})}\BibitemShut {NoStop}%
\bibitem [{\citenamefont {Ipsen}, \citenamefont {Hynne},\ and\ \citenamefont
  {S{\o}rensen}(2000)}]{ref23}%
  \BibitemOpen
  \bibfield  {author} {\bibinfo {author} {\bibfnamefont {M.}~\bibnamefont
  {Ipsen}}, \bibinfo {author} {\bibfnamefont {F.}~\bibnamefont {Hynne}}, \ and\
  \bibinfo {author} {\bibfnamefont {P.}~\bibnamefont {S{\o}rensen}},\
  }\bibfield  {title} {\enquote {\bibinfo {title} {Amplitude equations for
  reaction--diffusion systems with a hopf bifurcation and slow real modes},}\
  }\href@noop {} {\bibfield  {journal} {\bibinfo  {journal} {Physica D:
  Nonlinear Phenomena}\ }\textbf {\bibinfo {volume} {136}},\ \bibinfo {pages}
  {66--92} (\bibinfo {year} {2000})}\BibitemShut {NoStop}%
\bibitem [{\citenamefont {Yuan}, \citenamefont {Xu},\ and\ \citenamefont
  {Zhang}(2013)}]{ref20}%
  \BibitemOpen
  \bibfield  {author} {\bibinfo {author} {\bibfnamefont {S.}~\bibnamefont
  {Yuan}}, \bibinfo {author} {\bibfnamefont {C.}~\bibnamefont {Xu}}, \ and\
  \bibinfo {author} {\bibfnamefont {T.}~\bibnamefont {Zhang}},\ }\bibfield
  {title} {\enquote {\bibinfo {title} {Spatial dynamics in a predator-prey
  model with herd behavior},}\ }\href@noop {} {\bibfield  {journal} {\bibinfo
  {journal} {Chaos: An Interdisciplinary Journal of Nonlinear Science}\
  }\textbf {\bibinfo {volume} {23}},\ \bibinfo {pages} {033102} (\bibinfo
  {year} {2013})}\BibitemShut {NoStop}%
\bibitem [{\citenamefont {Wei-Ming}\ \emph {et~al.}(2011)\citenamefont
  {Wei-Ming}, \citenamefont {Wen-Juan}, \citenamefont {Ye-Zhi},\ and\
  \citenamefont {Yong-Ji}}]{ref24}%
  \BibitemOpen
  \bibfield  {author} {\bibinfo {author} {\bibfnamefont {W.}~\bibnamefont
  {Wei-Ming}}, \bibinfo {author} {\bibfnamefont {W.}~\bibnamefont {Wen-Juan}},
  \bibinfo {author} {\bibfnamefont {L.}~\bibnamefont {Ye-Zhi}}, \ and\ \bibinfo
  {author} {\bibfnamefont {T.}~\bibnamefont {Yong-Ji}},\ }\bibfield  {title}
  {\enquote {\bibinfo {title} {Pattern selection in a predation model with self
  and cross diffusion},}\ }\href@noop {} {\bibfield  {journal} {\bibinfo
  {journal} {Chinese Physics B}\ }\textbf {\bibinfo {volume} {20}},\ \bibinfo
  {pages} {034702} (\bibinfo {year} {2011})}\BibitemShut {NoStop}%
\bibitem [{\citenamefont {Ouyang}(2010)}]{ref21}%
  \BibitemOpen
  \bibfield  {author} {\bibinfo {author} {\bibfnamefont {Q.}~\bibnamefont
  {Ouyang}},\ }\href@noop {} {\enquote {\bibinfo {title} {Nonlinear science and
  the pattern dynamics introduction},}\ } (\bibinfo {year} {2010})\BibitemShut
  {NoStop}%
\bibitem [{\citenamefont {Liu}\ \emph {et~al.}(2019)\citenamefont {Liu},
  \citenamefont {Ye}, \citenamefont {Wei}, \citenamefont {Ma}, \citenamefont
  {Ma},\ and\ \citenamefont {Zhang}}]{ref22}%
  \BibitemOpen
  \bibfield  {author} {\bibinfo {author} {\bibfnamefont {H.}~\bibnamefont
  {Liu}}, \bibinfo {author} {\bibfnamefont {Y.}~\bibnamefont {Ye}}, \bibinfo
  {author} {\bibfnamefont {Y.}~\bibnamefont {Wei}}, \bibinfo {author}
  {\bibfnamefont {W.}~\bibnamefont {Ma}}, \bibinfo {author} {\bibfnamefont
  {M.}~\bibnamefont {Ma}}, \ and\ \bibinfo {author} {\bibfnamefont
  {K.}~\bibnamefont {Zhang}},\ }\bibfield  {title} {\enquote {\bibinfo {title}
  {Pattern formation in a reaction-diffusion predator-prey model with weak
  allee effect and delay},}\ }\href@noop {} {\bibfield  {journal} {\bibinfo
  {journal} {Complexity}\ }\textbf {\bibinfo {volume} {2019}} (\bibinfo {year}
  {2019})}\BibitemShut {NoStop}%
\bibitem [{\citenamefont {Garvie}(2007)}]{ref34}%
  \BibitemOpen
  \bibfield  {author} {\bibinfo {author} {\bibfnamefont {M.~R.}\ \bibnamefont
  {Garvie}},\ }\bibfield  {title} {\enquote {\bibinfo {title}
  {Finite-difference schemes for reaction-diffusion equations modeling
  predator-prey interactions in matlab},}\ }\href@noop {} {\bibfield  {journal}
  {\bibinfo  {journal} {Bulletin of Mathematical Biology}\ }\textbf {\bibinfo
  {volume} {69}},\ \bibinfo {pages} {931--956} (\bibinfo {year}
  {2007})}\BibitemShut {NoStop}%
\end{thebibliography}%
\end{document}